%% file: GFCutoff.tex
\documentclass{article}

\makeatletter
\usepackage[style=authoryear,natbib=true,sorting=nyt,backend=bibtex8]{biblatex}
\usepackage{rotating} 

\usepackage{indentfirst}

\usepackage{authblk}
\usepackage{url,amsmath,amsfonts,amssymb}
\usepackage{caption}
\usepackage{subcaption}
\usepackage{algorithm}
\usepackage{algorithmic}
\usepackage{graphicx}
\usepackage{listings}
\usepackage[T1]{fontenc}
\usepackage{enumerate}   
\usepackage{color}
\usepackage{relsize}
\usepackage{multirow}
\usepackage[utf8]{inputenc}
\usepackage{booktabs}       
\usepackage{nicefrac}       
\usepackage{microtype}      
\usepackage{appendix}
\usepackage[dvipsnames]{xcolor}
\usepackage{pifont,xr}

\usepackage{geometry}
\numberwithin{equation}{section}

\DeclareMathOperator*{\calO}{\mathcal O}

\newtheorem{assumption}{Assumption}

\newtheorem{example}{Example}
\newtheorem{theorem}{Theorem}
\newtheorem{lemma}[theorem]{Lemma}
\newtheorem{corollary}[theorem]{Corollary}
\newtheorem{proposition}[theorem]{Proposition}

\newcommand{\GFCModel}{BMU}
\newcommand{\GFCMethod}{Ucut}

\def \endprf{\hfill {\vrule height6pt width6pt depth0pt}\medskip}
\newenvironment{proof}{\noindent {\textbf{Proof}} }{\endprf\par}

\newcommand{\Rbb}{{\mathbb R}}

\newcommand{\Ebb}{{\mathbb E}}
\newcommand{\Ibb}{\mathbb{I}}
\newcommand{\Pbb}{\mathbb{P}}

\newcommand{\Ocal}{\mathcal{O}}
\newcommand{\Ecal}{\mathcal{E}}
\newcommand{\ftil}{\tilde{f}}

\newcommand{\Var}{\textit{var}}

\newlength{\toppush}
\bibliography{ref}
\begin{document}
\title{Binomial Mixture Model With U-shape Constraint}

\author[1]{Yuting Ye \thanks{yeyt@berkeley.edu}}
\author[2]{Peter J. Bickel \thanks{bickelp@berkeley.edu}}
\affil[1]{Biostatistics Division, UC Berkeley}
\affil[2]{Department of Statistics, UC Berkeley}

\maketitle

\begin{abstract}
  In this article, we study the binomial mixture model under the regime that the binomial size $m$ can be relatively large compared to the sample size $n$. This project is motivated by the GeneFishing method \citep{liu2019genefishing}, whose output is a combination of the parameter of interest and the subsampling noise. To tackle the noise in the output, we utilize the observation that the density of the output has a U shape and model the output with the binomial mixture model under a U shape constraint. We first analyze the estimation of the underlying distribution $F$ in the binomial mixture model under various conditions for $F$. Equipped with these theoretical understandings, we propose a simple method Ucut to identify the cutoffs of the U shape and recover the underlying distribution based on the Grenander estimator \citep{grenander1956theory}. It has been shown that when $m = \Omega(n^{\frac{2}{3}})$, the identified cutoffs converge at the rate $\calO(n^{-\frac{1}{3}})$. The $L_1$ distance between the recovered distribution and the true one decreases at the same rate. To demonstrate the performance, we apply our method to varieties of simulation studies, a GTEX dataset used in \citep{liu2019genefishing} and a single cell dataset from Tabula Muris.
\end{abstract}

\section{Introduction}\label{sec:bmu_intro}
The binomial mixture model has received much attention since the late 1960s. In the field of performance evaluation, \citet{lord1969estimating,lord1975empirical, sivaganesan1993robust} utilized this model to address the problem of psychological testing. \citet{thomas1989binomial} used a two-component binomial mixture distribution to model the individual differences in children's performance on a classification task. \citet{grilli2015binomial} employed a binomial finite mixture to model the number of credits gained by freshmen during the first year at the School of Economics of the University of Florence. In addition, the binomial mixture model is commonly applied to analyze population survey data in ecology. \citet{royle2004n,kery2008estimating, o2015partitioning,wu2015bayesian} estimate absolute abundance while accounting for imperfect detection using binomial detection models. The binomial mixture model was also used to estimate bird and bat fatalities at wind power facilities \citep{McDonald2020.01.21.914754}.

Formally, we say $X$ is a random variable which has a binomial mixture model if
\begin{equation}\label{eq:general_binomial_mixture}
X \sim \int \text{Binomial}(m, s) dQ(m, s)
\end{equation}
where $Q(\cdot, \cdot)$ is a bivariate measure of the binomial size $m$ and the success probability $s$ on $\mathbb{N} \times [0, 1]$. In the field of population survey in ecology, $m$ can be modeled as Poisson or negative binomial random variable while $s$ can be modeled as a beta random variable, linked to a linear combination of additional covariates by a logistic function \citep{royle2004n,kery2008estimating, o2015partitioning,wu2015bayesian, McDonald2020.01.21.914754}. Such models are always identifiable thanks to the parametric assumptions. In the field of performance evaluation, $m$ is always known and usually small (restricted by the intrinsic nature of the problem), thus \eqref{eq:general_binomial_mixture} is reduced to
\begin{equation}\label{eq:binomial_mixture}
  \left \{
  \begin{array}{l}
    s \sim F,\\
    X|s \sim \text{Binomial}(m, s),\\
  \end{array}
  \right .
\end{equation}
where $F$ is a probability distribution on $[0, 1]$. For instance, $m$ refers to the number of questions of a psychological test in \citet{lord1969estimating,lord1975empirical, sivaganesan1993robust}. When the univariate probability distribution $F$ corresponds to a finite point mass function (pmf), \eqref{eq:binomial_mixture} is a binomial finite mixture model as in \citet{thomas1989binomial, grilli2015binomial}; when $F$ corresponds to a density (either parametric or nonparametric), \eqref{eq:binomial_mixture} is a hierarchical binomial mixture model as in \citet{lord1969estimating,lord1975empirical, sivaganesan1993robust}. Such models suffer from an unidentifiability issue --- only the first $m$ moments of $F$ can be identified even with unlimited sample size \citep{teicher1963identifiability}; more identifiability details are discussed in Appendix \ref{appendix:bmu_identifiability}.  

In this study, we are interested in a regime unlike that for performance evaluation where $m$ is known and small or that for population survey in ecology where $m$ is unknown. We study \eqref{eq:binomial_mixture} with a known $m$ that can be as large as the magnitude of $n$, which has not been investigated before in the literature. This study is motivated by the GeneFishing method \citep{liu2019genefishing}, which is proposed to take care of extreme inhomogeneity and low signal-to-noise ratio in high-dimensional data. This method employs subsampling and aggregation techniques to evaluate the functional relevance of a candidate gene to a set of bait genes (the bait genes are pre-selected by experts concerning certain biological processes). Briefly speaking, it clusters a subsample of candidate genes and the bait genes repetitively. One candidate gene is thought of as functioning similarly to the bait genes if the former and the latter are always grouped during the clustering procedure. More details about GeneFishing can be found in Appendix \ref{appendix:bmu_gene_fishing}.

Mathematically, suppose there are $n$ objects (e.g., genes) and $m$ binomial trials (e.g., the number of clustering times which can be manually tuned). For each object, we assume it has a success probability $s_i$ (e.g., a rate reflecting how relevant one candidate gene is to the baits genes) that is i.i.d. sampled from an underlying distribution $F$. Then, an observation $X_i$ (e.g., the times one candidate gene and the bait genes are grouped together) or $\hat{s}_i := X_i/m$ (denoted as capture frequency ratio in GeneFishing, or CFR) is independently generated from the binomial mixture model \eqref{eq:binomial_mixture}. A discussion on the independence assumption is put in Appendix \ref{appendix:bmu_gene_fishing}.

Our goal is to estimate and infer the underlying distribution $F$. When $m$ is sufficiently large, the binomial randomness plays little role, and there is no difference between $\hat{s}_i$ and $s_i$ for estimation or inference purposes. In other words, the permission of a large $m$ alleviates the identifiability issue of the binomial mixture model. Thereby, a question naturally arises as follows. 
\begin{itemize}
  \item [\textbf{Q1}] For the binomial mixture model \eqref{eq:binomial_mixture}, what is the minimal binomial size $m$ so that any distribution estimators can work as if there is no binomial randomness?
\end{itemize}
We can reformulate Q1 in a more explicit way:
\begin{enumerate}
\item [\textbf{~~Q1'}] Considering that the binomial mixture model \eqref{eq:binomial_mixture} is trivially identifiable when $m = \infty$, is there a minimal $m$ such that we can have an ``acceptable'' estimator of $F$ under various conditions:
  \begin{itemize}
  \item General $F$ without additional conditions.
  \item $F$ with a density.
  \item $F$ with a continuously differentiable density.
  \item $F$ with a monotone density.
  \end{itemize}
\end{enumerate}

If the above question can be answered satisfactorily, we can cast our attention back to the GeneFishing example. In \citet{liu2019genefishing}, there remains an open question on how large $\hat{s}_i$ should be so that the associated gene is labeled as a ``discovery''. We have an important observation that can be leveraged towards this end: the histograms of $\hat{s}_i$ appear in a U shape; see Figure \ref{fig:hist_CFR}. Hence, the second question we want to answer is 
\begin{itemize}
\item [\textbf{Q2}] Suppose the underlying distribution $F$ has a U shape, how to make decisions based on the data generated from the binomial mixture model?
\end{itemize}

\begin{figure}[ht]
  \centering
  \begin{minipage}{0.24\textwidth}
      \centering
      \includegraphics[width = \textwidth]{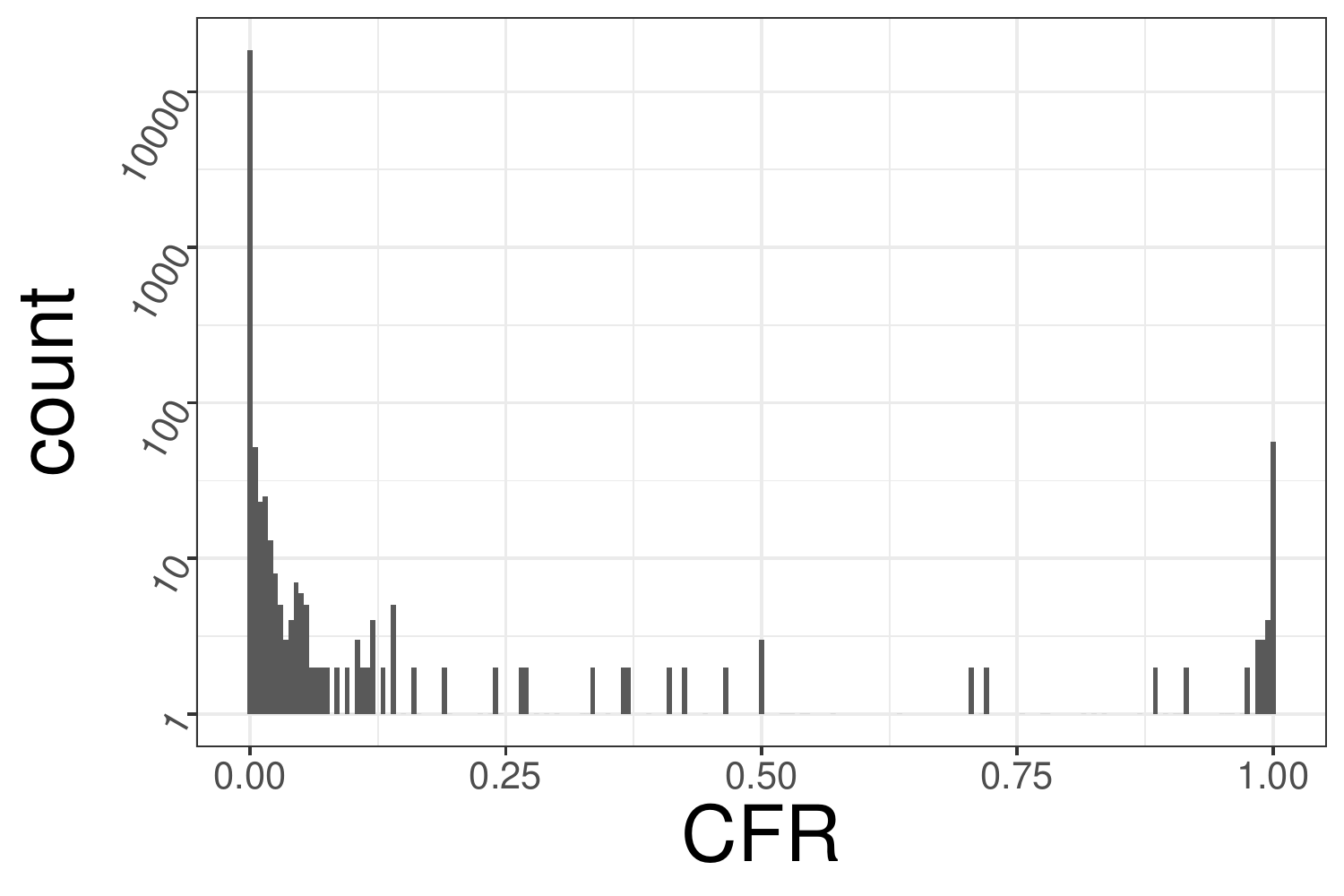}
      \subcaption{Liver}
  \end{minipage}
  \begin{minipage}{0.24\textwidth}
      \centering
      \includegraphics[width = \textwidth]{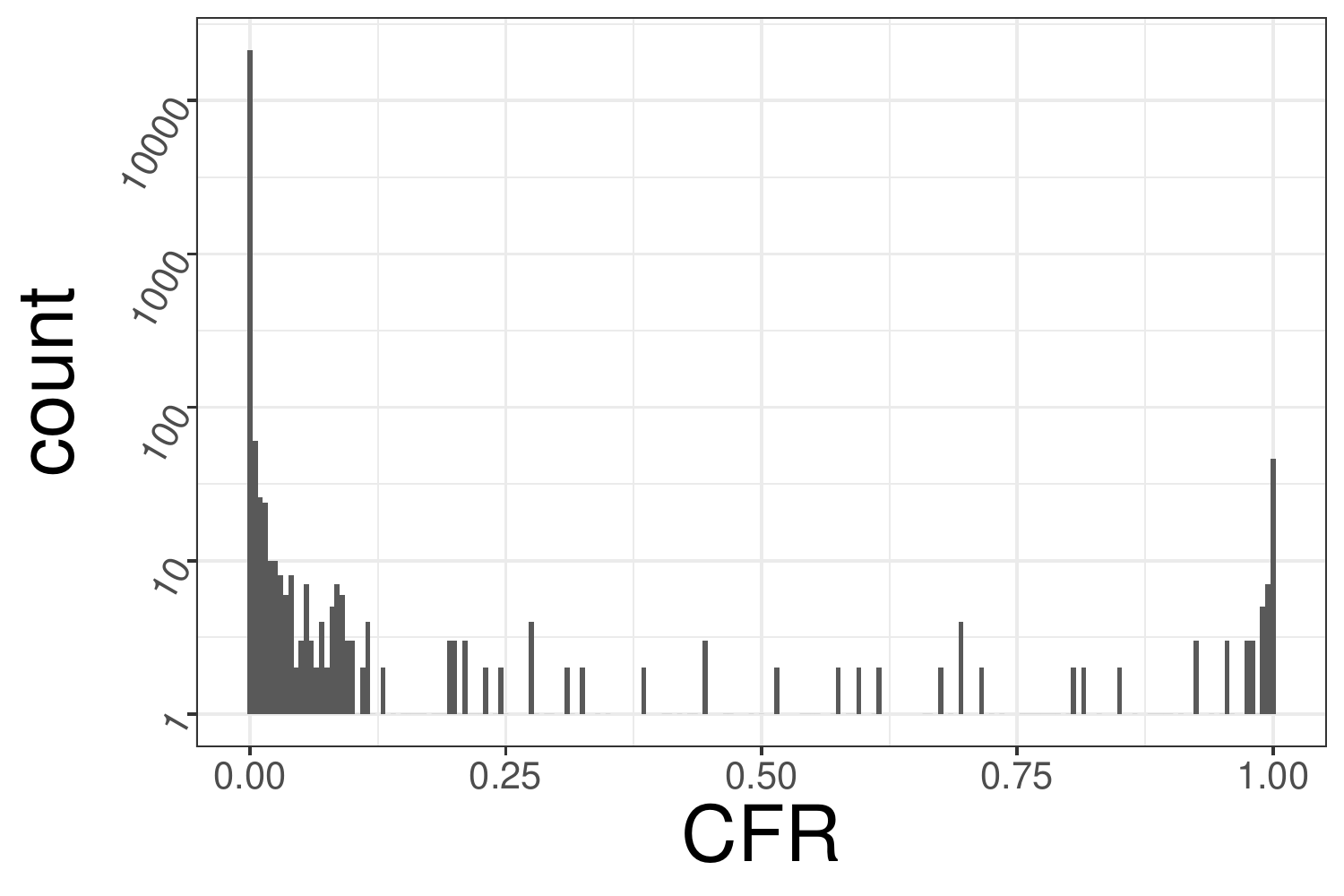}
      \subcaption{Colon Transverse}
  \end{minipage}
  \begin{minipage}{0.24\textwidth}
      \centering
      \includegraphics[width = \textwidth]{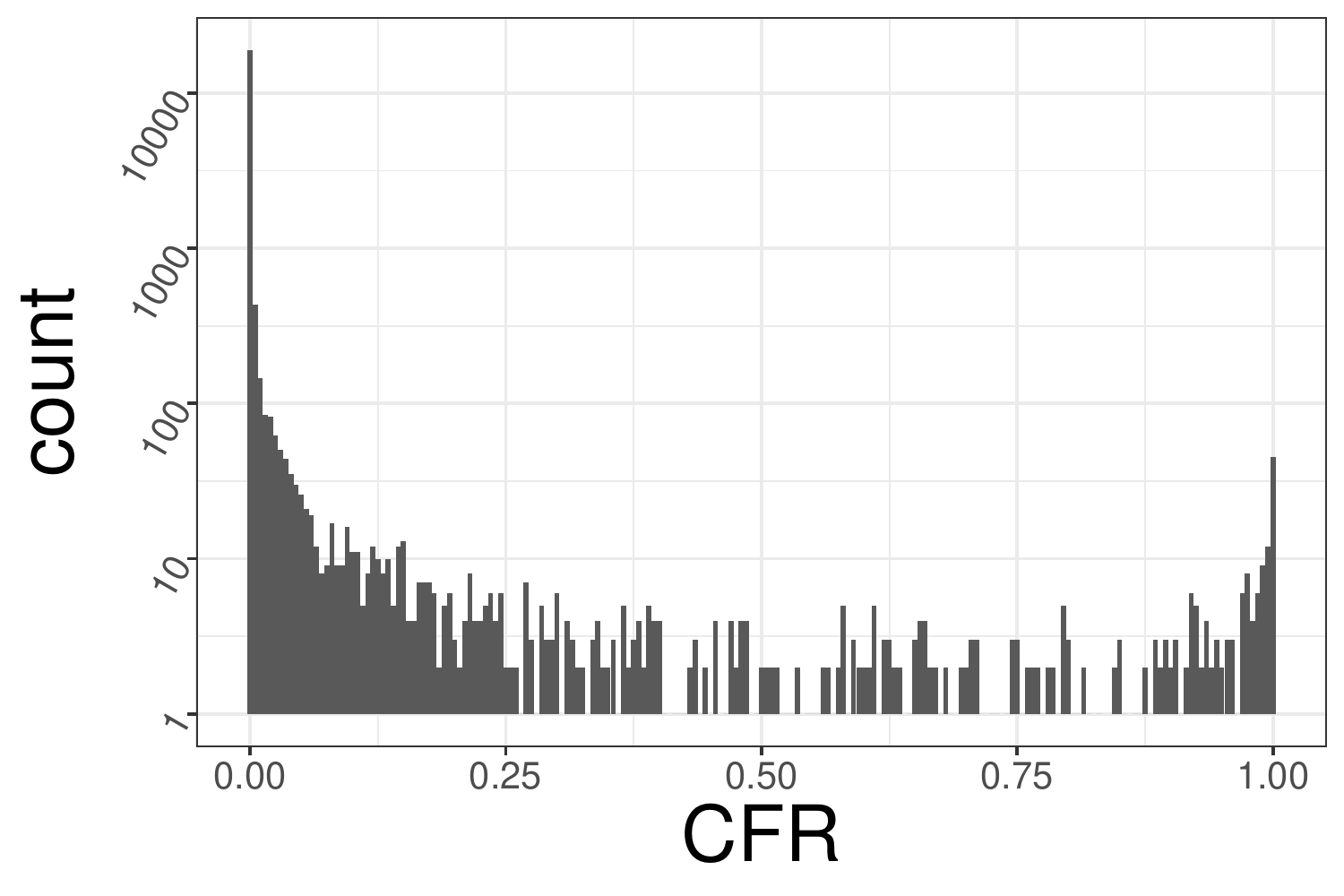}
      \subcaption{Artery Coronary}
  \end{minipage}
  \begin{minipage}{0.24\textwidth}
      \centering
      \includegraphics[width = \textwidth]{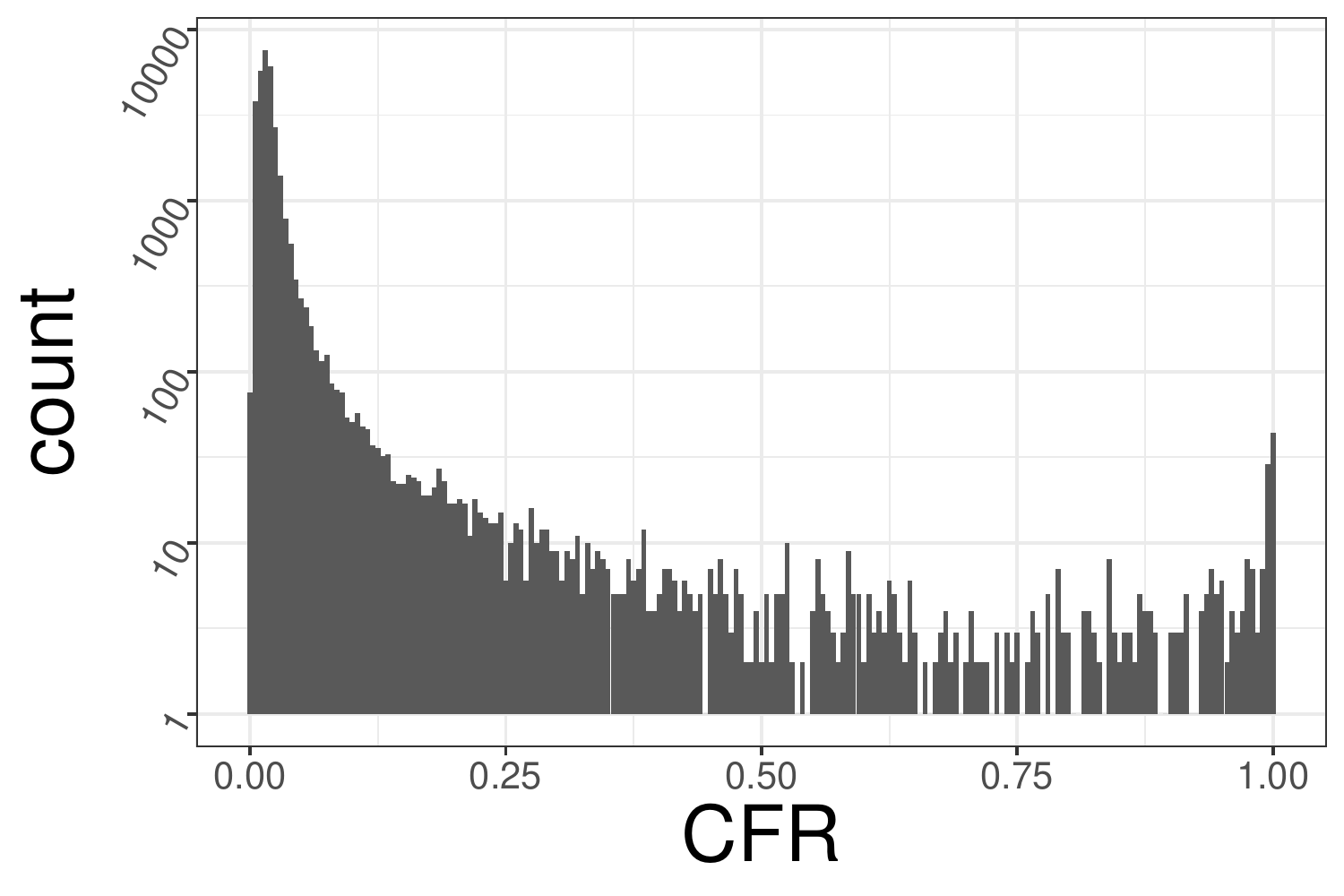}
      \subcaption{Testis}
  \end{minipage}
  \caption{Histograms of the CFRs obtained by applying GeneFishing to different tissues.}
  \label{fig:hist_CFR}
\end{figure}



\subsection{Main contributions}
Our contributions are twofold, which correspond to the two questions raised above. One tackles the identifiability issue for the binomial mixture model when $m$ can be relatively large compared to $n$. The other one answers the motivating question --- how to select the cutoff for the output of GeneFishing.

\subsubsection{New results for large $m$ in Binomial mixture model}\label{sec:contribution1}
Based on the results of \citet{teicher1963identifiability} and \citet{wood1999binomial}, the only hope is to use a large $m$ if we want to solve, or at least alleviate, the identifiability issue for arbitrary mixtures of binomial distributions. When $n$ is sufficiently large, we show that regardless of the identifiability of the model \eqref{eq:binomial_mixture}, we can find an estimator of $F$, according to the conditions on $F$, such that the estimator locates in a ball of radius $r(m)$ of $F$  in terms of some metrics such as $L_1$ distance and Kolmogorov distance, where $r(m)$ is a decreasing function of $m$. Specifically,
\begin{itemize}
\item \textbf{[Corollary of Theorem \ref{thm:CDF_deviation_Lp}]} For general $F$, if the $L_p$ distance is used, then $r(m) = \frac{C_1}{m^{\frac{1}{2p}}}$ for the empirical cumulative distribution function, where $C_1$ is a positive constant that depends on $p$.
\item \textbf{[Corollary of Theorem \ref{thm:noisy_DKW}]} If $F$ has a bounded density and the Kolmogorov distance is used, then $r(m) = \frac{C_2}{\sqrt{m}}$ for the empirical cumulative distribution function, where $C_2$ only depends on the maximal value of the density of $F$.
\item \textbf{[Corollary of Theorem \ref{thm:noisy_hist_density_upper_risk}]} If $F$ has a density whose derivative is absolutely continuous, and the truncated integrated $L_2$ distance is used, then $r(m) = \frac{C_3}{\sqrt{m}}$ for the histogram estimator, where $C_3$ only depends on the density of $F$.
\item \textbf{[Corollary of Theorem \ref{thm:noisy_l1_grenander}]} If $F$ has a bounded monotone density and the $L_1$ distance is used, then $r(m) = \frac{C_4}{\sqrt{m}}$ for the Grenander estimator, where $C_4$ only depends on the density of $F$.
\end{itemize}

\subsubsection{The cutoff selection for GeneFishing}
To model the CFRs generated by GeneFishing, we use the binomial mixture model with the U-shape constraint, under the regime where the binomial size $m$ can be relatively large compared to the sample size $n$. With the theoretical understandings mentioned in Section \ref{sec:contribution1}, we propose a simple method Ucut to identify the cutoffs of the U shape and recover the underlying distribution based on the Grenander estimator. It has been shown that when $m = \Omega(n^{\frac{2}{3}})$, the identified cutoffs converge at the rate $\calO(n^{-\frac{1}{3}})$. The $L_1$ distance between the recovered distribution and the true one decreases at the same rate. We also show that the estimated cutoff is larger than the true cutoff with a high probability. The performance of our method is demonstrated with varieties of simulation studies, a GTEX dataset used in \citep{liu2019genefishing} and a single cell dataset from Tabula Muris.

\subsection{Outline}\label{sec:bmu_outline}
The rest of the paper is organized as follows. 
Section \ref{sec:bmu_notation} introduces the notation used throughout the paper.  To answer Q1 (or more specifically, Q1'), Section \ref{sec:bmu_traditional_estimators} analyzes the estimation and inference of the underlying distribution $F$ in the binomial mixture model \eqref{eq:binomial_mixture}, under various conditions imposed on $F$. 
Equipped with these analysis tools, we cast our attention back to the GeneFishing method to answer Q2. Section \ref{sec:bmu_U_shape} first introduces a model with U-shape constraint, called BMU, to model the output of GeneFishing. The cutoffs of our interest are also introduced in this model. Then in the same section, we propose a non-parametric maximum likelihood estimation (NPMLE) method \GFCMethod\/ based on the Grenander estimator to identify the cutoffs and estimate the underlying U-shape density. We also provide a theoretical analysis of the estimator. Next, we apply \GFCMethod\/ to several synthetic datasets and two real datasets in Section \ref{sec:bmu_exp}. All the detailed proofs of the theorems mentioned in previous sections are in Appendix.

\section{Notation}\label{sec:bmu_notation}
Denote by $F$ a probability distribution. Let $s \sim F$ and $m \cdot \hat{s} \sim Binomial(m, s)$ given $s$, where $m$ is a positive integer; see Model \eqref{eq:binomial_mixture} for details. If there is no confusion, we also use $F$ to represent the associated cumulative distribution function (CDF), i.e, $F(x) = \Pbb[s \leq x]$. By $F^{(m)}$ denote the binomial mixture CDF for $\hat{s}$, i.e., $F^{(m)}(x) = \Pbb[\hat{s} \leq x]$.

Suppose there are $n$ samples independently generated from $F$, i.e., $s_1, \ldots, s_n$. Correspondingly, we have $m \cdot \hat{s}_i|s_i \sim Binomial(m, s_i)$ independently. By $F_n$ and $F_{n,m}$ denote the empirical CDF of $s_i$'s and $\hat{s}_i$'s, respectively. Specifically,
$$F_n(x) = \frac{1}{n} \sum_{i=1}^n \Ibb[s_i \leq x] ~~~~\text{and}~~~~ F_{n,m}(x) = \frac{1}{n} \sum_{i = 1}^n \Ibb[\hat{s}_i \leq x].$$
If $F$ has a density, define the Grenander estimator $\ftil_n(x)$ ($\ftil_{n,m}(x)$) for $s_i$'s ($\hat{s}_i$'s) as the left derivative of the least concave majorant of $F_n$ ($F_{n,m}$) evaluated at $x$ \citep{grenander1956theory}. 


For a density $f$, we use $f_{max}$ and $f_{\min}$ to denote its maximal and minimal function values on the domain of $f$. We use $\Ibb$ to denote the indicator function, and use $x^+$, $x^-$ to denote the right and left limit of $x$ respectively. 

\section{Estimation and Inference of $F$ under Various Conditions}\label{sec:bmu_traditional_estimators}
When $m$ is sufficiently large, $\hat{s}$ behaves like $s$, which implies that we can directly estimate the underlying true $F$ of the binomial mixture model \eqref{eq:binomial_mixture}. A natural question arises whether there exists a minimal binomial size $m$ so that the identifiability issue mentioned in Appendix \ref{appendix:bmu_identifiability} is not a concern. We investigate the estimation of general $F$, $F$ with a density, $F$ whose density has an absolutely continuous derivative, and $F$ with a monotone density. We consider the empirical CDF estimator for the first two conditions, the histogram estimator for the third condition, and the Grenander estimator for the last condition. The investigations into the estimation of $F$ under various conditions provide us with the analysis tools to design and analyze the cutoff method for GeneFishing.

\subsection{General $F$}\label{sec:bmu_bm_general_CDF}
We begin with the estimation of $F$, based on $\{\hat{s}_i\}_{i=1}^n$, without additional conditions except that $F$ is defined on $[0, 1]$. Towards this end, the empirical CDF estimator might be the first method that comes into one's mind. It is easy to interpret using a diagram, and it exists no matter whether $F$ corresponds to a density, a point mass function, or a mixture of both.

To measure the deviation of the empirical CDF from $F$, we use the $L_p$ distance with $p\geq 1$. The $L_p$ distance between two distributions $F_1$ and $F_2$ is defined as
$$\mathcal{L}_p(F_1, F_2) := \left (\int_{\mathbb{R}} \vert F_1(x) - F_2(x) \vert^p dx\right )^{\frac{1}{p}}.$$
The $L_p$ distance has two special cases that are easily interpretable from a geometric perspective. First, when $p = 1$, it looks at the cumulative differences between the two CDFs.
The $L_1$ distance is known to be equivalent to the 1-Wasserstein ($\mathcal{W}_1$) distance on $\mathbb{R}$, which is also known as the Kantorovich-Monge-Rubinstein metric.
Second, when $p = \infty$, it corresponds to the Kolmogorov-Smirnov (K-S) distance:
$$d_{KS}(F_1, F_2) := \sup_x |F_1(x) - F_2(x)|,$$
which measures the largest deviation between two CDFs. 
The K-S distance is a weaker notion of the total variation distance on $\mathbb{R}$ (total variation is often too strong to be useful). 


In the sequel, we study $\mathcal{L}_p(F^{(m)}, F)$ for $F$ supported on $[0, 1]$. Theorem \ref{thm:CDF_deviation_Lp} states that without any conditions imposed on $F$, the $L_p$ distance between $F^{(m)}$ and $F$ is bounded by $\Ocal(m^{-\frac{1}{2p}})$. One key to this theorem lies in the finite support of $F$, which enables the usage of the Fubini's theorem. Along with the Dvoretzky-Kiefer-Wolfowi (DKW) inequality \citep{massart1990tight}, it implies that the $L_p$ distance between $F_{n,m}$ and $F$ is bounded by $\Ocal(m^{-\frac{1}{2p}}) + \Ocal(n^{-\frac{1}{2}})$. Particularly, we have $\mathcal{L}_1(F^{(m)}, F) = \Ocal(m^{-\frac{1}{2}})$ and $\mathcal{L}_1(F_{n,m}, F) = \Ocal(m^{-\frac{1}{2}}) + \Ocal(n^{-\frac{1}{2}})$.

\begin{theorem}[The $L_p$ distance between $F^{(m)}$/$F_{n,m}$ and $F$]\label{thm:CDF_deviation_Lp}
  For a general $F$ on $[0, 1]$, it follows that
  $$\left ( \int_0^1 \vert F^{(m)}(x) - F(x) \vert^p dx\right )^{\frac{1}{p}} \leq \frac{C(p)}{m^{\frac{1}{2p}}},$$
  where $C(p)$ is a positive constant that depends on $p$. It indicates that
  $$\Ebb\left (\int_0^1 \vert F_{n,m}(x) - F(x) \vert^p dx\right )^{\frac{1}{p}} \leq \frac{C(p)}{m^{\frac{1}{2p}}} + \frac{K}{\sqrt{n}},$$
  where $K$ is another universal positive constant.
\end{theorem}
\begin{proof}
  By definition, it follows that
  \begin{eqnarray*}
    \vert F^{(m)}(x) - F(x) \vert &=& \vert \Ebb [\Ibb(\hat{s} \leq x)] - \Ebb [\Ibb (s \leq x)] \vert \\
                                  &=& \vert \Ebb [\Ibb(\hat{s} \leq x < s)] - \Ebb[\Ibb(s \leq x < \hat{s})] \vert.\\
                                  &\leq & \max\{\vert \Ebb [\Ibb(\hat{s} \leq x < s)] \vert, \vert \Ebb [\Ibb(s \leq x < \hat{s})]\vert\}
  \end{eqnarray*}
  Note that
  \begin{eqnarray*}
    \Ebb[\Ibb(\hat{s} \leq x < s)] &=& \Ebb[\Ebb[\Ibb(\hat{s} \leq x < s)|s]]\\
                                   &=& \Ebb[\Ebb[\Ibb(\hat{s} - s \leq x - s < 0)|s]]\\
                                   &\leq& \Ebb[\exp\{-m(x-s)^2/2\}],
  \end{eqnarray*}
  where $\hat{s} - s$ is bounded in $[-1, 1]$, and thus it is a sub-Gaussian random variable with the variance less or equal to $1$ \citep{hoeffding1993}. The same argument applies to $\Ebb[\Ibb(s\leq x < \hat{s})]$. Therefore, we have
  \begin{eqnarray*}
    && \left ( \int_0^1 \vert F^{(m)}(x) - F(x) \vert^p dx \right )^{\frac{1}{p}}\\
    &\leq&  \left (\int_0^1 \Ebb [\exp\{-mp(x-s)^2/2\}]dx \right )^{\frac{1}{p}}\\
                                              &=& \left( \int_0^1 \int_0^1 \exp\{-mp(x-s)^2/2\} dF(s)dx \right )^{\frac{1}{p}}\\
                                              &\overset{(i)}{=}& \left (\int_0^1 \int_0^1 \exp\{-mp(x-s)^2/2\} dxdF(s) \right )^{\frac{1}{p}}\\
                                              &\leq& \left (\int_0^1 \frac{\sqrt{2\pi}}{\sqrt{mp}} dF(s)\right )^{\frac{1}{p}}\\
    &=& \frac{(2\pi)^{\frac{1}{2p}}}{(mp)^{\frac{1}{2p}}},
  \end{eqnarray*}
  where Equation $(i)$ holds by the Fubini's theorem. Further, by noting that $F_{n,m}(x) - F(x) = F_{n,m}(x) - F^{(m)}(x) + F^{(m)}(x) - F(x)$ and using the DKW inequality, it follows that
  $$\Ebb \int_0^1 \vert F_{n,m}(x) - F(x)\vert dx \leq \frac{(2\pi)^{\frac{1}{2p}}}{(mp)^{\frac{1}{2p}}} + \frac{\sqrt{2\pi}}{\sqrt{n}}.$$
\end{proof}

Notwithstanding, Theorem \ref{thm:CDF_deviation_Lp} does not establish a useful bound for the K-S distance that corresponds to the $L_{\infty}$ distance --- there remains a non-negligible constant $\underset{p\rightarrow \infty}{\lim} \frac{(2\pi)^{\frac{1}{2p}}}{(mp)^{\frac{1}{2p}}} = 1$ which does not depend on $m$. In fact, the K-S distance might evaluate the estimate of $F$ from a too stringent perspective. Proposition \ref{prop:CDF_deviation_pmf} shows that no matter how small $m$ is, there is an $F$ with a point mass function and a point $x_0$ such that $\vert F^{(m)}(x_0) - F(x_0)\vert$ is larger than a constant that is independent of $m$. This result is attributed to the ``bad'' points with non-trivial masses like $x_0$. Such a ``bad'' point gives rise to a sharp jump in $F$, which $F^{(m)}$ cannot immediately catch up with due to the discretization nature of the binomial randomness. It leads to difficulty in recovering the underlying distribution $F$ of the binomial mixture model.

On the other hand, the $L_p$ distance with $p < \infty$ does not suffer from the issue of the K-S distance --- it can be regarded as looking at an average of the absolute distance between $F^{(m)}$ and $F$ when the support of $F$ is finite. To be specific, even if there are ``bad'' points $x_1, x_2, \ldots$ such that $|F^{(m)}(x_i) - F(x_i)|$ has a non-trivial difference, $i = 1, 2, \ldots$, this difference will diminish outside a small neighbor of $x_i$ of a width $\Ocal(\frac{1}{m})$. Therefore, when taking the integral, the averaging distance decreases as $m$ grows. Furthermore, if $F$ has a density, the issue of ``bad'' points no longer exists for the K-S distance either. In this case, the K-S distance is an appropriate choice to measure the difference between $F^{(m)}$ and $F$; see Section \ref{sec:bmu_bm_CDF_deviation} for details.

\begin{proposition}[$F^{(m)}$ can deviate in K-S far from $F$ with a pmf]\label{prop:CDF_deviation_pmf}  
  There exists an $F$ with a pmf, such that $\sup_{x}|F^{(m)}(x) - F(x)| \geq c > 0$, where $c$ is a constant.
\end{proposition}
\begin{proof}
  Let $F(x)$ be the delta function taking the mass at $\frac{1}{2} + \kappa$, where $\kappa$ is an extremely small positive irrational number. Then by CLT, with probability about $1/2$, $\hat{s}$ is no larger than $\frac{\tilde{m}}{m}$, where $\tilde{m}$ is the largest integer such that $\frac{\tilde{m}}{m} < \frac{1}{2} + \kappa$. Take any $x_0$ in $(\frac{\tilde{m}}{m}, \frac{1}{2} + \kappa)$. It follows that $F(x_0) = 0$ but $F^{(m)}(x_0) \approx \frac{1}{2}$. It implies that $F^{(m)}(x_0) - F(x_0) $ is larger than a constant that is independent of $m$.
\end{proof}

\subsection{$F$ with a density}\label{sec:bmu_bm_CDF_deviation}
In this section, we focus on $F$ with a density so that the K-S distance can be employed. 
In addition, we stick to this metric partly because it is related to the Grenander estimator for monotone density estimation \citep{grenander1956theory, birge1989grenander}, which constitutes our method for the GeneFishing cutoff selection; see Section \ref{sec:bmu_bm_grenander} and Section \ref{sec:bmu_method_npmle}. 

To bound the K-S distance between $F_{n,m}$ and $F$, i.e., $\sup_x |F_{n,m}(x) - F(x)|$, we just need to bound $F_{n,m} - F^{(m)}$ and $F^{(m)} - F$ by noticing that $F_{n,m} - F = F_{n,m} - F^{(m)} + F^{(m)} - F$.
By the DKW inequality, we have a tight bound
$$\mathbb{P}[\sup_x |F_{n,m}(x) - F^{(m)}(x)| > t ] \leq 2 \exp\{-2nt^2\},~~~~\forall t > 0.$$
So it only remains to study the deviation of $F^{(m)}$ from $F$. In Proposition \ref{prop:CDF_deviation}, we show that when $F$ has a derivative bounded from both  below and above, the K-S distance between $F^{(m)}$ and $F$ is bounded by $\Ocal(\frac{1}{m})$ from below and by $\Ocal(\frac{1}{\sqrt{m}})$ from above. 

\begin{proposition}[Deviation of $F^{(m)}$ from $F$ with a density]\label{prop:CDF_deviation}
  
    Suppose $f$ is a density function on $[0, 1]$ with $0 < f_{\min} \leq f_{\max} < \infty$. It follows that 
$$\frac{f_{\min}}{m+1} \leq \sup_x |F^{(m)}(x) - F(x)| \leq f_{\max} \cdot \frac{\sqrt{2\pi}}{\sqrt{m}}.$$
\end{proposition}
\begin{proof}
  For the lower bound, we have
  $$\Pbb(\hat{s} \leq 0) - \Pbb(s \leq 0) = \Pbb(\hat{s} \leq 0) = \int_0^1 (1 - u)^mf(u)du \geq f_{\min} \int_0^1 (1 - u)^m du = \frac{f_{\min}}{m+1}.$$  
  For the upper bound, note that
  \begin{eqnarray*}
    F^{(m)}(x) - F(x) = \Pbb(\hat{s} \leq x) - \Pbb(s \leq x)
                      = \Ebb (\Ibb[\hat{s} \leq x] - \Ibb [s \leq x])
                      = \Ebb (\Ibb [\hat{s} \leq x < s] - \Ibb [s \leq x < \hat{s}]),
  \end{eqnarray*}
  thus
  $$|F^{(m)}(x) - F(x)| \leq \max \{\Ebb \Ibb [\hat{s} \leq x < s], \Ebb \Ibb [s \leq x < \hat{s}]) \}.$$
We have
  \begin{eqnarray*}
    \Ebb \Ibb[\hat{s} \leq x < s] &=& \Pbb(\hat{s} \leq x < s)\\
    &=&\Pbb(\hat{s} - s \leq x - s < 0) \\
    &=& \Ebb[\Pbb(\hat{s} - s \leq x - s < 0) | s] \\
    &\leq& \Ebb[\exp(-m(x - s)^2/2)] \\
    &=& \int_0^1 \exp(-m(x - u)^2/2) f(u) du\\
    &\leq& f_{\max} \cdot \frac{\sqrt{2\pi}}{\sqrt{m}},
  \end{eqnarray*}
  Similarly, we have $\Ebb \Ibb[s \leq x < \hat{s}] \leq f_{\max} \cdot \frac{\sqrt{2\pi}}{\sqrt{m}}$. So it follows that
  $$\sup_x |\Pbb (\hat{s}\leq x) - \Pbb(s \leq x)| \leq \max \{ \sup_x \Ebb \Ibb [\hat{s} \leq x < s], \sup_x \Ebb \Ibb [s \leq x < \hat{s}]\} \leq f_{\max} \cdot \frac{\sqrt{2\pi}}{\sqrt{m}}.$$
\end{proof}
Proposition \ref{prop:CDF_deviation} shows that the largest deviation of the binomial mixture CDF from the underlying CDF is at least the order $\Ocal(m^{-1})$ and at most the order $\Ocal(m^{-\frac{1}{2}})$. In fact, the condition that $f$ is bounded can be relaxed to that $f$ is $L_p$-integrable with $p > 1$ using the H{\" o}lder inequality, but the rate will be $\Ocal(m^{-\frac{p}{2(p-1)}})$ correspondingly. Our result is the special case with $p = \infty$. Moreover, $F$ with a density is a necessary condition for Proposition \ref{prop:CDF_deviation} --- we have seen in Proposition \ref{prop:CDF_deviation_pmf} that if $F$ has a point mass function, the deviation of $F^{(m)}$ from $F$ cannot be controlled w.r.t $m$.

Proposition \ref{prop:CDF_deviation_bounds} shows that there exist two simple distributions that respectively attain the lower bound and the upper bound. On the other hand, Proposition \ref{prop:CDF_deviation} can be further improved: if the underlying density is assumed to be smooth, the upper bound is lowered to $\Ocal(m^{-1})$; see Proposition \ref{prop:CDF_deviation_smooth} in Section \ref{sec:bmu_bm_hist}. 

\begin{proposition}[Tightness of Proposition \ref{prop:CDF_deviation}]\label{prop:CDF_deviation_bounds}
  
  The upper bound and the lower bound in Proposition \ref{prop:CDF_deviation} are tight. In other words, there exist an $F_1$ and $F_2$ such that $\sup_{x}|F_1^{(m)}(x) - F_1(x)| \leq C_1 \cdot \frac{1}{m+1}$ and $\sup_{x}|F_2^{(m)}(x) - F_2(x)| \geq C_2 \cdot \frac{1}{\sqrt{m}}$, where $C_1$ and $C_2$ are two positive constants.
\end{proposition}
\begin{proof}
  The lower bound of Proposition \ref{prop:CDF_deviation} can be attained by the uniform distribution. Specifically, if $f \equiv 1$, $\Pbb(\hat{s} = k/m) = \frac{1}{m+1}$. So $|\Pbb(\hat{s} \leq 0) - \Pbb(s \leq 0)| = \frac{1}{m+1}$.

  On the other hand, the upper bound can be attained by another simple distribution. Let
  \begin{equation}\label{eq:two_step_func}
    f(x) = 1.8 \cdot \Ibb(x \in [0, 1/2]) + 0.2 \cdot \Ibb(x \in (1/2, 1])
  \end{equation}
  We can show that this density $f$ leads to $|\Pbb(\hat{s} \leq 1/2) - \Pbb(s \leq 1/2)| \geq \frac{C}{\sqrt{m}}$, where $C$ is a positive constant. It is a consequence of the central limit theorem for the binomial distribution. 
  The detailed proof is delegated to Appendix \ref{appendix:bmu_proof_CDF_deviation_bounds}.
\end{proof}

Given Proposition \ref{prop:CDF_deviation}, we can get the rate of of $\sup_{x}|F_{n,m}(x) - F(x)|$ along with the DKW inequality, which is $\Ocal(n^{-\frac{1}{2}})+\Ocal(m^{-\frac{1}{2}})$ as shown in Theorem \ref{thm:noisy_DKW}. By taking integral of $\Pbb(\sup_x [F_{n,m}(x) - F(x)] > t)$ w.r.t $t$, we immediately have Corollary \ref{corollary:noisy_DKW}. 

\begin{theorem}[The rate of K-S distance between $F_{n,m}$ and $F$ with a density]\label{thm:noisy_DKW}
  ~\\
  Suppose $f$ is a density function on $[0, 1]$ with $f_{\max} < \infty$. The data is generated as Model \eqref{eq:binomial_mixture}. It follows that
  $$\Pbb(\sup_x [F_{n,m}(x) - F(x)] > t) \leq \exp(-nt^2/2) + \Ibb (f_{\max} \cdot \frac{2\sqrt{2\pi}}{\sqrt{m}} > t),$$
where $t \geq \sqrt{\frac{\log 2}{2n}}$. The two-side tail bound also holds as follow
  $$\Pbb(\sup_x |F_{n,m}(x) - F(x)| > t) \leq 2\exp(-nt^2/2) + \Ibb (f_{\max} \cdot \frac{2\sqrt{2\pi}}{\sqrt{m}} > t),$$
  where $t > 0$.
\end{theorem}
\begin{proof}
  Note that
  $$\sup_x |F_{n,m}(x) - F(x)| \leq \sup_x |F_{n,m}(x) - F^{(m)}(x)| + \sup_x |F^{(m)}(x) - F(x)|.$$
  The first term can be bounded by the original DKW inequality and the second term can be bounded using the result of Proposition \ref{prop:CDF_deviation}. Then we conclude the second result. The first result can be obtained in the same fashion.
\end{proof}

\begin{corollary}\label{corollary:noisy_DKW}
  
  Under the same setup of Theorem \ref{thm:noisy_DKW}, we have
  $\Ebb \sup_x [F_{n,m}(x) - F(x)] \leq \frac{\sqrt{2\pi}}{\sqrt{n}} + \min\{1, \frac{f_{\max}\cdot 2\sqrt{2\pi}}{\sqrt{m}}\}$,
  and
  $\Ebb \sup_x |F_{n,m}(x) - F(x)| \leq \frac{2\sqrt{2\pi}}{\sqrt{n}} + \min\{1, \frac{f_{\max}\cdot 2\sqrt{2\pi}}{\sqrt{m}}\}$.
\end{corollary}
\begin{proof}
  By Theorem \ref{thm:noisy_DKW}, it follows that
  \begin{eqnarray*}
    \Ebb[\sup_x [F_{n,m}(x) - F(x)]] &=& \int_0^1 \Pbb(\sup_x [F_{n,m}(x) - F(x)] > t) dt\\
    &\leq& \int_0^1 [\exp(-nt^2/2) + \Ibb (f_{\max} \cdot \frac{2\sqrt{2\pi}}{\sqrt{m}} > t)] dt \\
    &\leq& \frac{\sqrt{2\pi}}{\sqrt{n}} + \min\{1, \frac{f_{\max}\cdot 2\sqrt{2\pi}}{\sqrt{m}}\}.
  \end{eqnarray*}
  The two-side expectation can be proved in a similar manner.
\end{proof}

\subsection{$F$ with A Smooth Density}\label{sec:bmu_bm_hist}
In this section, we investigate the estimation of $F$ with a smooth density. Under this condition, we first obtain a stronger result than Proposition \ref{prop:CDF_deviation} for $F^{(m)}$, based on a truncated K-S distance on the interval $[a, 1-a]$, where $0 < a < 1/2$. Proposition \ref{prop:CDF_deviation_smooth} shows that if the density of $F$ is smooth, the truncated K-S distance decreases at $\Ocal(\frac{1}{m})$. The proof is based on the fact the binomial distribution random variable $m\cdot \hat{s}_i$ behaves like a Gaussian random variable when the binomial probability $s_i$ is bounded away from $0$ and $1$. When $s_i$ is close to $0$ and $1$, the Gaussian approximation cannot be used since it has an unbounded variance $\frac{1}{ms_i(1-s_i)}$. The proof is deferred to Appendix \ref{appendix:bmu_proof_CDF_deviation_smooth}.

\begin{proposition}[Deviation of $F^{(m)}$ from $F$ with a smooth density]\label{prop:CDF_deviation_smooth}
  
    Suppose $f$ is a density function on $[0, 1]$ with $f'$ being absolutely continuous. Let $s \sim f$ and $m \cdot \hat{s} \sim Binomial(m, s)$. It follows for any $0< a < 1/2$ that
    $$\sup_{x \in [a, 1-a]} |F^{(m)}(x) - F(x)| \leq C \cdot \frac{1}{m},$$
    where $C$ is some constant that only depends on $f$ and $a$.
\end{proposition}

Then, we investigate the histogram estimator, since it is the one of simplest nonparametric density estimators and it has a theoretical guarantee when the density is smooth. Let $L$ be an integer and define bins
$$B_1 = [0, \frac{1}{L}), B_2 = [\frac{1}{L}, \frac{2}{L}), \ldots, B_L = [\frac{L-1}{L}, 1].$$
Let $Y_l$ be the number of observations in $B_l$, $\hat{p}_l = Y_l/n$ and $p_l = \Pbb(s_1 \in B_l)$. It is known that under certain smoothness conditions, the histogram converges in a cubic root rate for the rooted mean squared error (MSE) \citep[Chapter 6]{wasserman2006all}. Next we study how the histogram behaves on the binomial mixture model. Denote
\begin{equation*}
\left \{
\begin{array}{ll}
  \hat{f}_{n,m}(x) = \frac{1}{n} \sum_{i = 1}^n \frac{1}{h} \Ibb (\hat{s}_i \in B(x))\\
  \hat{f}_n(x) = \frac{1}{n} \sum_{i = 1}^n \frac{1}{h} \Ibb (s_i \in B(x)),
\end{array} \right.
\end{equation*}
where $h$ is the bandwidth, $B(x)$ denotes the bin that $x$ falls in. Theorem \ref{thm:noisy_hist_density_upper_risk} shows that the histogram estimator based on $\hat{s}_i$'s has the same convergence rate in terms of the MSE metrics as the histogram estimator based on $s_i$'s if $m = \Omega(n^{\frac{2}{3}})$ and $h = \Ocal(n^{-\frac{1}{3}})$. This rate might not be improved even if $f$ has higher order continuous derivatives since $\Ebb |\hat{f}_{n,m}(x) - \Ebb \hat{f}_n(x)|$ is bounded by $\Ocal(\frac{1}{\sqrt{m}} + h + \frac{1}{mh})$, which dominates $\Ebb|\hat{f}_n(x) - f(x)| = \Ocal(h^{\nu})$ when $f$ has a $\nu$-order continuous derivative. The proof of Theorem \ref{thm:noisy_hist_density_upper_risk} is delegated to Appendix \ref{appendix:bmu_proof_noisy_hist_density_upper_risk}.

\begin{theorem}[Upper bound of the histogram risk for binomial mixture model]\label{thm:noisy_hist_density_upper_risk}
  ~\\
  Let $R(a, b) = \int_a^{b} \Ebb(f(x) - \hat{f}_{n,m}(x))^2 dx$ be the integrated risk on the interval $[a, b]$. Assume that $f^{\prime}$ is absolutely continuous. It follows that 
  $$R(a, 1 - a) \leq C_1 \cdot (h^2 + \frac{1}{m} + \frac{1}{m^2h^2} + \frac{1}{nh}), \forall 0 < a < \frac{1}{2},$$
  Furthermore, if $m \geq C_2 \cdot n^{\frac{2}{3}}$, $h = C_3 \cdot n^{-\frac{1}{3}}$, we have
  $$R(a, 1 - a) \leq C_4 \cdot n^{-\frac{2}{3}}, \forall 0 < a < \frac{1}{2}.$$
  Here $C_1$, $C_2$ and $C_4$ are positive constants that only depend on $a$ and $f$, $C_3 > 0$ only depends on $f$.
\end{theorem}

Finally, we conclude this section with a study on $F$ whose density is discretized into a point mass function of $K$ non-zero masses. In contrast to the existing results in \citet{teicher1963identifiability}, we allow $K$ to be as large as $\sqrt{n}$, and study the finite-sample rate of the histogram estimator. Let $p(x)$ be a point mass function such that
\begin{equation*}
p(x) = \sum_{k=1}^K \alpha(k) \Ibb(x = x_k),\label{eq:hist_pmf}
\end{equation*}
where $\alpha(k) \geq 0$ and $\sum_{k=1}^K \alpha(k) = 1$, $x_k = \frac{(k-1) + 1/2}{K}$. Denote by $I_k$ the interval centered at $x_k$ and of length $1/K$. We can make $(\alpha(1), \ldots, \alpha(K))$ ``smooth'' by letting $\alpha(k) = \int_{I_k} f(t)dt$, where $f$ is a smooth function. Denote
\begin{equation*}
\left \{
\begin{array}{ll}
  \hat{\alpha}_{n,m}(k) = \frac{1}{n} \sum_{i = 1}^n \Ibb (\hat{s}_i \in I_k)\\
  \hat{\alpha}_n(k) = \frac{1}{n} \sum_{i = 1}^n \Ibb (s_i \in I_k) \\
\end{array} \right.
\end{equation*}
Then the MSE can be defined as $R(\hat{\alpha}, \alpha) = \frac{1}{K}\sum_{k=1}^K (\hat{\alpha}(k) - \alpha(k))^2$. It is known that $R(\hat{\alpha}_n, \alpha) = \Ocal(\frac{1}{n})$. Theorem \ref{thm:noisy_hist_mass_upper_risk} shows that the same rate can be achieved by $\hat{\alpha}_{n,m}$ when $m = \Omega(\sqrt{n}\max\{\sqrt{n}, K\})$. The proof is deferred to Appendix \ref{appendix:bmu_proof_noisy_hist_mass_upper_risk}.

\begin{theorem}[Upper bound of the histogram risk for finite binomial mixture]\label{thm:noisy_hist_mass_upper_risk}
  ~\\
  Let $\alpha(k) = \int_{I_k} f(t)dt$, where $f^{\prime}$ is absolutely continuous and $\int (f^{\prime})^2 dx < \infty$. Let
  $R(a, b) = \frac{\sum_{aK \leq k \leq bK} (\hat{\alpha}_{n,m}(k) - \alpha(k))^2}{\sum_{k = 1}^K \Ibb(aK \leq k \leq bK)}$ be the risk on the interval $[a, b]$.  It follows that
  $$R(a, 1 - a) \leq C_1 \cdot (\frac{1}{n} + \frac{1}{m} + \frac{K^2}{m^2}), \forall 0 < a < \frac{1}{2}.$$
  Furthermore, if $m \geq C_2 \cdot \sqrt{n}\max\{\sqrt{n}, K\}$, then 
  $$R(a, 1 - a) \leq C_3 \cdot \frac{1}{n}, \forall 0 < a < \frac{1}{2}.$$
  Here $C_1, C_2, C_3$ are positive constants that only depend on $a$ and $f$.
\end{theorem}

\subsection{$F$ with A Monotone Density}\label{sec:bmu_bm_grenander}
Next, we shift our attention to $F$ with a monotone density $f$. It is motivated by the U shape in the histograms of the GeneFishing output, where we can decompose the U shape into a decreasing part, a flat part, and an increasing part; see Section \ref{sec:bmu_model}. To estimate $f$, a natural solution is the Grenander estimator \citep{grenander1956theory, jankowski2009estimation}. Specifically, we first construct the least concave majorant of the empirical CDF of $F$; then its left derivative is the desired estimator. The detailed review of the Grenander estimator is deferred to Appendix \ref{appendix:bmu_review_grenander}.

In the sequel, we establish convergence and inference results for $\ftil_{n,m}$ that is the Grenander estimator based on $\hat{s}_i$'s. Theorem \ref{thm:noisy_l1_grenander} states $\tilde{f}_{n,m}$ achieves the convergence rate of $\Ocal(n^{-\frac{1}{3}})$ w.r.t the $L_1$ distance if $m = \Omega(n^{\frac{2}{3}})$.


\begin{theorem}[$L_1$ convergence of $\ftil_{m,n}$]\label{thm:noisy_l1_grenander}
  
  Suppose $f$ is a decreasing density on $[0, 1]$ with $f_{\max} < \infty$. It follows that 
  $$\Ebb_f \int_{0}^1 |\ftil_{n, m}(x) - f(x)| dx \leq C_1 \cdot n^{-\frac{1}{3}} + C_2 \cdot m^{-\frac{1}{2}}.$$
  Furthermore, if $m \geq C_3 \cdot n^{\frac{2}{3}}$, we have
    $$\Ebb_f \int_{0}^1 |\ftil_{n, m}(x) - f(x)| dx \leq  C_4 \cdot n^{-\frac{1}{3}}.$$
  Here $C_1, C_2, C_3, C_4$ are positive constants that only depend on $f$.
\end{theorem}
Theorem \ref{thm:noisy_l1_grenander} follows by Corollary \ref{corollary:noisy_DKW} and \citet{birge1989grenander}. The details can be found in Appendix \ref{appendix:bmu_proof_noisy_l1_grenander}. Furthermore, we discuss the minimal $m$ we need to ensure the $\Ocal(n^{-\frac{1}{3}})$ convergence rate in terms of the $L_1$ distance. We have seen that $m = \Omega(n^{\frac{2}{3}})$ is a sufficient condition. The question arises naturally whether it is necessary. Based on the density \eqref{eq:two_step_func} used in Proposition \ref{prop:CDF_deviation_bounds}, we can show that $m = \Omega(n^{\frac{2}{3}})$ is also the minimal binomial size; see Proposition \ref{prop:lower_bound_noisy_l1_grenander}.

\begin{proposition}\label{prop:lower_bound_noisy_l1_grenander}
  When $f$ takes \eqref{eq:two_step_func}, i.e.,
$$f(x) = 1.8 \cdot \Ibb(x \in [0, 1/2]) + 0.2 \cdot \Ibb(x \in (1/2, 1]),$$
  it needs $m \geq C_1 \cdot n^{\frac{2}{3}}$ to ensure
  $$\Ebb_f \int_0^1 |\ftil_{n,m}(x) - f(x) | dx \leq C_2 \cdot n^{-\frac{1}{3}},$$
  where $C_1, C_2 > 0$ only depend on $f$.
\end{proposition}
\begin{proof}
  Note that
  \begin{eqnarray*}
    \Ebb_f \int_0^1 \vert \ftil_{n,m}(x) - f(x) \vert dx &\overset{(i)}{\geq}&  \Ebb_f \int_0^{x^*} \vert \ftil_{n,m}(x)  - f(x)\vert dx\\                                                 
                                                 &\geq&  \Ebb_f\vert \int_0^{x^*} [\ftil_{n,m}(x)  - f(x)]dx \vert \\
                                                 &=&  \Ebb_f \vert \tilde{F}_{n,m}(x^*) - F(x^*)\vert\\
                                                  &\geq&  \Ebb_f \vert F^{(m)}(x^*) - F(x^*) \vert - \Ebb_f \vert \tilde{F}_{n,m}(x^*) - F^{(m)}(x^*) \vert\\
    &\overset{(ii)}{\geq}&  \Ebb_f\vert F^{(m)}(x^*) - F(x^*) \vert -  \sup_x \Ebb_f\vert F_{n,m}(x) - F^{(m)}(x) \vert\\
    &\overset{(iii)}{\geq}& \frac{K_1}{\sqrt{m}} - \frac{K_2}{\sqrt{n}},
  \end{eqnarray*}
  where $x^* := arg\max_x \vert F^{(m)}(x) - F(x) \vert$ in Inequality (i); Inequality (ii) holds because $\sup_{x}\vert \tilde{F}_{n,m}(x) - F^{(m)}(x)\vert \leq \sup_x \vert F_{n,m}(x) - F^{(m)}(x)\vert$, which can be easily derived by Marshall's lemma \citep{marshall1970discussion}; Inequality (iii) holds because of Proposition \ref{prop:CDF_deviation_bounds} and the DKW inequality, and $K_1, K_2 > 0$ are two positive constants that only depend $f$. To ensure the expected $L_1$ distance between $\ftil_{n,m}$ and $f$ is bounded by $C_2 \cdot n^{-\frac{1}{3}}$, it is necessary to have
  $$\frac{C_2}{\sqrt[3]{n}} \geq \frac{K_1}{\sqrt{m}} - \frac{K_2}{\sqrt{n}},$$
  which implies that $m \geq C_1 \cdot n^{\frac{2}{3}}$, for some positive constant $C_1$ that depends on $f$.
\end{proof}

Next, we study the local asymptotics of $\tilde{f}_{n,m}$. For the binomial mixture model, we yield a similar result for $\ftil_{n,m}$ as the local asymptotics of $\ftil_n$ when $m$ grows faster than $n$, as shown in Theorem \ref{thm:noisy_grenander_local_asymp}.

\begin{theorem}[Local Asymptotics of $\ftil_{m,n}$]\label{thm:noisy_grenander_local_asymp}
  
  Suppose $f$ is a decreasing density on $[0, 1]$ and is smooth at $t_0 \in (0, 1)$.  Then when $\frac{m}{n} \rightarrow \infty$ as $n \rightarrow \infty$, if $f_{\max} < \infty$, we have
  \begin{itemize}
  \item [(A)] If $f$ is flat in a neighborhood of $t_0$, and $[a, b]$ is the flat part containing $t_0$, it follows that
    $$\sqrt{n}(\ftil_{m,n}(t_0) - f(t_0)) \overset{d}{\rightarrow} \hat{S}_{a, b}(t_0),$$
    where $\hat{S}_{a, b}(t)$ is the slope at $F(t)$ of the least concave majorant in $[F(a), F(b)]$ of a standard Brownian Bridge in $[0, 1]$.
  \item [(B)] If $f(t) - f(t_0) \sim f^{(k)}(t_0)(t - t_0)^k$ near $t_0$ for some $k$ and $f^{(k)}(t_0) < 0$, it follows that
    $$n^{\frac{k}{2k+1}}[\frac{f^k(t_0) |f^{(k)}(t_0)|}{(k+1)!}]^{-\frac{1}{2k+1}} (\ftil_{m,n}(t_0) - f(t_0)) \overset{d}{\rightarrow} V_k(0),$$
    where $V_k(t)$ is the slope at $t$ of the least concave majorant of the process $\{B(t) - |t|^{k+1}, t \in (-\infty, \infty)\}$, and $B(t)$ is a standard two-sided Brownian motion on $\Rbb$ with $B(0) = 0$.
  \end{itemize}
\end{theorem}

The proof of Theorem \ref{thm:noisy_grenander_local_asymp} relies on Proposition \ref{prop:CDF_deviation} and the Koml{\' o}s-Major-Tusn{\' a}dy (KMT) approximation \citep{komlos1975approximation}. Given these two results, we can show that if $f$ is upper bounded, there exists a sequence of Brownian bridges $\{B_n(x), 0 \leq x \leq 1\}$ such that
$$\Pbb \left \{ \sup_{0 \leq x \leq 1} |\sqrt{n} (F_{n,m}(x) - F(x)) - B_n(F(x))| > \frac{\tilde{a}\sqrt{n}}{\sqrt{m}} + \frac{a\log n}{\sqrt{n}} + t \right \} \leq b (e^{-c\sqrt{n} t} + e^{-dmt^2}),$$
where $\tilde{a} > 0$ only depends on $f$ and $a, b, c, d$ are universal positive constants. 
Then following \citet{wang1992nonparametric}, we can prove Theorem \ref{thm:noisy_grenander_local_asymp}. The details are deferred to Appendix \ref{appendix:bmu_proof_noisy_grenander_local_asymp}.

Finally, we conclude this section by discussing the histogram estimator and the Grenander estimator (for a density). Both of them are bin estimators but differ in the choice of the bin width. One can pick the bin width for the histogram to attain optimal convergence rates \citep{wasserman2006all}. On the other hand, the bin widths of the  Grenander estimator are chosen completely automatically by the estimator and are naturally locally adaptive \citep{birge1989grenander}. The consequence is that the Grenander estimator can guarantee monotonicity, but the histogram estimator cannot. If the underlying model is monotone, the Grenander estimator has a better convergence rate than the histogram estimator. Notably, the convergence theory of the histogram estimator cannot be established unless the density is smooth, while that of the Grenander estimator only requires the density is monotone and $L_p$ integrable ($p > 2$) \citep{birge1989grenander}.

In our setup, we show that when $m = \Omega(n^{\frac{2}{3}})$, both the Grenander estimator and the histogram estimator, based on $\{\hat{s}_i\}_{i=1}^n$, have the same rate at $\calO(n^{-\frac{1}{3}})$ in $L_1$ distance (the $L_2$ convergence of the histogram can imply the $L_1$ convergence). It seems that both methods are comparable. Nonetheless, we mention that the conditions for the convergence of the two estimators are different. The Grenander estimator requires a bounded monotone density, while the histogram requires a smooth density. 

\section{Estimation and Inference of $F$ with U-shape Constraint}\label{sec:bmu_U_shape}
Now we have sufficient insight into the estimation of $F$ under various conditions in the binomial mixture model \eqref{eq:binomial_mixture}. We are ready to cast our attention back to the cutoff selection problem for the GeneFishing method, i.e., distinguishing the relevant genes (to the baits genes or the associated biological process) from the irrelevant ones. To answer this question, we leverage the observation that the histogram of $\{\hat{s}_i\}_{i=1}^n$ appears to have a U shape as shown in Figure \ref{fig:hist_CFR}.

\subsection{Model}\label{sec:bmu_model}
We decompose the density or the pmf of $F$ into three parts: the first part decreases, the second part remains flat, and the last part increases. The first part is assumed to be purely related to the irrelevant genes; the second part is associated with the mixture of the irrelevant and the relevant genes; the last part is purely corresponding to the relevant genes. Denote by $c_l$ and $c_r$ the transition points from the first part to the second part and the second part to the third part, respectively. Then the question is reduced to identifying $c_r$ and getting an upper confidence bound on $c_r$. In the sequel, we formally write this assumption as Assumption \ref{asmpt:U_shape} when $F$ is associated with a continuous random variable. The corresponding mathematical formulations for the pmf are similar, so we omit them here.

\begin{assumption}\label{asmpt:U_shape}
Let $f$ be the derivative of $F$, i.e., the probability density function. We assume $f$ consists of three parts:
\begin{equation}
  f(x) = \left \{
  \begin{array}{ll}
    f_l(x) = \alpha_l \cdot g_l(x), &\text{ if } x \in [0, c_l]\\
    \frac{\alpha_{mid}}{c_r - c_l}, & \text{ if } x \in (c_l, c_r]\\
    f_r(x) = \alpha_r \cdot g_r(x), & \text{ if } x \in (c_r, 1]
  \end{array}
\right.,
\label{eq:U_constraint}
\end{equation}
where $0 < c_l < c_r < 1$, $g_l$ is a decreasing function, $g_r$ is an increasing function such that $\int_{0}^{c_l} g_l(x) dx = 1$, $\int_{1}^{c_r} g_r(x) dx = 1$, and $\alpha_l + \alpha_r + \alpha_{mid} = 1$ with $\alpha_l, \alpha_r, \alpha_{mid} > 0$.
\end{assumption}
For the U-shaped constraint, we also need
$$\min\{f_l(c_l^-), f_r (c_r^+)\} \geq \frac{\alpha_{mid}}{c_r - c_l}.$$

The shape constraint \eqref{eq:U_constraint} is determined by six parameters $\{\alpha_l, \alpha_r, c_l, c_r, g_l, g_r\}$, but they are not identifiable. Below is an example of such unidentifiability.
\begin{example}[Identifiability Issue for \eqref{eq:U_constraint}]\label{ex:identifiability_U}
  \begin{equation*}
  \begin{array}{ll}
    \tilde{\alpha}_l = \alpha_l + \frac{\alpha_{mid}}{c_r - c_l} \cdot \tau;  & \tilde{\alpha}_r = \alpha_r\\
    \tilde{c}_l = c_l + \tau; & \tilde{c}_r = c_r \\
    \tilde{g}_l = \left \{ \begin{array}{ll}
      g_l\cdot \alpha_l /\tilde{\alpha}_l, & \text{ if } x \in [0, c_l]\\
      \frac{\alpha_{mid}}{(c_r - c_l) \cdot \tilde{\alpha}_l}, & \text{ if } \in (c_l, c_l + \tau]
    \end{array}
    \right . ; & \tilde{g}_r = g_r
  \end{array}
\end{equation*}
The parameters $\{\tilde{\alpha}_l, \tilde{\alpha}_r, \tilde{c}_l, \tilde{c}_r, \tilde{g}_l, \tilde{g}_r\}$ yield the same model as $\{\alpha_l, \alpha_r, c_l, c_r, g_l, g_r\}$ if $\tau < c_r - c_l$.
\end{example}
The identifiability issue results from the vague transitions from one part to the next adjacent part in Model \eqref{eq:U_constraint}. To tackle it, we need to introduce some assumptions to sharpen the transitions. For example, if $f$ is smooth, then a sharp transition means the first derivative of $f$, i.e., the slope, significantly changes at this point. In general, we do not impose the smoothness on $f$ and use the finite difference as the surrogate of the slope. To be specific, suppose there exist $\delta_l, \delta_r > 0$ and neighborhoods of  $c_l$ and $c_r$, whose sizes are $\tau_l$ and $\tau_r$ respectively, such that

$$
\begin{array}{ll}
  f_l(x) \geq \frac{\alpha_{mid}}{c_r - c_l} + \delta_l \cdot (c_l - x),& \text{ if } x \in [c_l - \tau_l, c_l)\\
  f_r(x) \geq \frac{\alpha_{mid}}{c_r - c_l} + \delta_r \cdot (x - c_r),& \text{ if } x \in (c_r, c_r + \tau_r].\\
\end{array}
$$
For the sake of convenience, we consider a stronger condition that drop off the factors $c_l - x$ and $x - c_r$, which is called \textbf{Assumption \ref{asmpt:density_gap}}. It indicates that the density jumps at the transition points $c_l$ and $c_r$. 
\begin{assumption}
  There are two positive parameters $\delta_l$ and $\delta_r$ such that
  $$
  \begin{array}{l}
    f_l(c_l^-) \geq \frac{\alpha_{mid}}{c_r - c_l} + \delta_l\\
    f_r(c_r^+) \geq \frac{\alpha_{mid}}{c_r - c_l} + \delta_r.
  \end{array}
  $$
  \label{asmpt:density_gap}
\end{assumption}

Together, we refer to the Binomial mixture with Assumption \ref{asmpt:U_shape} and Assumption \ref{asmpt:density_gap} as the \textbf{\GFCModel\/} model.


\subsection{Method}\label{sec:bmu_method_npmle}
Let $c_l^{(0)}$ and $c_r^{(0)}$ be the underlying ground truth of the two cutoffs in \GFCModel\/. Our goal is to identify the cutoff that separates the relevant genes and the irrelevant genes in the GeneFishing method. Specifically, we want to find an estimator $\hat{c}_r$ for $c_r^{(0)}$ and study the behavior of $\Pbb[\hat{c}_r \geq c_r^{(0)}]$.

Define $\alpha_l(x; v) = \Pbb[v \leq x]$, $\alpha_{mid}(x, y; v) = \Pbb[x < v < y]$, $\alpha_r(y; v) = \Pbb[v \geq y]$, where $v$ can be $s$ or $\hat{s}$. Define $N_l(x; \{v_i\}_{i=1}^n) := \#\{v_i \leq x, i = 1,\ldots, n\}$, $N_{mid}(x, y; \{v_i\}_{i=1}^n) := \#\{x < v_i \leq y, i = 1,\ldots, n\}$, $N_r(y; \{v_i\}_{i=1}^n) := \#\{v_i > y, i = 1,\ldots, n\}$, $N(x; \{v_i\}_{i = 1}^n) = \#\{v_i = x, i = 1, \ldots, n\}$, where $\{v_i\}_{i=1}^n$ can be $\{s_i\}_{i=1}^n$ or $\{\hat{s}_i\}_{i=1}^n$. 

Since we are working on $\hat{s}_i$'s rather than on $s_i$'s, we denote $\alpha_l(x) = \Pbb[\hat{s} \leq x]$ and $N_l(x) = N_l(x; \{\hat{s}_i\}_{i=1}^n)$ for simplicity. Similarly, we can get simplified notation $\alpha_{mid}(x, y)$, $\alpha_r(y)$, $N_{mid}(x, y)$, $N_r(y)$, $N(x)$. In the rest of the paper, we sometimes use $\alpha_l$, $\alpha_r$, $\alpha_{mid}$ for $\alpha_l(c_l)$, $\alpha_r(c_r)$ and $\alpha_{mid}(c_l, c_r)$ respectively if no confusion arises. 

\subsubsection{The Non-parametric Maximum Likelihood Estimation}\label{sec:bmu_bmgU_npmle}
To estimate the parameters in \GFCModel\/, we consider the non-parametric maximum likelihood estimation (NPMLE). We first solve the problem given $c_l$ and $c_r$, then searching for optimal $c_l$ and $c_r$ using grid searching. The NPMLE problem is:
\begin{eqnarray}
  H_{full}(c_l, c_r) := \max && \sum_{\hat{s}_i \leq c_l} \log g_l(\hat{s}_i) + \sum_{\hat{s}_i > c_r} \log g_r(\hat{s}_i)   \label{opt:gU_mle_known_c0_c1} \\
  && +  N_l(c_l)\log \alpha_l + N_{mid}(c_l, c_r) \log \frac{\alpha_{mid}}{c_r - c_l} + N_r(c_r) \log \alpha_r \nonumber\\
    s.t. && \int_{0}^{c_l} g_l = 1, \int_{c_r}^1 g_r = 1, g_l \text{ decreasing }, g_r \text{ increasing}\nonumber \\
         && \alpha_l, \alpha_r, \alpha_{mid} > 0, \alpha_l + \alpha_r + \alpha_{mid} = 1\nonumber \\
         && \left .
           \begin{array}{l}
         \alpha_l g_l(c_l^-) \geq \frac{\alpha_{mid}}{c_r - c_l} + d_l\\
         \alpha_r g_r(c_r^+) \geq \frac{\alpha_{mid}}{c_r - c_l} + d_r
           \end{array}
          \right \}\label{ineq:changepoint_gap}
\end{eqnarray}
Here $d_l$ and $d_r$ are two parameters to tune, and we call the inequalities \eqref{ineq:changepoint_gap} the \textbf{change-point-gap} constraint, which correspond to Assumption \ref{asmpt:U_shape} and Assumption \ref{asmpt:density_gap}. Given $c_l$ and $c_r$, the variables to optimize over are
$$S := \{\alpha_l, \alpha_r, g_l(\hat{s}_1), \ldots, g_l(\hat{s}_{i_l}), g_r(\hat{s}_{i_r}), \ldots, g_r (\hat{s}_n)\},$$
where $i_l := \max_i \cdot \Ibb (\hat{s}_i < c_l)$, $i_r := \min_i \cdot \Ibb (\hat{s}_i \geq c_r)$. Since $\log x$ is continuous and concave w.r.t $x$, and the feasible set is convex (it is easy to check that $\{(x, y, z): xy \geq z; x, y, z \geq 0\}$ is a convex set), the problem \eqref{opt:gU_mle_known_c0_c1} is a convex optimization with a unique optimizer. 

There are mainly two difficulties for the optimization problem \eqref{opt:gU_mle_known_c0_c1}. First, the change-point-gap constraint \eqref{ineq:changepoint_gap} complicates the monotone density estimation. Second, it is not easy to optimize over $\alpha_l$, $\alpha_r$ and $g_l$, $g_r$ simultaneously. 

\subsubsection{\GFCMethod\/: A Simplified Estimator}\label{sec:bmu_gU_simplified_est}
Fortunately, 
the following observation suggests a simplified optimization problem.
\begin{itemize}
\item Note that $\alpha_l$ and $\alpha_r$ are the population masses for $x \leq c_l$ and $x \geq c_r$, which can be well estimated by the empirical masses $\hat{\alpha}_l(c_l) = N_l(c_l)/n$, $\hat{\alpha}_{mid}(c_l, c_r) = N_{mid}(c_l, c_r)/n$ and $\hat{\alpha}_r(c_r) = N_r(c_r)/n$.
\item If \GFCModel\ is true with $\delta_l \geq d_l$ and $\delta_r \geq d_r$, and the solution to the optimization \eqref{opt:gU_mle_known_c0_c1} without \eqref{ineq:changepoint_gap} at $c_l = c_l^{(0)}$, $c_r = c_r^{(0)}$ is good enough, then the change-point-gap constraint \eqref{ineq:changepoint_gap} is satisfied with high probability.
\item From Figure \ref{fig:hist_CFR}, we can see that the flat region is wide. We can easily pick an interior point within the flat region.
\end{itemize}
  Inspired by these observations, we replace the population masses with the empirical masses and drop off the change-point-gap constraint. We obtain the simplified objective function as follows:
\begin{eqnarray}
  H_{simplified}(c_l, c_r) := \max &&\sum_{\hat{s}_i \leq c_l} \log g_l(\hat{s}_i) + \sum_{\hat{s}_i > c_r} \log g_r(\hat{s}_i)  \label{opt:gU_mle_known_c0_c1_emp_mass} \\
  && + N_l(c_l)\log \hat{\alpha}_l(c_l) + N_{mid}(c_l, c_r) \log \frac{\hat{\alpha}_{mid}(c_l, c_r)}{c_r - c_l} + N_r(c_r) \log \hat{\alpha}_r(c_r)\nonumber\\
    s.t. && \int_{0}^{c_l} g_l = 1, \int_{c_r}^1 g_r = 1, g_l \text{ decreasing }, g_r \text{ increasing}\nonumber 
\end{eqnarray}
where 
\begin{align*}
  &\hat{\alpha}_l(c_l) = N_l(c_l)/n = \frac{\# \{i | \hat{s}_i \leq c_l\}}{n}\\
  &\hat{\alpha}_{mid}(c_l, c_r) = N_{mid}(c_l, c_r)/n = \frac{\# \{i | c_l < \hat{s}_i \leq c_r\}}{n}\\
  &\hat{\alpha}_{1}(c_r) = N_r(c_r)/n = \frac{\# \{i | \hat{s}_i > c_r\}}{n}.
\end{align*}

The problem \eqref{opt:gU_mle_known_c0_c1_emp_mass} is reduced to two monotone density estimations, which the Grenander estimator can solve. As we point out in the above observations, we can easily identify an interior point $\mu$ in the flat region. We fit an Grenander estimator for the decreasing $g_l$ on $[0, \mu]$ and an Grenander estimator for the increasing $g_r$ on $(\mu, 1]$. There are three advantages of using the interior point $\mu$. First, it significantly reduces the computational cost by estimating the two Grenander estimators just once, regardless of the choices of $c_l$ and $c_r$. Second, it bypasses the boundary issue of the Grenander estimators since we are mainly concerned with the behaviors of the estimators at the points $c_l < \mu$ and $c_r > \mu$. Moreover, the usage of $\mu$ disentangles the mutual influences of the left decreasing part and the right increasing part; thus, it makes the analysis of the estimators simple. 

Once we fit the Grenander estimator, we check whether the change-point-gap constraint \eqref{ineq:changepoint_gap} holds for different pairs of $c_l$ and $c_r$. Finally, we pick the feasible pair with the maximal likelihood. We call this algorithm \textbf{\GFCMethod\/} (U-shape cutoff), which is summarized in Algorithm \ref{algo:gU_grid_search}.

\begin{algorithm}[ht]
  \caption{\GFCMethod\/: estimation of the \GFCModel\ model by grid-searching the optimal cutoff pair.}
  \label{algo:gU_grid_search}
  \begin{algorithmic}[1]
    \REQUIRE Data: $\{\hat{s}_1, \ldots, \hat{s}_n\}$;\\ ~~~~~~~~~The density gaps: $d_l$, $d_r$;\\ ~~~~~~~~~The interior point $\mu$ of the flat region;\\ ~~~~~~~~~The searching interval: $[0, c_l^{(\max)}]$, and $(c_r^{(\min)}, 1]$, where $c_l^{(\max)} <\mu$ and $c_r^{(\min)} > \mu$;\\ ~~~~~~~~~The unit for grid searching: $\gamma$.
    \STATE Initiate $c_l^{*} = \text{NULL}$, $c_r^{*} = \text{NULL}$; $\ell(c_l^*, c_r^*) = -\infty$.
    \STATE Estimate $\hat{\alpha}_l(\mu) = N_l(\mu)/n$.     
    \STATE Fit the Grenander estimator on $[0, \mu]$ to get $\tilde{g}_l$ and on $(\mu, 1]$ to get $\tilde{g}_r$.
    \FOR{$c_l \in \{0, \gamma, 2\gamma, \ldots, c_l^{(\max)}\}$}
    \FOR{$c_r \in \{c_r^{(\min)}, c_r^{(\min)} + \gamma, c_r^{(\min)} + 2\gamma, \ldots, 1\}$}
    \STATE Estimate $\hat{\alpha}_{mid}(c_l, \mu)$, and $\hat{\alpha}_{mid}(\mu, c_r)$.
    \STATE Let $\tilde{d}_l = \frac{\hat{\alpha}_{mid}(c_l, \mu)}{\hat{\alpha}_l(\mu)\cdot (\mu - c_l)} + \frac{d_l}{\hat{\alpha}_l(\mu)}$, $\tilde{d}_r = \frac{\hat{\alpha}_{mid}(\mu, c_r)}{(1 - \hat{\alpha}_l(\mu))\cdot (c_r - \mu)} + \frac{d_r}{1-\hat{\alpha}_l(\mu)}$.
    \STATE Let $\ell(c_l, c_r)$ be the corresponding $H_{simplified}(c_l, c_r)$ defined in the problem \eqref{opt:gU_mle_known_c0_c1_emp_mass}.
    \STATE Let $flag = \Ibb [\tilde{g}_l(c_l) \geq \tilde{d}_l$ \text{ and } $\tilde{g}_r(c_r) \geq \tilde{d}_r]$.
    \IF{$flag$ and $\ell(c_l, c_r) > \ell(c_l^*, c_r^*)$}    
    \STATE $(c_l^*, c_r^*) \leftarrow (c_l, c_r)$.
    \ENDIF
    \ENDFOR
    \ENDFOR
    \IF{$\ell(c_l^*, c_r^*) > -\infty$}
    \STATE \textbf{Return:} $c_l^*$, $c_r^*$, $\tilde{g}_l, \tilde{g}_r$, $\hat{\alpha}_l(\mu)$, $\ell(c_l^*, c_r^*)$. 
    \ELSE
    \STATE \textbf{Return:} $False$.
    \ENDIF
  \end{algorithmic}
\end{algorithm}

\subsection{Analysis}\label{sec:bmu_analysis_simplified_est}

For Algorithm \ref{algo:gU_grid_search}, the question arises whether $(c_l^{(0)}, c_r^{(0)})$ is a feasible pair for the change-point-gap constraint \eqref{ineq:changepoint_gap}. Theorem \ref{thm:true_c0_c1_valid} answers this question by claiming that there exist $c_l$ in a small neighborhood of $c_l^{(0)}$ and $c_r$ in a small neighborhood of $c_r^{(0)}$ such that $\tilde{g}_l(c_l) > \tilde{d}_l$ and $\tilde{g}_r(c_r) > \tilde{d}_r$ for appropriate choices of $d_l$ and $d_r$. This implies that we can safely set aside the constraint \eqref{ineq:changepoint_gap} when solving the problem \eqref{opt:gU_mle_known_c0_c1_emp_mass}. The proof of Theorem \ref{thm:true_c0_c1_valid} is deferred to Appendix \ref{appendix:bmu_proof_true_c0_c1_valid}.

\begin{theorem}[Feasibility of Gap Constraint for $(c_l^{(0)}, c_r^{(0)})$]\label{thm:true_c0_c1_valid}
  
  Suppose $f$ is a distribution satisfying Assumption \ref{asmpt:U_shape} and Assumption \ref{asmpt:density_gap}, with $\alpha_{mid} (c_l^{(0)}, c_r^{(0)}) > 0$, $f_{\max} < \infty$. If $m \geq C_1 \cdot \left (\max\{N_l(c_l^{(0)}), N_r(c_r^{(0)})\}\right)^{\frac{2}{3}}$, and $d_l < \delta_l$, $d_r < \delta_r$, there exist $c_l$ and $c_r$ such that
  $$c_l \leq c_l^{(0)} \text{ with } |c_l - c_l^{(0)}| \leq C_2 \cdot N_l(c_l^{(0)})^{-\frac{1}{3}}$$
  and
  $$c_r \geq c_r^{(0)} \text{ with } |c_r - c_r^{(0)}| \leq C_3 \cdot N_r(c_r^{(0)})^{-\frac{1}{3}}$$
  such that $\tilde{g}_l$, $\tilde{g}_r$, $\tilde{d}_l$ and $\tilde{d}_r$ produced by Algorithm \ref{algo:gU_grid_search} satisfy $\tilde{g}_l(c_l) > \tilde{d}_l$ and $\tilde{g}_r(c_r) > \tilde{d}_r$, provided the input $\mu \in (c_l^{(0)}, c_r^{(0)})$. Furthermore, the resulting density estimator $\ftil_{m,n}$ satisfies
  $$\int_0^1 |\ftil_{m,n} - f(x)|dx \leq  C_4 \cdot \left \{ N_l(c_l^{(0)})^{-\frac{1}{3}} + N_r(c_r^{(0)})^{-\frac{1}{3}} \right \}.$$
  Here $C_1, C_2, C_3, C_4$ are positive constants that only depend on $f$.
\end{theorem}

Besides knowing there are some feasible points near $c_l^{(0)}$ and $c_r^{(0)}$, we want to have a clear sense of the optimal cutoff pair produced by Algorithm \ref{algo:gU_grid_search}. Theorem \ref{thm:feasible_c0_c1_CI} says that the optimal cutoff for the left (right) part is smaller (larger) than $c_l^{(0)}$ ($c_r^{(0)}$) with high probability.

\begin{theorem}[Tail Bounds of Identified Cutoffs]\label{thm:feasible_c0_c1_CI}
  
  Suppose $(\hat{c}_l, \hat{c}_r)$ is the identified optimal cutoff pair produced by Algorithm \ref{algo:gU_grid_search}, provided an input $\mu \in (c_l^{(0)}, c_r^{(0)})$. Under the same assumptions as Theorem \ref{thm:true_c0_c1_valid}, particularly $n\rightarrow \infty$,  $m/\max\{N_l(c_l^{(0)}), N_r(c_r^{(0)})\} \rightarrow \infty$,
  $$\Pbb[\hat{c}_l > c_l^{(0)}] \leq \Pbb[\hat{S}_{c_l^{(0)}, \mu}(c_l^{(0)}) \geq \sqrt{N_l(c_l^{(0)})}  \cdot \frac{d_l}{\alpha_l(\mu)} - C_1],$$
  and
  $$\Pbb[\hat{c}_r < c_r^{(0)}] \leq \Pbb[\hat{S}_{\mu, c_r^{(0)}}(c_r^{(0)}) \geq \sqrt{N_r(c_r^{(0)})}  \cdot \frac{d_r}{1 - \alpha_l(\mu)} - C_2] ,$$
 where $C_1$, $C_2$ are positive constants, and $C_1$ only depends on $\alpha_l(\mu)$, $d_l$, $C_2$ only depends on $\alpha_r(\mu)$, $d_r$; $\hat{S}_{a, b}(t)$ is the slope at $F(t)$ of the least concave majorant in $[F(a), F(b)]$ of a standard Brownian Bridge in $[0, 1]$. 
\end{theorem}

The proof of Theorem \ref{thm:feasible_c0_c1_CI} is in fact reduced to proving any cutoff pair $(c_l, c_r)$ with $c_l > c_l^{(0)}$ or $c_r < c_r^{(0)}$ does not satisfy the change-point-constraint with high probability. Since $\tilde{g}_l$ and $\tilde{g}_r$ estimated in Algorithm \ref{algo:gU_grid_search} are decreasing and increasing respectively, if $c_l > c_l^{(0)}$ (or $c_r < c_r^{(0)}$) violates the constraint, then $c_l'$ (or $c_r'$) will violate it with high probability if $c_l' > c_l$ (or $c_r' < c_r$). So it is reduced to considering the smallest $c_l > c_l^{(0)}$ and the largest $c_r < c_r^{(0)}$ in the grid searching space of Algorithm \ref{algo:gU_grid_search}. Then the result can be concluded using Theorem \ref{thm:noisy_grenander_local_asymp}. The detail is deferred to Appendix \ref{appendix:bmu_proof_feasible_c0_c1_CI}.

Finally, we show in Theorem \ref{thm:c0_c1_converge_finite_sample} that the identified $\hat{c}_r$ converges to $c_r^{(0)}$ at the rate of $\Ocal ([N_r(c_r^{(0)})]^{-\frac{1}{3}})$ if $m = \Omega(n^{\frac{2}{3}})$. And the estimated density also converges to the true one at the rate of $\Ocal (\max \{[N_r(c_r^{(0)})]^{-\frac{1}{3}}, [N_l(c_l^{(0)})]^{-\frac{1}{3}}\})$. The proof can be found in Appendix \ref{appendix:bmu_proof_c0_c1_converge_finite_sample}.

\begin{theorem}[$L_1$ Convergence of Identified Cutoff]\label{thm:c0_c1_converge_finite_sample}
  
  Suppose $f$ is a distribution satisfying Assumption \ref{asmpt:U_shape} and Assumption \ref{asmpt:density_gap}, with $\alpha_{mid} (c_l^{(0)}, c_r^{(0)}) > 0$, $f_{\max} < \infty$. Let $\Delta_l = c_l - c_l^{(0)}$, $\Delta_r = c_r - c_r^{(0)}$. If we have
  $$m \geq C_1 \cdot \left (\max\{N_l(c_l^{(0)}), N_r(c_r^{(0)})\} \right )^{\frac{2}{3}},$$
  then
  $$|\widehat{\Delta}_l| \leq C_2 \cdot N_l(c_l^{(0)})^{-\frac{1}{3}},~|\widehat{\Delta}_r| \leq C_3 N_r(c_r^{(0)})^{-\frac{1}{3}},$$
  where $\widehat{\Delta}_l$ and $\widehat{\Delta}_r$ are associated with $\hat{c}_l$ and $\hat{c}_r$ output by Algorithm \ref{algo:gU_grid_search}. Furthermore, the resulting $\tilde{f}_{n,m}$ satisfies
  $$\Ebb_f \int_0^1 |\tilde{f}_{n,m}(x) - f(x)|dx \leq C_4 \cdot \left \{ N_l(c_l^{(0)})^{-\frac{1}{3}} + N_r(c_r^{(0)})^{-\frac{1}{3}} \right \}.$$
  Here $C_1, C_2, C_3, C_4$ are positive constants that do not depend on $n$, $N_l(c_l^{(0)})$ and $N_r(c_r^{(0)})$.
\end{theorem}


\section{Experiments of Ucut}\label{sec:bmu_exp}
\subsection{Numerical Experiments}\label{sec:bmu_numerical_exp}

\subsubsection{Data generating process}\label{sec:bmu_sim_data_generation}
To confirm and complement our theory, we use extensive numerical experiments to examine the finite performance of \GFCMethod\/ on the estimation of $c_r$. We study a BMU model that is comprised of linear components. Specifically, the model consists of three parts with boundaries $c_l$ and $c_r$. The middle part is a flat region of height $\delta_m$. The left part is a segment with the right end at $(c_l, \delta_m + \delta_l)$ and slope $s_l < 0$ while the right part is a segment with the left end at $(c_r, \delta_m + \delta_r)$ and slope $s_r > 0$; see Figure \ref{fig:lin_exp} (a) for illustration. We call it the \textbf{linear valley} model. We normalize the linear valley model to produce the density of interest. We call the normalized gaps $\tilde{\delta}_l$ and $\tilde{\delta}_r$.

\begin{figure}[ht]
  \begin{minipage}{0.49\textwidth}
    \includegraphics[width=\textwidth]{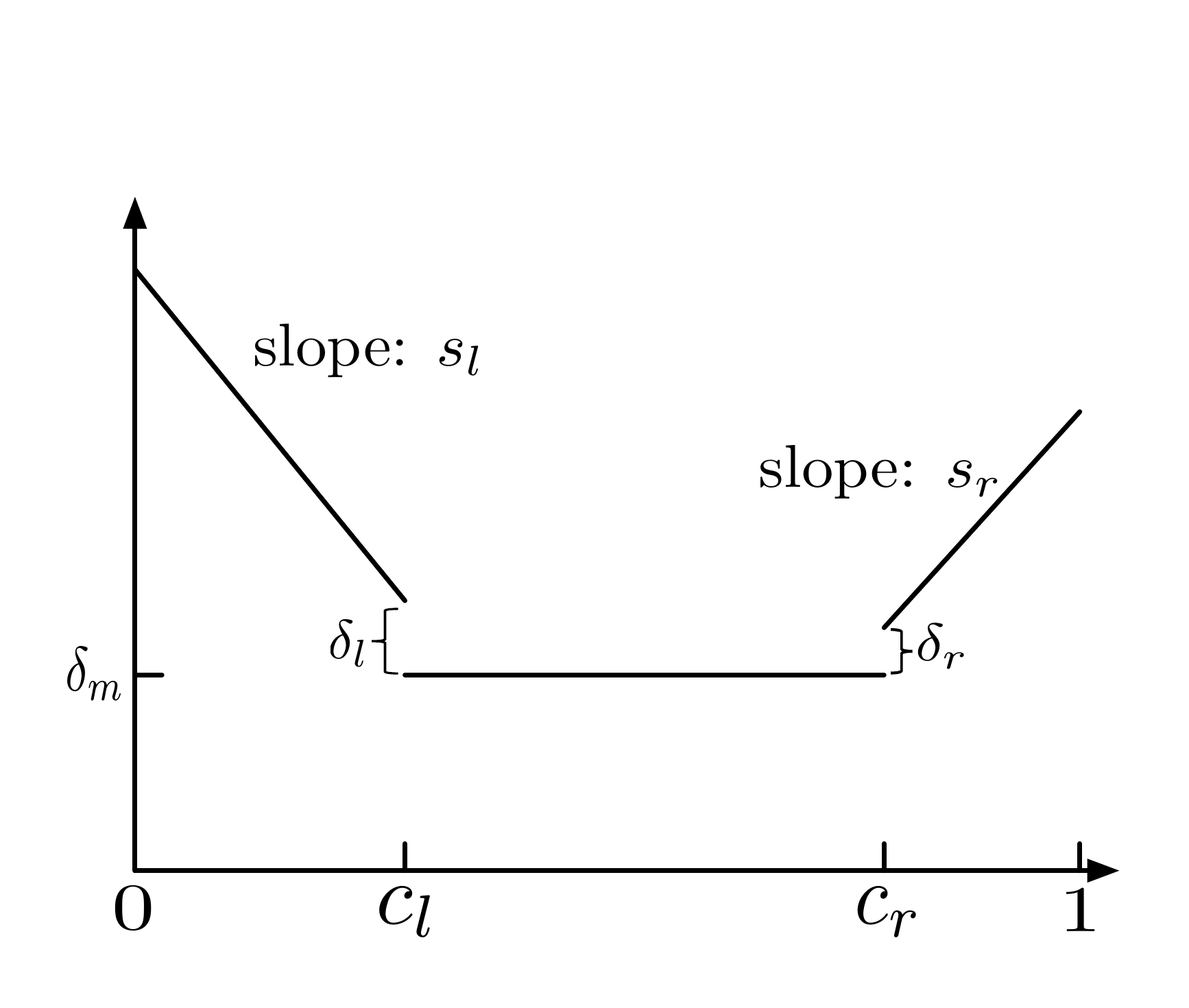}
    \subcaption{original model}
  \end{minipage}
  \begin{minipage}{0.49\textwidth}
    \includegraphics[width=\textwidth]{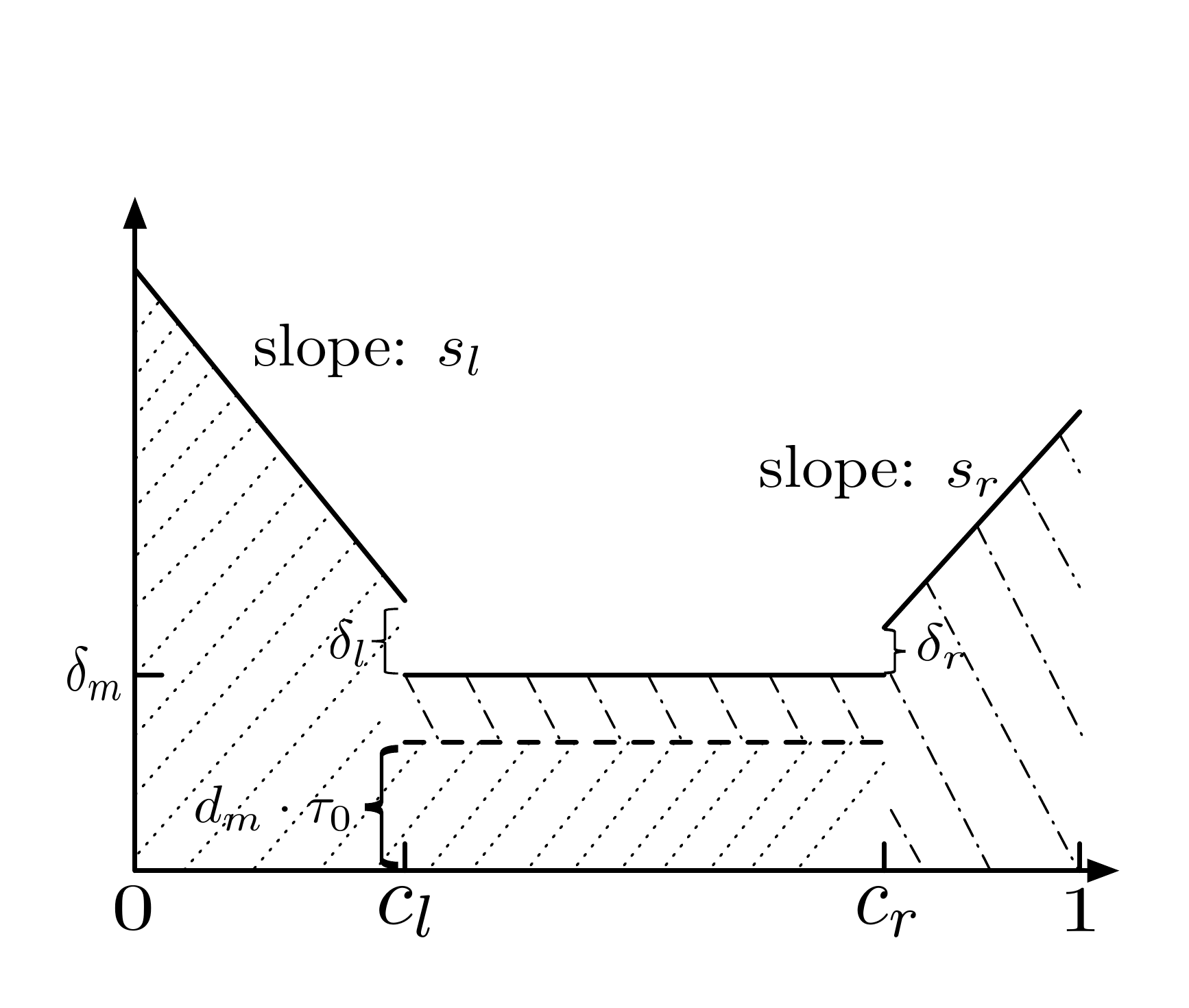}
    \subcaption{two-group model}
  \end{minipage}
  \caption{The linear valley model.}\label{fig:lin_exp}
\end{figure}

The linear valley model depicts the mixture density of the null distribution and the alternative distribution. We only assume that the left part ahead of $c_l$ purely belongs to the null distribution while the right part ahead of $c_r$ purely belongs to the alternative distribution. The null and the alternative distributions can hardly be distinguished in the flat middle part. To understand the linear valley model from the perspective of the two-group model, we assume that $\tau_0 \times 100\%$ of the middle part belongs to the null distribution while the remaining belongs to the alternative model. In Figure \ref{fig:lin_exp} (b), the part in left slash corresponds to the null density $f_0$ while the part in right slash corresponds to the alternative density $f_1$. Let $\pi_0$ be the area in the left slash divided by the total area. Then the marginal density can be written as $f = \pi_0 f_0 + (1 - \pi_0) f_1$. Since any $\tau_0 \in [0, 1]$ gives the same $f$, the middle part is not identifiable. It is necessary to estimate and infer the right cutoff $c_r$, so that we can safely claim all the samples beyond $c_r$ are from the alternative distribution.

By default, we set $c_l = 0.3$, $c_r = 0.9$, $\delta_m = 1$, $\delta_l = 0.5$, $\delta_r = 0.5$, $s_l = -3$, $s_r = 1$. We sample $n = 10,000$ samples $\{s_1, \ldots, s_n\}$ from the linear valley model. Then for each $i \in \{1, \ldots, n\}$, we get the observations $\hat{s}_i \sim Binom(m, s_i)$ independently, where $m = 1,000$ if it is not specified particularly. The value of $\tau_0$ does not affect the data generation but it affects the FDR and the power of any method that yields discoveries.

Finally, the left part and the right part are not necessary to be linear. To investigate the effect of general monotone cases and misspecified cases (e.g., unimodal densities), we replace the left part and the right part with other functions; see Appendix \ref{sec:bmu_more_simulation}. 
All codes to replicate the results in this paper can be found at github.com/Elric2718/GFcutoff/.


\subsubsection{Robustness to model parameters}
When using Algorithm \ref{algo:gU_grid_search}, we use the middle point $\mu = 0.5$, the left gap parameter $d_l = 0.8\cdot \tilde{\delta}_l$, the right gap parameter $d_r = 0.8 \cdot \tilde{\delta}_r$, the searching unit $\gamma = 0.001$. We first investigate how the binomial size $m$ affects the estimation of $c_r$. Using the default setup as described in Section \ref{sec:bmu_sim_data_generation}, we vary the binomial size $m \in \{10^2, 10^3, 2\times 10^3, 5\times 10^3, 10^4, Inf\}$, where $Inf$ refers to the case that there is no binomial randomness and we observe $s_i$'s directly. As shown in Figure \ref{fig:lin_convergence} (a), $\hat{c}_r$ converges to the true $c_r^{(0)}$ as $m$ grows. When $m = 10^3$, the estimated $c_r$ is as good as that of using $s_i$'s directly. This corroborates our theory that we only need $m \sim n^{\frac{2}{3}} \approx 500$ when the BMU model holds. Note that in the linear setup, even with $m = 10^2$, $\hat{c}_r$ is larger than true $c_r^{(0)}$ with large probability. It implies that \GFCMethod\/ is safe in the sense that it will make few false discoveries by using $\hat{c}_r$ as the cutoff.

In the sequel, we stick to $m = 10^3$ since it works well enough for the linear valley model. We investigate whether the width of the middle flat region affects the estimation of $c_r$. We consider $c_l = 0.5 - w/2$, $c_r = 0.5 + w/2$ with $w \in \{0.6, 0.4, 0.2, 0.1, 0.\}$ while other model parameters are set by default. In Figure \ref{fig:lin_convergence} (b), the estimation of $c_r$ is quite satisfying when the width is no smaller than $0.2$. When the width drops to $0.1$ or smaller, the estimation is not stable but still conservative in the sense that $\hat{c}_r > c_r$ in most cases.

Finally, we examine how the gap size influences the estimation of $c_r$. We take $\delta_l = \{0.5, 0.3, 0.2 0.1, 0.01\}$ and $\delta_r = \{0.5, 0.3, 0.2 0.1, 0.01\}$. Figure \ref{fig:lin_convergence} (c) shows that the estimation of $c_r$ is robust to the gap sizes as long as the input $d_l$ and $d_r$ are smaller than the true gaps. This gives us confidence in applying \GFCMethod\/ to identify the cutoff even when there is no gap, which is more realistic.

\begin{figure}
  \begin{minipage}{0.33\textwidth}
    \centering
    \includegraphics[width=\textwidth]{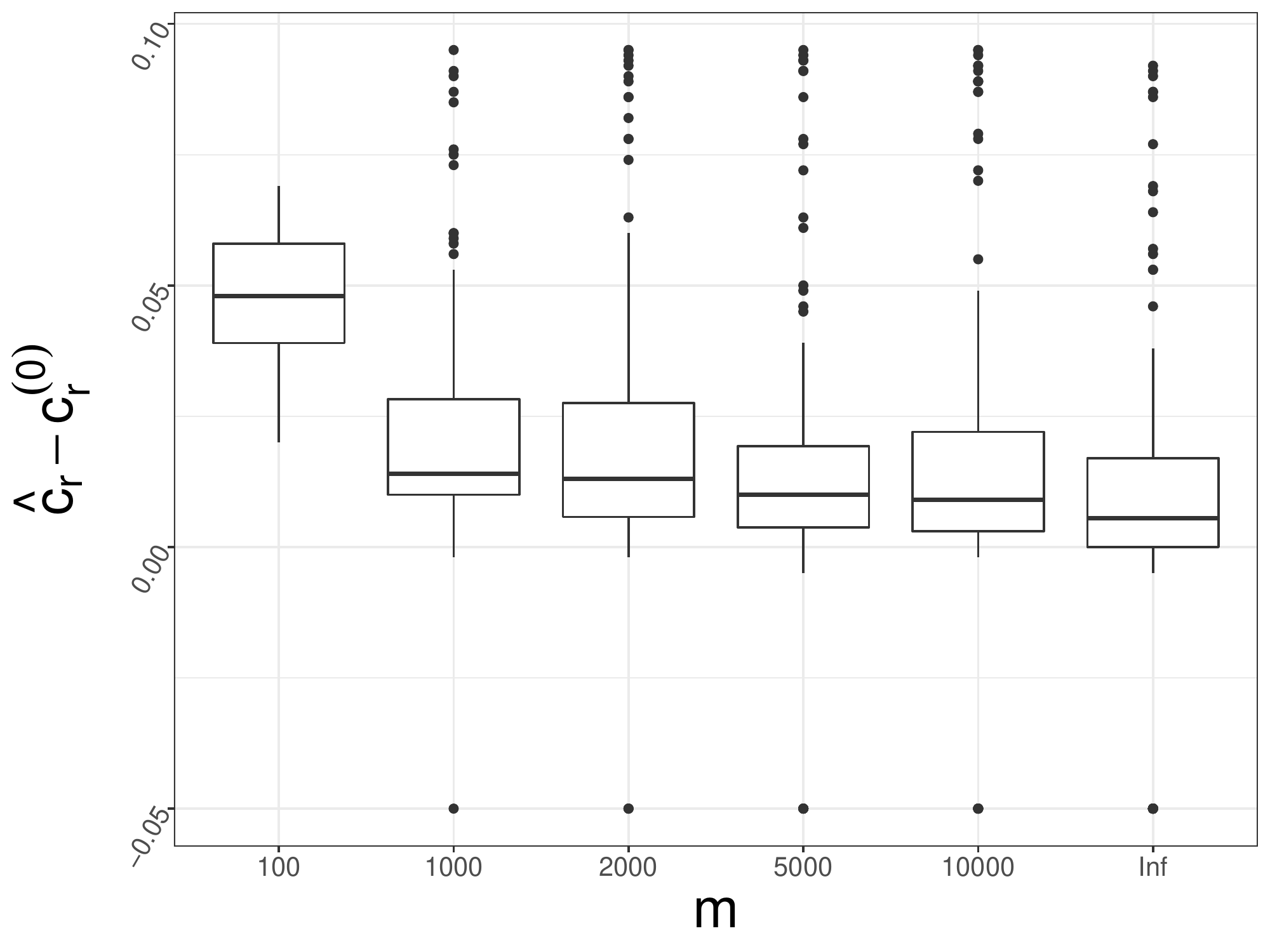}
    \subcaption{}
  \end{minipage}
  \begin{minipage}{0.33\textwidth}
    \centering
    \includegraphics[width=\textwidth]{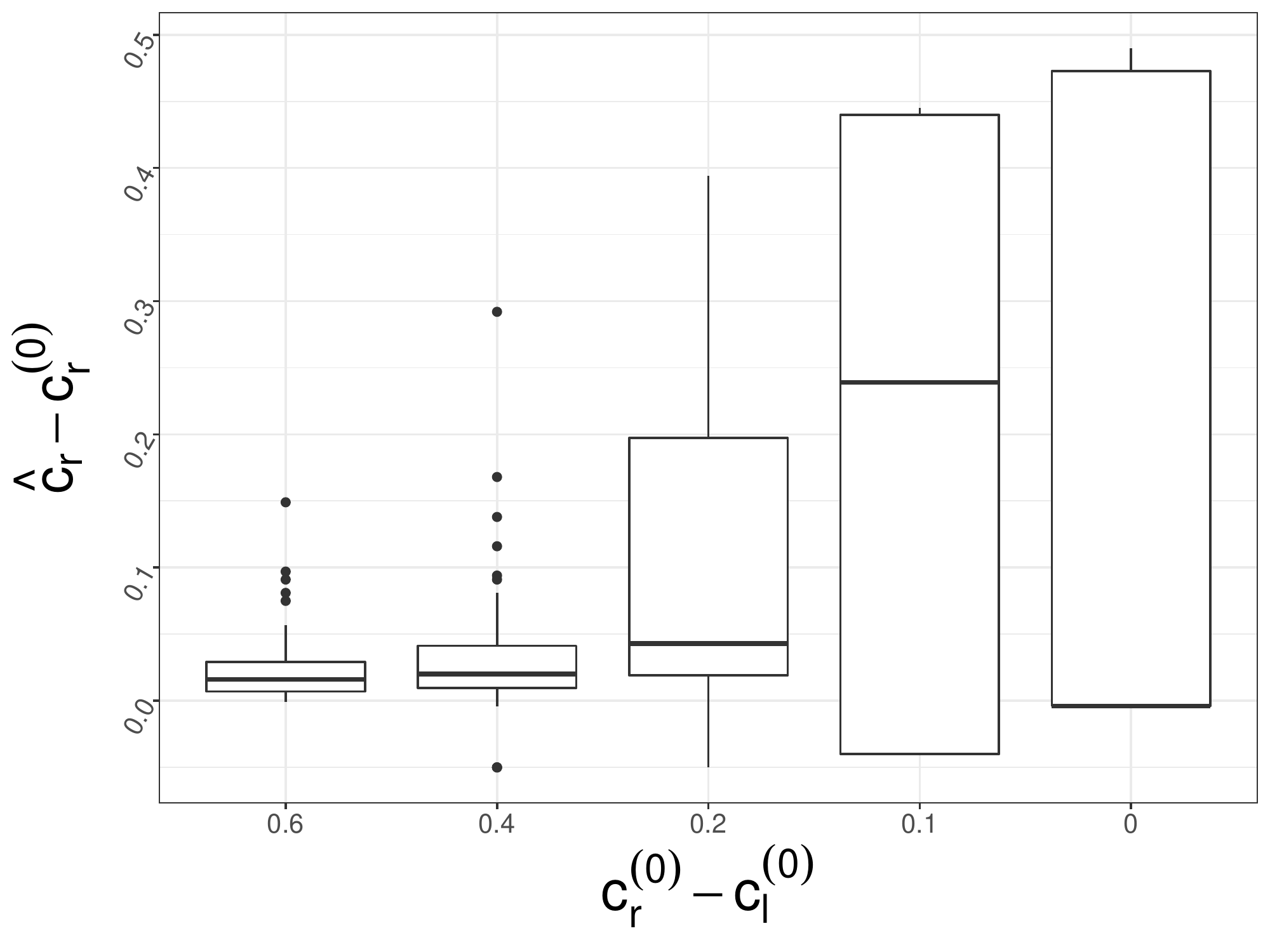}
    \subcaption{}
  \end{minipage}
    \begin{minipage}{0.33\textwidth}
      \centering
      \includegraphics[width=\textwidth]{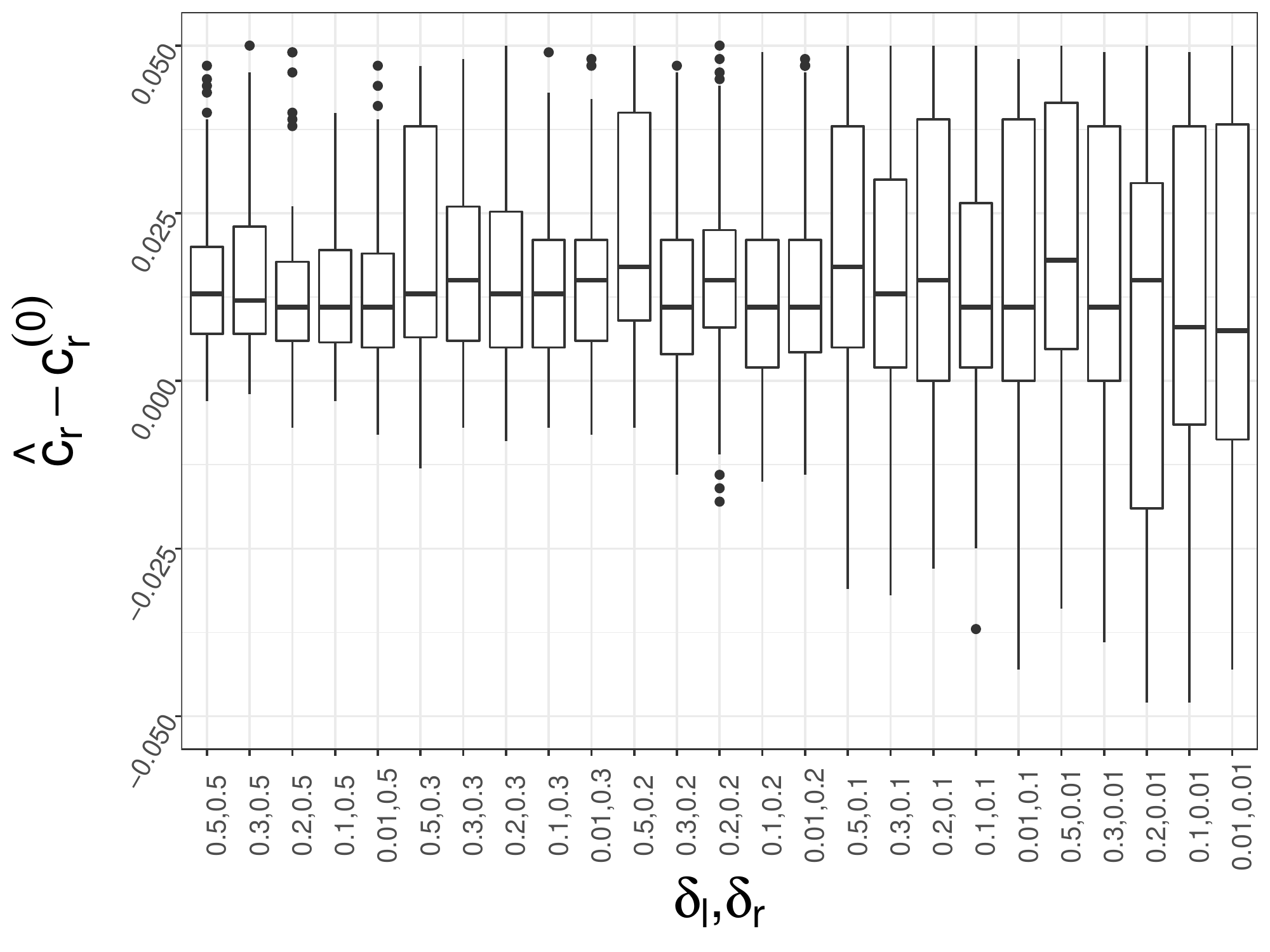}
    \subcaption{}
  \end{minipage}
  \caption{The estimation of $\hat{c}_r$ under the linear valley model. (a) with respect to $m$; (b) with respect to the width of the middle flat region; (c) with respect to the gap sizes.}\label{fig:lin_convergence}
\end{figure}

\subsubsection{Sensitivity of the algorithm hyper-parameters}
Algorithm \ref{algo:gU_grid_search} (Ucut) mainly have three tuning parameters: the middle point $\mu$, the left gap $d_l$ and the right gap $d_r$. For practical use, the three tuning parameters may be misspecified. For example, the middle point is not easy to spot, or the left gap and the right gap are too small. We use the default model parameters as specified in Section \ref{sec:bmu_sim_data_generation}. Let $d_l$ and $d_r$ be $0.8$ times the true normalized gaps $\tilde{\delta}_l$ and $\tilde{\delta}_r$ respectively. We vary the choice of the middle point $\mu$. Figure \ref{fig:lin_sensitivity} (a) shows that the estimation of $c_r$ is not sensitive to the choice of $\mu$ as long as it is picked within the flat region $[0.3, 0.9]$. If $\mu$ is picked left to the flat region, the $\hat{c}_r$ has a larger variance but it is more conservative in the sense that $\hat{c}_r > c_r^{(0)}$ in most cases. If $\mu$ is picked right to the flat region, the $\hat{c}_r$ tends to be $\min\{\mu, c_r^{(0)}\}$.

Next, we fix $\mu = 0.5$ but consider $d_l = \kappa \times \tilde{\delta}_l$ and $d_r = \kappa \times \tilde{\delta}_r$, where $\kappa \in \{1, 0.9, 0.8, 0.5, 0.2,$ $0.1, 0.01\}$. We do not consider $\kappa > 1$ because there might not exist feasible $(c_l, c_r)$ that satisfies the change-point-constraint. Figure \ref{fig:lin_sensitivity} (b) shows that when $\kappa$ is within $[0.5, 1]$ the estimation is satisfying. The estimated $c_r$ can be slightly smaller than the true $c_r^{(0)}$ when $\kappa < 0.5$ but in a tolerable range.

In a nutshell, the choices of $\mu$, $d_l$ and $d_r$ are crucial to Algorithm \ref{algo:gU_grid_search}. But the sensitivity analysis indicates that it is not necessary to be excessively cautious. In practice, picking these parameters by eyeballs can give a safe estimation in most cases.

\begin{figure}
  \begin{minipage}{0.49\textwidth}
    \centering
    \includegraphics[width=\textwidth]{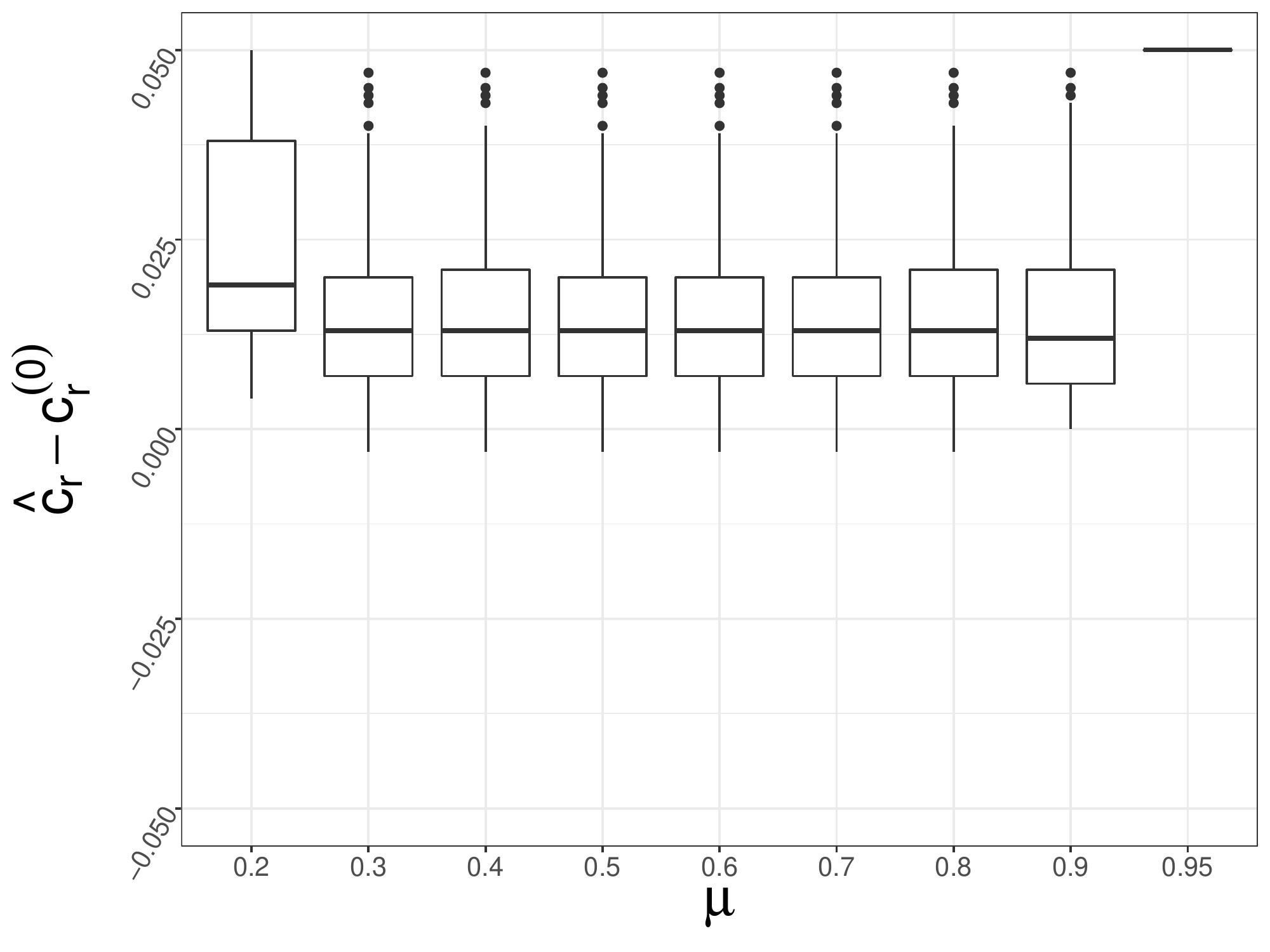}
    \subcaption{}
  \end{minipage}
  \begin{minipage}{0.49\textwidth}
    \centering
  \includegraphics[width=\textwidth]{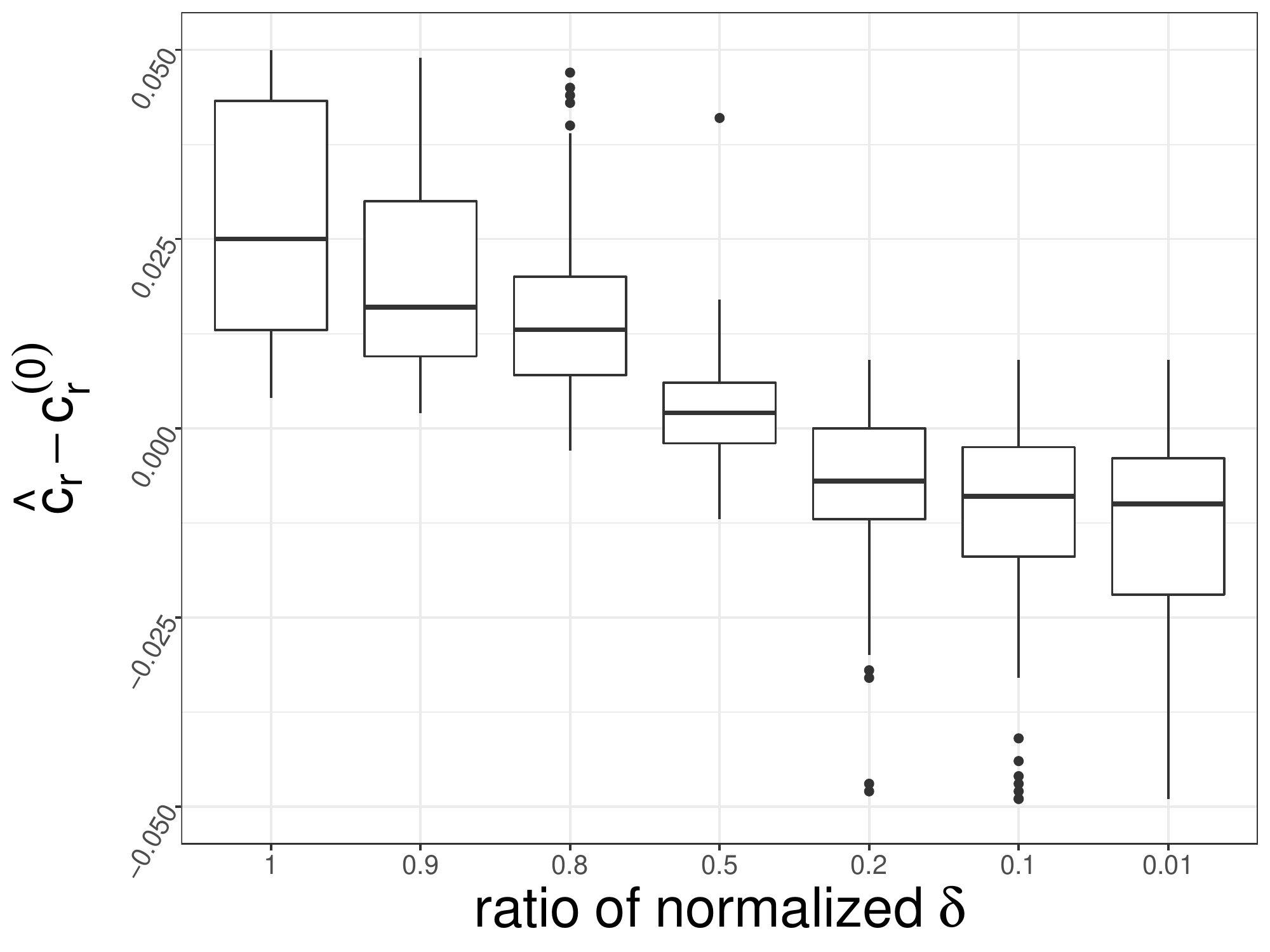}      
    \subcaption{}
  \end{minipage}
      \caption{The estimation of $\hat{c}_r$ under the linear valley model. (a) with respect to the choice of the middle point $\mu$; (b) with respect to the choice of the input $d_l$ and $d_r$. Here $d_l = \kappa \times \tilde{\delta}_l$ and $d_r = \kappa \times \tilde{\delta}_r$, where $\kappa$ is a ratio of the normalized $\delta$'s.}\label{fig:lin_sensitivity}
\end{figure}

\subsubsection{Comparison to other methods}
To examine the power and the capacity of controlling FDR of \GFCMethod\/, we consider the two-group model as specified in Figure \ref{fig:lin_exp} (b). The middle part is not identifiable, which means that the samples of the alternative distribution can not be distinguished from those of the null distribution. To reflect this point, we arbitrarily set the proportion of the null distribution in the middle part $\tau_0$. The goal is to identify the right cutoff $c_r$ but not $\tau_0$ because it is impossible to infer $\tau_0$.

We compare \GFCMethod\/ to the other four methods that are studied in \citet{gauran2018empirical}, i.e., ZIGP (Zero-inflated Generalized Poisson), ZIP (Zero Inflated Poisson), GP (Generalized Poisson) and P (Poisson). The data studied in this paper and the associated primary purpose are highly similar to ours. First, they study mutation counts for a specific protein domain, which have excess zeros and are assumed to follow a zero-inflated model. Moreover, the four mentioned methods are used to select significant mutation counts (i.e., cutoff) while controlling a given level of Type I error via False Discovery Rate procedures. Therefore, the four methods are good competitors for \GFCMethod\/ in view of the data distribution and the method usage.

Figure \ref{fig:lin_comp_rlt} shows that GP and P use rather small cutoffs and have too large FDRs. ZIGP and ZIP are over-conservative if the target FDR level is too low at $0.005$ or $0.01$, thus having quite low power. They perform better when the target FDR level is set to be $0.05$. On the other hand, \GFCMethod\/ can control FDR at $0.01$ if we directly use $\hat{c}_r$ as the cutoff. The associated power is better than those of ZIGP and ZIP. In order to loosen the FDR control and get higher power, it is fine to use a slightly smaller cutoff than $\hat{c}_r$. From this result, we confirm that \GFCMethod\/ is a better fit for the scenario where the middle part is not distinguishable.

\begin{figure}
  \centering
  \includegraphics[width=0.6\textwidth]{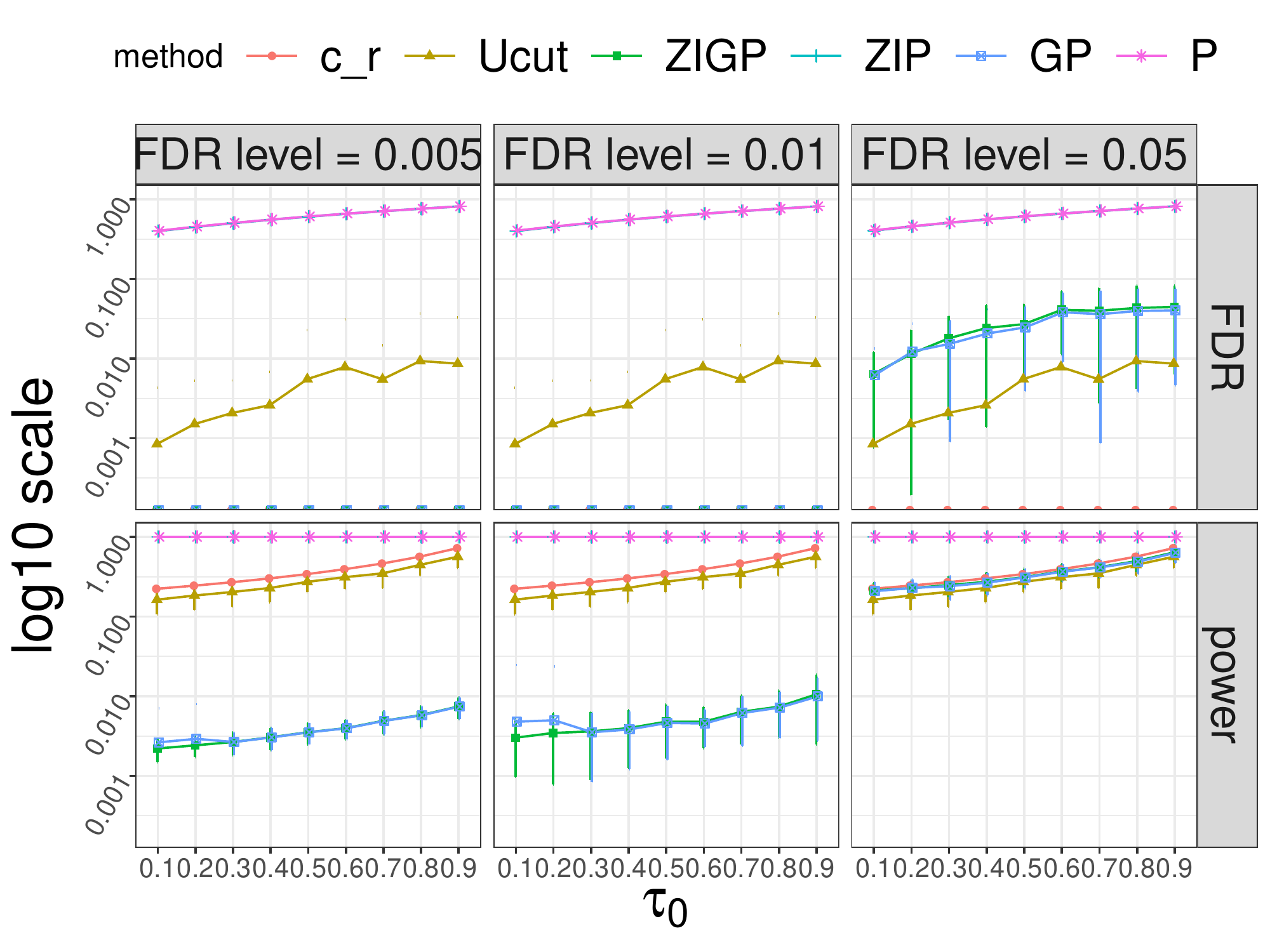}
  \caption{FDR and power of \GFCMethod\/ and other competing methods.}\label{fig:lin_comp_rlt}
\end{figure}

\subsection{Application to Real Data}\label{sec:bmu_real_exp}
We further demonstrate the performance of \GFCMethod\/ on the real datasets used in \citet{liu2019genefishing}. To be specific, we apply the GeneFishing method to four GTEx RNAseq datasets, liver, Artery Coronary, Transverse Colon, and Testis; see Table \ref{tbl:GTEx_profile} for details. We leverage the same set of $21$ bait genes used in \citet{liu2019genefishing}. The number of fishing rounds is set to be $m = 10,000$.

\begin{table}[ht]
  \caption{Details of GTEx RNAseq datasets.}\label{tbl:GTEx_profile}
  \centering
\begin{tabular}{c|c|c}
\hline
                 & \# samples & \# genes  \\\hline
Liver            &$119$ & $18,845$\\
Artery-Coronary  &$133$ &$20,597$ \\
Colon-Transverse & $196$ & $21,695$\\
Testis           & $172$ & $31,931$\\\hline
\end{tabular}
\end{table}

Once the CFRs are generated, we apply Algorithm \ref{algo:gU_grid_search} with the middle point $\mu = 0.5$, $d_l = 0.1$ and $d_r = 0.01$. We take $d_l$ to be ten times $d_r$ because there are much more zeros than ones in CFRs. As shown in Section \ref{sec:bmu_numerical_exp}, \GFCMethod\/ is not sensitive to the three hyper-parameters. The change of these parameters lays little influence on the results. Table \ref{tbl:GTEx_result} shows that for each tissue \GFCMethod\/ gives the estimator of $c_r$ that yields $50$ to $80$ discoveries. We estimate the false discovery rate using the second approach in \citet{liu2019genefishing}. Note that for Artery-Coronary, the estimated $c_r = 0.972$ by \GFCMethod\/ (which gives $\widehat{FDR} \approx 10^{-3}$) is less extreme than simply using $0.990$ (which gives $\widehat{FDR} \approx 10^{-4}$). It implicates that \GFCMethod\/ adapts to the tissue and can pick a cutoff with a reasonable false discovery rate.

\begin{table}[ht]
  \caption{Estimation of $c_r$ by Algorithm \ref{algo:gU_grid_search} on four tissues, where $\mu = 0.5$, $d_l = 0.1$ and $d_r = 0.01$. The second column is the estimated $c_r$ using bootstrap by sampling $70\%$ of the CFRs.}\label{tbl:GTEx_result}
  \centering
\begin{tabular}{c|c|c|c|c}
\hline
                 & $\hat{c}_r$  & bootstrapping estimation  & \# discovery (use $\hat{c}_r$) & $\widehat{FDR}$ \\ \hline
Liver            & $0.995$ & $0.993 (0.005)$ & $52$                         & $1.4 \times 10^{-3}$                                                \\ 
Artery-Coronary  &$0.972$ &$0.976 (0.009)$ & $85$                        & $5.7\times 10^{-3}$                                               \\ 
Colon-Transverse & $0.989$ & $0.991 (0.049)$ & $57$                         & $1.2 \times 10^{-4}$                                                \\ 
Testis           & $0.993$ & $0.992 (0.001) $ & $73$                        & $0.010$                                               \\ \hline
\end{tabular}
\end{table}

In addition, we also apply the GeneFishing method to a single-cell data of the pancreas cells from Tabula Muris\footnote{https://tabula-muris.ds.czbiohub.org}. It contains $849$ cells from mice and $5,220$ genes expressed in enough cells out of about $20,000$ genes. We find out $9$ bait genes based on the pancreas insulin secretion gene ontology (GO) term.

Unlike Figure \ref{fig:hist_CFR}, the CFRs of this data set do not appear in a U shape. Instead, we observe a unimodal pattern around zero and an increasing pattern around one (Figure \ref{fig:pancreas_hist}). Nonetheless, it does not hinder us from using \GFCMethod\/ to determine the cutoff, since we are mainly concerned about the right cutoff and Section \ref{sec:bmu_numerical_exp} demonstrates that \GFCMethod\/ is conservative even if the model is misspecified. By using $\mu = 0.5$, $d_l = 0.1$ and $d_r = 0.01$, \GFCMethod\/ gives $0.994$ as the estimation of the right cutoff, which discovers $77$ genes. By doing the GO enrichment analysis, we find out that these identified genes are enriched for the GO of response to ethanol with p-value $0.0021$, the GO of positive regulation of fatty acid biosynthesis with p-value $0.0055$, and the GO of eating behavior with p-value $0.0079$. These GOs have been shown to relate to insulin secretion in literature  \citep{huang2008ethanol,nolan2006fatty, tanaka2003eating}, which indicates the effectiveness of \GFCMethod\/.

\begin{figure}
  \centering
  \includegraphics[width = 0.5\textwidth]{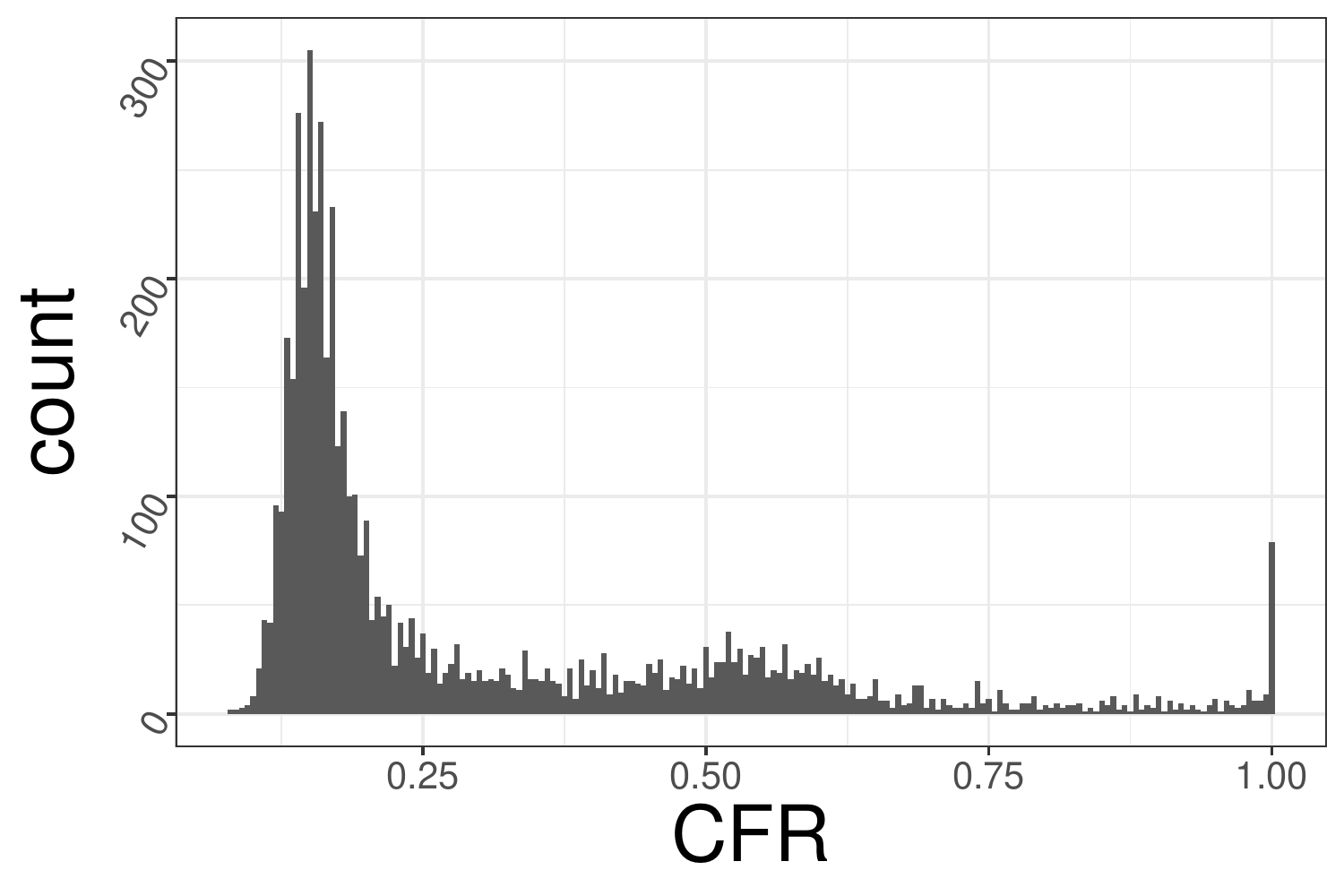}
  \caption{The CFRs of the single cell data of the pancreas tissue.}\label{fig:pancreas_hist}
\end{figure}

\section{Discussion}\label{sec:bmu_discussion}
In this work, we analyze the binomial mixture model \eqref{eq:binomial_mixture} under the U-shape constraint, which is motivated by the results of the GeneFishing method \citep{liu2019genefishing}. The contributions of this work are two-fold. First, to the best of our knowledge, this is the first paper to investigate the relationship between the binomial size $m$ and the sample size $n$ for the binomial mixture model under various conditions for $F$. Second, we provide a convenient tool for the downstream decision-making of the GeneFishing method.

Despite the identifiability issue of the binomial mixture model, we show that the estimator of the underlying distribution deviates from the true distribution, in some distance, at most a small quantity that depends on $m$. The implication is that to have the same convergence rate as if there is no binomial randomness, we need $m = \Omega(n^{\frac{2}{3}})$ for the empirical CDF and the Grenander estimator when the underlying density is respectively bounded and monotone. 
It is also of great interest to further investigate how the minimal $m$ hinges on the smoothness of the underlying distribution, e.g., studying the kernel density estimator and the smoothing spline estimator under the binomial mixture model.

To answer the motivating question of how large the CFR should be so that the associated sample can be regarded as a discovery in the GeneFishing method, we propose the BMU model to depict the underlying distribution and the NPMLE method \GFCMethod\/ to determine the cutoff. This estimator comprises two Grenander estimators, thus having a cubic convergence rate as the Grenander estimator when $m = \Omega(n^{\frac{2}{3}})$. We also show that the estimated cutoff is larger than the true cutoff with high probability. On the other hand, all these results are established for a bounded density. We leave it as future work to study the relationship between the convergence rate and the smoothness of the density.

Finally, we conclude the paper with a discussion on the empirical performance of \GFCMethod\/. The simulation studies reassure the practitioners that \GFCMethod\/ is not sensitive to the choice of its hyper-parameters --- the default setup can handle most cases. In the real data example used in \citet{liu2019genefishing}, \GFCMethod\/ gives a cutoff with a low estimated FDR using much less time than \citet{liu2019genefishing}. In the other example of the pancreas single-cell dataset, we use GeneFishing with \GFCMethod\/ and find out $77$ genes that are shown to be related to pancreas insulin secretion by GO enrichment analysis. Therefore, we are interested and confident in discovering exciting signals in biology or other fields if GeneFishing with \GFCMethod\/ can be applied to more real datasets.        



\printbibliography

 \clearpage
 \appendix

\begin{center}
  {\huge Appendix}
\end{center}
\input{appendix}

\end{document}

%% file: appendix.tex
\section{The Identifiability Issue of Binomial Mixture Model}\label{appendix:bmu_identifiability}

We review existing results for mixtures of distributions and the binomial mixture model in the literature. A general mixture model is defined as
\begin{equation}
  H(x) = \int_{\Omega} h(x|\theta) dG(\theta), \label{eq:general_mixture}
\end{equation}
where $h(\cdot | \theta)$ is a density function for all $\theta \in \Omega$, $h(x| \cdot)$ is a Borel measurable function for each $x$ and $G$ is a distribution function defined on $\Omega$. The family $h(x|\theta)$, $\theta \in \Omega$ is referred to as the kernel of the mixture and $G$ as the mixing distribution function. In order to devise statistical procedures for inferential purposes, an important requirement is the identifiability of the mixing distribution. Without this condition, it is not meaningful to estimate the distribution either non-parametrically or parametrically. The mixture $H$ defined by \eqref{eq:general_mixture} is said to be identifiable if there exists a unique $G$ yielding $H$, or equivalently,if the relationship
$$H(x) = \int_{\Omega} h(x| \theta) dG_1(\theta) = \int_{\Omega} h(x| \theta) dG_2(\theta)$$
implies $G_1(\theta) = G_2(\theta)$ for all $\theta \in \Omega$.

The identifiability problems concerning finite and countable mixtures (i.e. when the support of $F$ in \eqref{eq:general_mixture} is finite and countable respectively) have been investigated by \citet{teicher1963identifiability, patil1966identifiability,yakowitz1968identifiability,tallis1969identifiability,fraser1981identifiability,tallis1982identifiability,kent1983identifiability}. Examples of identifiable finite mixtures include: the family of Gaussian distribution $\{N(\mu, \sigma^2), -\infty < \mu < \infty, 0 < \sigma^2 < \infty \}$, the family of one-dimensional Gaussian distributions, the family of one-dimensional gamma distributions, the family of multiple products of exponential distributions, the multivariate Gaussian family, the union of the last two families, the family of one-dimensional Cauchy distributions, etc.

For sufficient conditions for identifiability of arbitrary mixtures, \citet{teicher1961identifiability} studied the identifiability of mixtures of additive closed families, while \citet{barndorff1965identifiability} discussed the identifiability of mixtures of some restricted multivariate exponential families. \citet{luxmann1987identifiability} has given a sufficient condition for the identifiability of a large class of discrete distributions, namely that of the power-series distributions. Using topological arguments, he has shown that if the family in question is infinitely divisible, mixtures of this family are identifiable. For example, Poisson, negative binomial, logarithmic series are infinitely divisible, so arbitrary mixtures are identifiable. 

On the other hand, despite being a power-series distribution, the binomial distribution is not infinitely divisible. So its identifiability is not established for the success parameter \citep{sapatinas1995identifiability}. In fact, the binomial mixture has often been regarded as unidentifiable, as $F$ can be determined only up to its first $m$ moments when $H$ is known exactly. To be more specific,  $h(x| s)$ is a linear function of the first $m$ moments $ \mu_r= \int_0^1 s^r dF(s), r = 1,2,\ldots, m$ of $F(s)$, for every $x$. Therefore, with the same first $m$ moments, any other $F'(s)$  will make the same mixed distribution $H(x)$. To ensure the identifiability for the binomial mixture model, it is a common practice to assume that $F$ corresponds to a finite discrete pmf or a parametric density (e.g., beta distribution). 

In particular, there are two results for the identifiability of binomial mixture model \citep{teicher1963identifiability}:
\begin{enumerate}
\item If $h(x| s)$ in \eqref{eq:general_mixture} is a binomial distribution with a known binomial size $m$, and the support of $F$ only contain $K$ points. A necessary and sufficient condition that the identifiability holds is that $m \geq 2K - 1$.
\item Consider $h(x| m_j, s_j)$ as a binomial distribution with binomial size $m_j$ and probability $s_j$, where $0 < s_j < 1$, $j = 1, 2, \ldots$ and the support of $F$ is $\{s_1, s_2, \ldots \}$. If $m_j \neq m_{j'}$ for $j = j'$, then \eqref{eq:general_mixture} is identifiable.
\end{enumerate}

In this study, we are interested in the regime where the support size of $F$ may not be finite, and thus the identifiability may fail for the binomial mixture model. Some efforts are made without identifiability. \citet{lord1975empirical, sivaganesan1993robust} constructed credible intervals for the Bayes estimators of each point mass and that of $F$,  which rely on the lower order moments of the mixing distribution. \citet{wood1999binomial} empirically shows that their proposed nonparametric maximum likelihood estimator of $F$ is unique with high probability when $m$ is large. However, these results are far from satisfying in terms of our ultimate goal --- estimate or infer the underlying distribution $F$.

\section{A motivating Example: GeneFishing}\label{appendix:bmu_gene_fishing}
The study of model \eqref{eq:binomial_mixture} with a large $m$ is motivated by the GeneFishing method \citep{liu2019genefishing}. Provided some knowledge involved in a biological process as ``bait'', GeneFishing was designed to ``fish'' (or identify) discoveries that are yet identified related to this process. In this work, the authors used a set of pre-identified $21$ ``bait genes'', all of which have known roles in cholesterol metabolism, and then applied GeneFishing to three independent RNAseq datasets of human lymphoblastoid cell lines. They found that this approach identified other genes with known roles not only in cholesterol metabolism but also with high levels of consistency across the three datasets. They also applied GeneFishing to GTEx human liver RNAseq data and identified gene glyoxalase I (GLO1). In a follow-up wet-lab experiment, GLO1 knockdown increased levels of cellular cholesterol esters. 

The GeneFishing procedure is as follows, as shown in Figure \ref{fig:genefishing}:
\begin{enumerate}
\item Split the $n$ candidate genes randomly into many sub-search-spaces of $L$ genes per sub-group (e.g., $L = 100$), then added to by the bait genes. This step is the key reduction of search space, facilitating making the ``signal'' standing out from the ``noise''.
\item Construct the Spearman co-expression matrices for gene pairs contained within each sub-search-space. Apply the spectral clustering algorithm (with the number of clusters equal to $2$) to each matrix separately. In most cases, the bait genes cluster separately from the candidate genes. But in some instances, candidate gene(s) related to the bait genes will cluster with them. When this occurs, the candidate gene is regarded as being ``fished out''.
\item Repeat steps 1 and 2 (defining one round of GeneFishing) $m$ times (e.g., $m = 10,000$) to reduce the impact that a candidate gene may randomly co-cluster with the bait genes.
\item Aggregate the results from all rounds, and the $i$-th gene is fished out $X_i$ times out of $m$. The final output is a table that records the ``capture frequency rate'' ($CFR_i := \hat{s}_i = X_i/m$). The ``fished-out'' genes with large CFR values are thought of as ``discoveries''. Note, however, instead of considering these ``discoveries'' to perform a specific or similar function as the bait genes, we only believe they are likely to be functionally related to the bait genes. Figure \ref{fig:hist_CFR} displays the distribution of $X_i$'s with $m = 10,000$ and the number of total genes $n = 21,000$ on four tissues for the cholesterol-relevant genes.
\end{enumerate}

\begin{figure}[ht]
  \centering
  \includegraphics[width=0.8\textwidth]{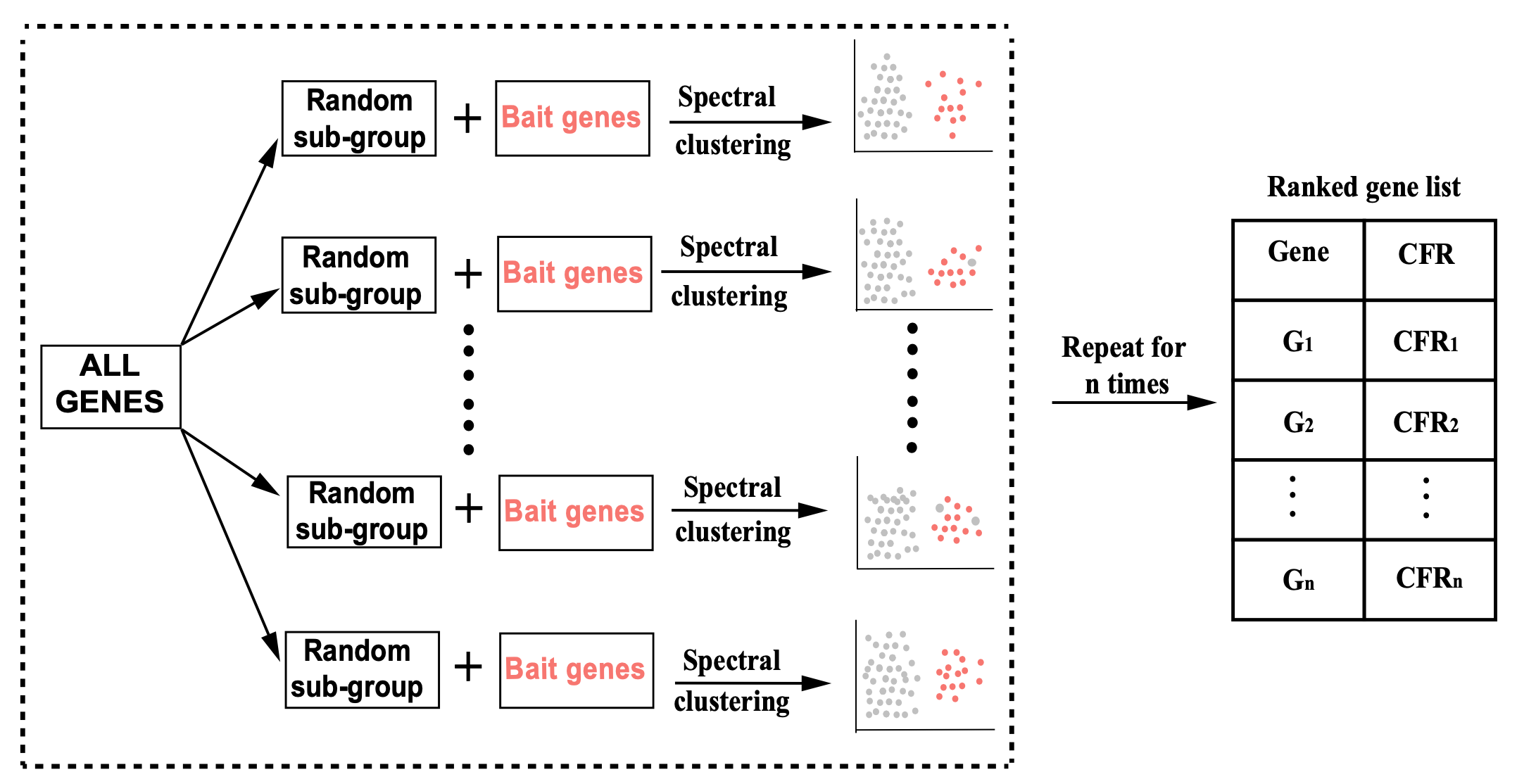}
  \caption{Workflow of GeneFishing (Fig 1 (e) of \citet{liu2019genefishing}). CFR stands for Capture Frequency Rate.}\label{fig:genefishing}
\end{figure}

Compared to other works for ``gene prioritization'', GeneFishing has three merits. First, it takes care of extreme inhomogeneity and low signal-to-noise ratio in high-throughput data using dimensionality reduction by subsampling. Second, it uses clustering to identify $21$ tightly clustered bait genes, which is a data-driven confirmation of the domain knowledge. Finally, GeneFishing leverages the technique of results aggregation (motivated by a bagging-like idea) in order to prioritize genes relevant to the biological process and reduce false discoveries.

However, there remains an open question on how large a CFR should be so that the associated gene is labeled as a ``discovery''. \citet{liu2019genefishing} picked a large cutoff $0.99$ by eye. It is acceptable when the histograms are sparse in the middle as in the liver or the transverse colon tissues (Figure \ref{fig:hist_CFR} (a)(b)), since cutoff$=0.25$ and cutoff$=0.75$ make little difference in determining the discoveries. On the other hand, for the artery coronary and testis tissues (Figure \ref{fig:hist_CFR} (c)(d)), the non-trivial middle part of the histogram is a mixture of the null (irrelevant to the biological process) distribution and the alternative (relevant to the biological process) distribution. The null and the alternative are hard to separate since the middle part is quite flat. Existing tools using parametric models to estimate local false discovery rates \citep{gauran2018empirical} are not able to account for the middle flatness well. They tend to select a smaller cutoff and produce excessive false discoveries. \citet{liu2019genefishing} provided a permutation-like procedure to compute approximate p-values and false discovery rates (FDRs). Nonetheless, there are two problems with this procedure. On the one hand, it is substantially computationally expensive, considering another round of permutation is added on top of numerous fishing rounds. On the other hand, the permutation idea is based on a strong null hypothesis that none of the candidate non-bait genes are relevant to the bait genes, thus producing an extreme p-value or FDR, which is unrealistic.

To identify a reasonable cutoff on CFRs, we assume there is an underlying fishing rate $s_i$, reflecting the probability that the $i$-th gene is fished out in each GeneFishing round. The fishing rates are assumed to be independently sampled from the same distribution $F$. Thus, the observation $CFR_i$ (or $\hat{s_i}$) can be modelled by \eqref{eq:binomial_mixture}. We mention that the independence assumption is raised mainly for convenience but not realistic de facto since genes are interactive and correlated in the same pathways or even remotely. Consequently, the effective sample size is smaller than expected. However, this assumption is still acceptable by assuming only a handful of candidate genes are related to the bait genes, which is reasonable from the biological perspective. Furthermore, we notice that in Figure \ref{fig:hist_CFR} there exists a clear pattern that the histogram is decreasing on the left-hand side and increasing on the right-hand side for all four tissues. In the middle, liver and colon transverse display sparse densities while artery coronary and testis exhibit flat ones. Thus, we can impose a U shape constraint on the associated density of $F$ (see Section \ref{sec:bmu_model} for details). Then the original problem becomes finding out the cutoff where the flat middle part transits to the increasing part on the right-hand side. 

\section{Review of Grenander Estimator}\label{appendix:bmu_review_grenander}
Monotone density models are often used in survival analysis and reliability analysis in economics---see \citet{huang1994estimating,huang1995estimation}. We can apply maximum likelihood for the monotone density estimation. Suppose that $X_1, \ldots, X_n$ is a random sample from a density $f$ on $[0, \infty)$ that is known to be nonincreasing; the maximum likelihood estimator $\ftil_n$ is defined as the nonincreasing density that maximizes the log-likelihood $\ell(f) = \sum_{i=1}^n \log f(X_i).$ \citet{grenander1956theory} first showed that this optimization problem has a unique solution under the monotone assumption --- so the estimator is also called the Grenander estimator. The Grenander estimator is explicitly given as the left derivative of the least concave majorant of the empirical distribution function $F_n$. The least concave majorant of $F_n$ is defined as the smallest concave function $\tilde{F}_n$ with $\tilde{F}_n \geq F_n$ for every $x$. 
Because $\tilde{F}_n$ is concave, its derivative is nonincreasing.

\citet{marshall1965maximum} showed that Grenander estimator is consistent when $f$ is decreasing, as stated in Theorem \ref{thm:grenander_consistency}.
\begin{theorem}[\citet{marshall1965maximum}]\label{thm:grenander_consistency}
  
  Suppose that $X_1, \ldots, X_n$ are i.i.d random variables with a decreasing density $f$ on $[0, \infty)$. Then the Grenander estimator $\ftil_n$ is uniformly consistent on closed intervals bounded away from zero: that is, for each $c > 0$, we have
  $$\sup_{c \leq x < \infty} |\ftil_n(x) - f(x)| \rightarrow 0 \text{ a.s.}$$
\end{theorem}
The inconsistency of the Grenander estimator at $0$, when $f(0)$ is bounded, was first pointed out by \citet{woodroofe1993penalized}. \citet{balabdaoui2011grenander}  later extended this result to other situations, where they consider different behavior near zero of the true density. Theorem \ref{thm:grenander_inconsistency_0} explicitly characterizes the behavior of $\ftil_n$ at zero.
\begin{theorem}[\citet{woodroofe1993penalized}]\label{thm:grenander_inconsistency_0}
  
  Suppose that $f$ is a decreasing density on $[0, \infty)$ with $0 < f(0) < \infty$, and let $N(t)$ denote a rate $1$ Poisson process. Then
  $$\frac{\ftil_n(0)}{f(0)} \overset{d}{\rightarrow} \sup_{t > 0} \frac{N(t)}{t} \overset{d}{=} \frac{1}{U},$$
  where $U$ is a uniform random variable on the unit interval.
\end{theorem}
\citet{birge1989grenander} proved that Grenander estimator has a cubic root convergence rate in the sense of $L_1$ norm, as in Theorem \ref{thm:l1_grenander}. \citet{van2003smooth} pointed out that the rate of convergence of the Grenander estimator is slower than that of the monotone kernel density estimator when the underlying function is smooth enough.
\begin{theorem}[\citet{birge1989grenander}]\label{thm:l1_grenander}
  
  Suppose $f$ is a decreasing density on $[0, \infty)$ with $0 < f(0) < \infty$. it follows that 
  $$\Ebb_f \int_{0}^1 |\ftil_{n}(x) - f(x)| dx \leq C \cdot n^{-\frac{1}{3}},$$
  where $C$ is a constant that only depends on $f$.
\end{theorem}
\citet{rao1970estimation} first obtained the limiting distribution of the Grenander estimator at a point. He has proved that $\sqrt[3]{n}(\ftil_n(t_0) - f(t_0))$ converges to the location of the maximum of the process $\{B(x) - x^2, x \in \Rbb\}$, where $f'(t_0) < 0$ and $B(x)$ is the standard two-sided Brownian motion on $\Rbb$ such that $B(0) = 0$; see \citet{rao1970estimation}. \citet{wang1992nonparametric} extends this result to the flat region and a higher order derivative, as stated in Theorem \ref{thm:grenander_local_asymp}. 

\begin{theorem}[\citet{wang1992nonparametric}]\label{thm:grenander_local_asymp}
  
  Suppose $f$ is a decreasing density on $[0, 1]$ and is smooth at $t_0 \in (0, 1)$. It follows that
  \begin{itemize}
  \item [(A)] If $f$ is flat in a neighborhood of $t_0$. Let $[a, b]$ be the flat part containing $t_0$. Then,
    $$\sqrt{n}(\ftil_{n}(t_0) - f(t_0)) \overset{d}{\rightarrow} \hat{S}_{a, b}(t_0),$$
    where $\hat{S}_{a, b}(t)$ is the slope at $F(t)$ of the least concave majorant in $[F(a), F(b)]$ of a standard Brownian Bridge in $[0, 1]$.
  \item [(B)] If $f(t) - f(t_0) \sim f^{(k)}(t_0)(t - t_0)^k$ near $t_0$ for some $k$ and $f^{(k)}(t_0) < 0$. Then,
    $$n^{\frac{k}{2k+1}}[\frac{f^k(t_0) |f^{(k)}(t_0)|}{(k+1)!}]^{-\frac{1}{2k+1}} (\ftil_{n}(t_0) - f(t_0)) \overset{d}{\rightarrow} V_k(0),$$
    where $V_k(t)$ is the slope at $t$ of the least concave majorant of the process $\{B(t) - |t|^{k+1}, t \in (-\infty, \infty)\}$, and $B(t)$ is a standard two-sided Brownian motion on $\Rbb$ with $B(0) = 0$.
  \end{itemize}
\end{theorem}

\section{Proof of Proposition \ref{prop:CDF_deviation_bounds}}\label{appendix:bmu_proof_CDF_deviation_bounds}
\begin{proof}
  Suppose
  $$f(x) = 1.8 \cdot \Ibb(x \in [0, 1/2]) + 0.2 \cdot \Ibb(x \in (1/2, 1]).$$
  Note that
  \begin{eqnarray}
    && \Pbb(s \leq k/m) - \Pbb(\hat{s} \leq k/m)\nonumber\\
    &=& \int_0^{k/m} f(u) du - \sum_{r \leq k} \int_0^1 {m \choose r} u^r (1 - u)^{m - r} f(u)du \nonumber \\
    &=& \int_0^1 \left (\Ibb[u \leq k/m] - \sum_{r \leq k}{m \choose r} u^r (1 - u)^{m - r} \right ) f(u)du \label{eq:CDF_diff}
  \end{eqnarray}
  Decompose $f(x) = f_1(x) + f_2(x)$, where $f_1(x) = 1.6 \cdot \Ibb(x \in [0, 1/2])$, $f_2(x) \equiv 0.2$. The previous example shows that the difference for the $f_2$ part in Equation \eqref{eq:CDF_diff} is at most $\frac{0.2}{m+1}.$ So we only need to take care of the $f_1$ part in Equation \eqref{eq:CDF_diff}, i.e.,
  $$1.6 \times \int_0^{1/2} \left (\Ibb[u \leq k/m] - \sum_{r \leq k}{m \choose r} u^r (1 - u)^{m - r} \right ) du = 1.6 \times \int_0^{1/2} B_k(m, u) du,$$  
  provided $k/m \leq 1/2$. Here $B_k(m, x) = \sum_{r = k}^m {m \choose k} x^r (1 - x)^{m - r}$. Define
  $$A_k(m, x) = [{m \choose k} x^k (1 - x)^{m - k + 1}] \cdot [(k+1)/(k+1 - (m + 1)x)].$$
  \citet{bahadur1960some}[Theorem 1] indicates that $1 \leq A_{k}(m, x)/B_k(m, x) \leq 1 + z^{-2}$, where $z = (k - mx)/(mx(1-x))^{\frac{1}{2}}$. Let $x = 1/2 - \epsilon$. By Stirling's formula and Taylor expansion on $\log(1 + \epsilon)$, we can obtain 
  $$A_{m/2}(m, 1/2 - \epsilon) \sim \frac{1}{\sqrt{2\pi m}}e^{-2m\epsilon^2},$$
  and we can rewrite
  $$z = \frac{\sqrt{m} \epsilon}{\sqrt{1/4 - \epsilon^2}}.$$
  So we have
  \begin{eqnarray*}
    B_{m/2}(m, 1/2 - \epsilon) \geq \frac{A_{m/2}(m, 1/2 - \epsilon)}{1 + z^{-2}} \sim  \frac{2\sqrt{m} \epsilon e^{-2m\epsilon^2}}{4(m - 1)\epsilon^2 + 1} \geq \frac{2\sqrt{m} \epsilon e^{-2m\epsilon^2}}{4m\epsilon^2 + 1}.
  \end{eqnarray*}
  Then it follows that
  $$\int_0^{1/2} B_{m/2}(m, 1/2 - \epsilon) d\epsilon \geq \int_0^{1/\sqrt{m}} \frac{2\sqrt{m} \epsilon e^{-2m\epsilon^2}}{4m\epsilon^2 + 1} d\epsilon= \frac{1}{\sqrt{m}} \int_0^1 \frac{2e^{-2u^2}}{4u^2 + 1} du \approx \frac{0.81}{\sqrt{m}}.$$
  Together, we have $\Pbb(s \leq 1/2) - \Pbb(\hat{s} \leq 1/2) \geq \frac{C}{\sqrt{m}} + \varepsilon \cdot m^{-1}$, where $\varepsilon$ is a residual term with $|\varepsilon| \leq K$, $C$ and $K$ are positive constants.
\end{proof}

\section{Proof of Proposition \ref{prop:CDF_deviation_smooth}}\label{appendix:bmu_proof_CDF_deviation_smooth}
\begin{proof}  
  The proof of this theorem follows from that of Theorem \ref{thm:noisy_hist_density_upper_risk}, so we defer most details to the proof of the latter.

  From Equation \eqref{eq:hist_bias_decomposition}, we know that
    \begin{eqnarray*}
   && \Pbb (\hat{s}_1 \in B(x)) - \Pbb(s_1 \in B(x)) \nonumber\\
                                          &=& \sum_{d = 1}^D [\Pbb(\hat{s}_1 \in B(x), s_1 \in B(x + d\cdot h)) - \Pbb(s_1 \in B(x), \hat{s}_1 \in B(x + d\cdot h))]\nonumber \\
   && + \sum_{d = 1}^D [\Pbb(\hat{s}_1 \in B(x), s_1 \in B(x - d\cdot h)) - \Pbb(s_1 \in B(x), \hat{s}_1 \in B(x - d\cdot h))]\nonumber\\
                                          && + \Pbb(\hat{s}_1 \in B(x), s_1 \in B(x + d \cdot h): d \geq D + 1\text{ or } d \leq -D-1)\nonumber\\
                                          && + \Pbb(s_1 \in B(x), \hat{s}_1 \in B(x + d \cdot h): d \geq D + 1\text{ or } d \leq -D-1),\nonumber
    \end{eqnarray*}
    where $B(x)$ is still defined as the bin that contains $x$ among $[0,h], (h, 2h], \ldots, (1-h, 1]$. Note that
    \begin{eqnarray}
      && \Pbb(\hat{s}_1 \leq x) - \Pbb(s_1 \leq x) \nonumber\\
      &=& \Pbb(\hat{s}_1 \leq x, s_1 > x) - \Pbb(\hat{s}_1 >x, s_1 \leq x)\nonumber\\
      &=& \Pbb(\hat{s}_1 \in B(x), s_1 \in B(x), \hat{s}_1 \leq x, s_1 > x) - \Pbb(\hat{s}_1 \in B(x), s_1 \in B(x), \hat{s}_1 > x, s_1 \leq x) \nonumber\\
      &&+\sum_{{d = 1, 2, \ldots} \atop {d'=1,2,\ldots}} [\Pbb(\hat{s}_1 \in B(x - dh), s_1 \in B(x + d'h)) - \Pbb(\hat{s}_1 \in B(x + dh), s_1 \in B(x - d'h))].\nonumber\\
      &=& \Pbb(\hat{s}_1 \in B(x), s_1 \in B(x), \hat{s}_1 \leq x, s_1 > x) - \Pbb(\hat{s}_1 \in B(x), s_1 \in B(x), \hat{s}_1 > x, s_1 \leq x) \nonumber\\
      &&+(\sum_{|d - d'| \leq D} + \sum_{|d-d'| > D})[\Pbb(\hat{s}_1 \in B(x - dh), s_1 \in B(x + d'h))\nonumber\\
      && ~~~~~~~~~~~~~~~~~~~~~~~~~~~- \Pbb(\hat{s}_1 \in B(x + dh), s_1 \in B(x - d'h))].\nonumber
    \end{eqnarray}
    Following the proof of bounding $\Pbb(\hat{s}_1 \in B(x), s_1 \in B(x + d \cdot h): d \geq D + 1\text{ or } d \leq -D-1)$ in Theorem \ref{thm:noisy_hist_density_upper_risk}, we can bound
    \begin{eqnarray*}
      &&\vert \sum_{|d-d'| > D}[\Pbb(\hat{s}_1 \in B(x - dh), s_1 \in B(x + d'h)) - \Pbb(\hat{s}_1 \in B(x + dh), s_1 \in B(x - d'h)) \vert \\
     &\leq& 2f_{\max} \cdot \exp(-2mD^2h^2) = \frac{2f_{\max}}{m},
      \end{eqnarray*}
     where $D = \lceil \sqrt{\frac{\log m}{2mh^2}} \rceil$.\\
     Following the proof of bounding $\vert \sum_{d = 1}^D [\Pbb(\hat{s}_1 \in B(x), s_1 \in B(x + d\cdot h)) - \Pbb(s_1 \in B(x), \hat{s}_1 \in B(x + d\cdot h))] \vert $ in Theorem \ref{thm:noisy_hist_density_upper_risk}, we can bound
     \begin{eqnarray*}
       &&|\Pbb(\hat{s}_1 \in B(x), s_1 \in B(x), \hat{s}_1 \leq x, s_1 > x) - \Pbb(\hat{s}_1 \in B(x), s_1 \in B(x), \hat{s}_1 > x, s_1 \leq x)|\\
       &\leq& K_1 \cdot (\frac{f_{\max}\cdot h}{\sqrt{m}} + \frac{f_{\max}}{m} + \frac{h\cdot f'_{\max}}{\sqrt{m}}) + \vert \Ecal\vert,
       \end{eqnarray*}
    where $\vert \Ecal \vert \leq K_2 \cdot (\frac{h\cdot f_{\max}}{\sqrt{m}} + h^2 f_{\max} + \frac{f_{\max}}{m})$ for some constant $K_2$ that only depends on $a$.\\
    To bound
    $$\vert \sum_{|d-d'| \leq D}[\Pbb(\hat{s}_1 \in B(x - dh), s_1 \in B(x + d'h)) - \Pbb(\hat{s}_1 \in B(x + dh), s_1 \in B(x - d'h)) \vert ,$$
    we still follow the proof of bounding $\vert \sum_{d = 1}^D [\Pbb(\hat{s}_1 \in B(x), s_1 \in B(x + d\cdot h)) - \Pbb(s_1 \in B(x), \hat{s}_1 \in B(x + d\cdot h))] \vert $ in Theorem \ref{thm:noisy_hist_density_upper_risk}. The problem is that there are $D^2$ terms instead $D$ terms. By carefully arranging the $D^2$ terms, and using the fact that $\sum_{d =1}^{\infty} \exp(-C_0d^2)< \infty$ for any positive $C_0$, it can be easily seen that we still get the same bound as above. In total, we have
    $$\vert \Pbb(\hat{s}_1 \leq x) - \Pbb(s_1 \leq x)\vert  \leq K_3 \cdot (f_{\max} + f'_{\max}) \cdot (\frac{h}{\sqrt{m}} + h^2 + \frac{1}{m}),$$
    where $K_3$ is some constant that only hinges on $a$. To minimize the upper bound, take $h = \frac{1}{\sqrt{m}}$. Thus, it follows that
    $$\sup_{x \in [a, 1-a]} |F^{(m)}(x) - F(x)| \leq \frac{C}{m},$$
    where $C$ is some constant that only depends on $f$ and $a$.
\end{proof}

\section{Proof of Theorem \ref{thm:noisy_hist_density_upper_risk}}\label{appendix:bmu_proof_noisy_hist_density_upper_risk}

Theorem \ref{thm:noisy_hist_density_upper_risk} relies on the local limit theorem of binomial distribution, as follows.

\begin{lemma}
  Suppose $X \sim Binom(m, s)$, with $0 < s < 1$. For any $a < b$ such that $\Delta := \max\{|\frac{a}{m} - s|, |\frac{b}{m} - s|\} \rightarrow 0$ as $m \rightarrow \infty$, then 
  $$\Pbb(a \leq X \leq b) = \left [\int_{\frac{a - ms}{\sqrt{ms(1-s)}}}^{\frac{b - ms}{\sqrt{ms(1-s)}}} \frac{1}{\sqrt{2\pi}} \exp(-t^2/2) dt \right ](1 + \varepsilon_1) + \varepsilon_2,$$
  where $\varepsilon_1$ is the error from the Gaussian approximation to Binomial point mass function, and $\varepsilon_2$ is the error from the summation series approximating the integral. Specifically, we have
  $$\vert \varepsilon_1 \vert \leq K \cdot \Delta,$$
  $$\vert \varepsilon_2\vert \leq C \cdot [\frac{\exp(-m\delta^2)}{\sqrt{m}} + \frac{b-a}{\sqrt{m}} \cdot \Delta  \cdot \exp(-m\delta^2)],$$
  where $\delta:= \min_{x \in [\frac{a}{m}, \frac{b}{m}]} |s - x|$, $K$ and $C$ are two positive constants that depend on $s$ via $\frac{1}{s(1-s)}$.
  \label{lemma:DeMoivreLaplace}
\end{lemma}
The detailed proof of Lemma \ref{lemma:DeMoivreLaplace} can be easily obtained by adapting that of \citet{dunbar2011moivre}. Now we prove Theorem \ref{thm:noisy_hist_density_upper_risk}:\\
\begin{proof}
  Denote by $R_x = \Ebb (f(x) - \hat{f}_{n,m}(x))^2$ the risk at a point $x$. We decompose the risk into the variance and the bias square as follows.
  \begin{equation}
    R_x = \Var(\hat{f}_{n,m}(x)) + (\Ebb \hat{f}_{n,m}(x) - f(x))^2.\label{eq:hist_bias_var}
  \end{equation}
  For the variance part, denote by $p_l = \Pbb(m\cdot \hat{s}_1 \in B_l)$. We have
  \begin{eqnarray}
    \int_a^{1-a} \Var(\hat{f}_{n,m}(x))dx = \sum_{aL \leq l \leq (1 - a)L} \int_{B_l} \frac{p_l(1-p_l)}{nh^2} \leq K_1 \cdot \frac{1}{nh},\label{eq:hist_var}
  \end{eqnarray}
  where $K_1$ is a positive constant that relies on $a > 0$.
  For the bias part, since
  \begin{equation}
    (\Ebb \hat{f}_{n,m}(x) - f(x))^2 \leq 2(\Ebb \hat{f}_{n,m}(x) - \Ebb \hat{f}_n(x))^2 + 2 (\Ebb \hat{f}_n(x) - f(x))^2, \label{eq:hist_noiseless_noisy_bias}
  \end{equation}
    and it is well known that
\begin{equation}
  \int_0^1 (\Ebb\hat{f}_n(x) - f(x))^2 dx \leq K_2 \cdot h^2,\label{eq:hist_noiseless_bias}
\end{equation}
where $K_2$ is positive constant that only relies on $f$. We only need to consider $\Ebb \hat{f}_{n,m}(x) - \Ebb \hat{f}_n(x)$. By definition,
  \begin{eqnarray}
    && \Ebb \hat{f}_{n,m}(x) - \Ebb \hat{f}_n(x)\nonumber\\
    &=& \frac{1}{h}[\Pbb (\hat{s}_1 \in B(x)) - \Pbb(s_1 \in B(x))]\nonumber\\
                                          &=& \frac{1}{h}[\Pbb(\hat{s}_1 \in B(x), s_1 \not \in B(x)) - \Pbb(s_1 \in B(x), \hat{s}_1 \not \in B(x))]\nonumber\\
                                          &=& \frac{1}{h}\{\sum_{d = 1}^D [\Pbb(\hat{s}_1 \in B(x), s_1 \in B(x + d\cdot h)) - \Pbb(s_1 \in B(x), \hat{s}_1 \in B(x + d\cdot h))]\nonumber \\
   && + \sum_{d = 1}^D [\Pbb(\hat{s}_1 \in B(x), s_1 \in B(x - d\cdot h)) - \Pbb(s_1 \in B(x), \hat{s}_1 \in B(x - d\cdot h))]\nonumber\\
                                          && + \Pbb(\hat{s}_1 \in B(x), s_1 \in B(x + d \cdot h): d \geq D + 1\text{ or } d \leq -D-1)\nonumber\\
                                          && + \Pbb(s_1 \in B(x), \hat{s}_1 \in B(x + d \cdot h): d \geq D + 1\text{ or } d \leq -D-1)\}.\label{eq:hist_bias_decomposition}
  \end{eqnarray}
  By McDiarmid’s inequality, it follows that
  \begin{eqnarray*}
    &&\Pbb(s_1 \in B(x), \hat{s}_1 \in B(x + d \cdot h): d \geq D + 1\text{ or } d \leq -D-1)\nonumber\\
    &=&\int_{s_1 \in B(x)} [\Ebb \Ibb(\hat{s}_1 \in B(x + d \cdot h): d \geq D + 1\text{ or } d \leq -D-1 |s_1)] f(s_1)ds_1\nonumber\\                                                                                             &\leq& \int_{s_1 \in B(x)} \Pbb(|\hat{s}_1 - s_1| \geq D\cdot h | s_1)f(s_1)ds_1\nonumber\\                                                                                        &\leq& \int_{s_1 \in B(x)} \exp(-2mD^2h^2) f(s_1) ds_1\nonumber\\
                                                                                             &\leq& f_{\max}\cdot h \cdot \exp(-2mD^2h^2),
  \end{eqnarray*}
  where $f_{\max}$ is the maximal value of $f$ in $[0, 1]$. Take $D = \lceil \sqrt{\frac{\log m}{4mh^2}} \rceil$, i.e., the least integer that is not smaller than $\sqrt{\frac{\log m}{4mh^2}}$. Then we have
  \begin{equation}
    \Pbb(s_1 \in B(x), \hat{s}_1 \in B(x + d \cdot h): d \geq D + 1\text{ or } d \leq -D-1) \leq K_3 \cdot \frac{f_{\max}\cdot h}{\sqrt{m}},\label{eq:hist_remote_rate}
  \end{equation}
   where $K_3$ is a universal positive constant. Similarly, we can show that $\Pbb(\hat{s}_1 \in B(x), s_1 \in B(x + d \cdot h): d \geq D + 1\text{ or } d \leq -D-1) \leq K_4 \cdot \frac{f_{\max}}{m}$ for some positive constant $K_4$. Next, we investigate $\Pbb(\hat{s}_1 \in B(x), s_1 \in B(x + d\cdot h)) - \Pbb(s_1 \in B(x), \hat{s}_1 \in B(x + d\cdot h))$. Denote by $l(x)$ and $r(x)$ the left boundary and the right boundary of the interval $B(x)$. By Lemma \ref{lemma:DeMoivreLaplace}, it follows that
  \begin{eqnarray*}
    && \Pbb(s_1 \in B(x), \hat{s}_1 \in B(x + d\cdot h))\\
    &=& \int_{s_1 \in B(x)} \left [\sum_{k: \frac{k}{m} \in B(x + d\cdot h)} {m \choose k} (s_1)^k(1-s_1)^{m-k} \right ] f(s_1)ds_1\\
    &= & \int_{s \in B(x)} \left \{ \underbrace{\left [ \int_{t \in B(x+d\cdot h)}\frac{\sqrt{m}}{\sqrt{2\pi s(1-s)}}\exp(-\frac{m(t - s)^2}{2s(1-s)})dt\right ] \cdot (1+ \varepsilon_5 \cdot (d+1)h)}_{\text{(I)}} \right .\\
    && \left . ~~~~~~~~~~~ + \underbrace{\varepsilon_6 \cdot \frac{\exp(-m[r(x)+(d-1)h - s]^2) }{\sqrt{m}}}_{\text{(II)}} \right .\\
    && \left . ~~~~~~~~~~~ + \underbrace{\varepsilon_7 \cdot \sqrt{m}h\cdot [r(x) + dh - s] \cdot \exp (-m[r(x) + (d-1)h - s]^2 }_{\text{(III)}} \right \}  f(s)ds,
  \end{eqnarray*}
  where $\vert \varepsilon_5 \vert \leq K_5$, $\vert \varepsilon_6 \vert \leq K_6$, $\vert \varepsilon_7 \vert \leq K_7$ and $K_5, K_6, K_7$ are positive constants that only depend on $a$.
  We consider the summation of the $D$ error terms in $\Pbb(s_1 \in B(x), \hat{s}_1 \in B(x + d\cdot h))$, $d = 1, \ldots, D$:
  \begin{enumerate}
  \item [(I)]
    \begin{eqnarray*}
      &&  \sum_{d = 1}^D \int_{s \in B(x)} \left [ \int_{t \in B(x+d\cdot h)}\frac{\sqrt{m}}{\sqrt{2\pi s(1-s)}}e{-\frac{m(t - s)^2}{2s(1-s)}}dt\right ]  K_5 (d+1)hf(s)ds\\
      &\leq & K_8h^2f_{\max} + \sum_{d = 3}^D K_9 \int_{s \in B(x)} \left [ \int_{t \in B(x+d\cdot h)} \sqrt{m}(d+1)h \cdot e^{-\frac{m(t - s)^2}{2s(1-s)}}dt\right ] f(s)ds\\
      &\leq& K_8h^2f_{\max} + \sum_{d = 3}^D K_9 \int_{s \in B(x)} \left [ \int_{t \in B(x+d\cdot h)} \sqrt{m}(d+1)h \cdot e^{-m(t - s)^2}dt\right ]  f(s)ds\\
      &\leq& K_8h^2f_{\max} +  \sum_{d = 3}^D K_9 \int_{s \in B(x)} \left [ \int_{t \in B(x+d\cdot h)} 2\sqrt{m}(d-1)h \cdot e^{-m(t - s)^2}dt\right ] f(s)ds\\
      &\leq& K_8h^2f_{\max} + \sum_{d = 3}^D K_9 \int_{s \in B(x)} \left [ \int_{t \in B(x+d\cdot h)} 2\sqrt{m}(t-s) \cdot e^{-m(t - s)^2}dt\right ] f(s)ds\\
      &=& K_8h^2f_{\max} \\
      && + \sum_{d = 3}^D \frac{2K_9}{\sqrt{m}}\int_{s \in B(x)} \left [ e^{-m(r(x)  - s + (d-1)h)^2} - e^{-m(r(x) - s + dh)^2} \right ]f(s)ds\\
      &=& K_8h^2f_{\max} + \tilde{K}_9 \frac{h f_{\max}}{\sqrt{m}},
    \end{eqnarray*}
    where $K_8, K_9, \tilde{K}_9$ are positive constants that only depend on $a$.
  \item [(II)]
  \begin{eqnarray*}
    && \sum_{d = 1}^D \int_{s\in B(x)} K_6 \cdot \frac{\exp(-m[r(x)+(d-1)h - s]^2) }{\sqrt{m}} \cdot  f(s) ds\\
    &\leq& K_6 \cdot \frac{f_{\max}}{\sqrt{m}} \cdot \int_{0}^{(D-1)h} \exp(-mt^2)dt \\
    &\leq& \tilde{K}_6 \cdot \frac{f_{\max}}{m},
  \end{eqnarray*}
  where $\tilde{K}_6$ is a positive constant that only depends on $a$.
\item [(III)]
    \begin{eqnarray*}
      && \sum_{d = 1}^D \int_{s\in B(x)} K_7 \cdot \sqrt{m}h\cdot [r(x) + dh - s] \cdot e^{-m[r(x) + (d-1)h - s]^2} \cdot f(s) ds\\
      &=& \sum_{d = 1}^D \int_{s\in B(x)} K_7 \cdot \sqrt{m}h\cdot [r(x) + (d-1)h - s + h] \cdot e^{-m[r(x) + (d-1)h - s]^2} \cdot f(s) ds\\
    &\leq& \tilde{K}_7 \cdot \left \{h^2 \cdot f_{\max} + \frac{h\cdot f_{\max}}{\sqrt{m}} \cdot \sum_{d = 1}^D [e^{-m (d-1)^2h^2} - e^{-md^2h^2}] \right \}\\
    &=& \tilde{K}_7 \cdot \left \{ h^2 \cdot f_{\max} + \frac{h\cdot f_{\max}}{\sqrt{m}} \cdot  [1 - e^{-mD^2h^2}]\right \}\\
    &\leq& \tilde{K}_7 \cdot (h^2 \cdot f_{\max} + \frac{h\cdot f_{\max}}{\sqrt{m}}),       
    \end{eqnarray*}
    where $\tilde{K}_7 > 0$ only depends on $a$.
\end{enumerate}
We call the summation of the $D$ error terms by $\Ecal$, which satisfies $\vert \Ecal\vert \leq  K_{10} \cdot (\frac{h\cdot f_{\max}}{\sqrt{m}} + h^2 f_{\max} + \frac{f_{\max}}{m})$, where $K_{10} > 0$ only depends on $a$. Similarly, for the summation of the $D$ error terms $\tilde{\Ecal}$ in $\Pbb(\hat{s}_1 \in B(x), s_1 \in B(x + d\cdot h))$, $d = 1, \ldots,D$, we have the same rate. Now we consider
\begin{eqnarray}
  &&\vert \sum_{d = 1}^D \Pbb(\hat{s}_1 \in B(x), s_1 \in B(x + d\cdot h)) - \Pbb(s_1 \in B(x), \hat{s}_1 \in B(x + d\cdot h)) \vert \nonumber\\
  &\overset{(i)}{=}& \sum_{d = 1}^D \int_{s \in B(x)} \int_{t \in B(x+d\cdot h)} \left \vert \frac{\sqrt{m}}{\sqrt{2\pi s(1-s)}}e^{-\frac{m(t - s)^2}{2s(1-s)}}f(s) \right .\nonumber\\
  && ~~~~~~~~~~~~~~~~~~~~~~~~~~~~~ \left .- \frac{\sqrt{m}}{\sqrt{2\pi t(1-t)}}e^{-\frac{m(s - t)^2}{2t(1-t)}}f(t)\right \vert dtds + \vert \Ecal \vert  + \vert \tilde{\Ecal} \vert \nonumber\\
  &\overset{(ii)}{\leq}& \sum_{d = 1}^D \int_{s \in B(x)} \int_{t \in B(x+d\cdot h)} \left \vert \frac{|2\tilde{s}-1|}{(1-\tilde{s})^{\frac{5}{2}}\tilde{s}^{\frac{5}{2}}} \cdot \frac{m^{\frac{3}{2}}(t-s)^{3}}{2\sqrt{2\pi}} e^{-\frac{m(t-s)^2}{2\tilde{s}(1-\tilde{s})}}f(\tilde{s})\right \vert dtds\nonumber\\
  && + \sum_{d = 1}^D \int_{s \in B(x)} \int_{t \in B(x+d\cdot h)} \left \vert \frac{|2\tilde{s}-1|}{(1-\tilde{s})^{\frac{3}{2}}\tilde{s}^{\frac{3}{2}}} \cdot \frac{m^{\frac{1}{2}}(t-s)}{2\sqrt{2\pi}} e^{-\frac{m(t-s)^2}{2\tilde{s}(1-\tilde{s})}}f(\tilde{s}) \right \vert dtds\nonumber\\
    && + \sum_{d = 1}^D \int_{s \in B(x)} \int_{t \in B(x+d\cdot h)} \left \vert \frac{1}{(1-\tilde{s})^{\frac{1}{2}}\tilde{s}^{\frac{1}{2}}} \cdot \frac{m^{\frac{1}{2}}(t-s)}{\sqrt{2\pi}} e^{-\frac{m(t-s)^2}{2\tilde{s}(1-\tilde{s})}}f'(\tilde{s}) \right \vert dtds + \vert \Ecal \vert + \vert \tilde{\Ecal}\vert \nonumber\\
   &\overset{(iii)}{\leq}& \sum_{d = 1}^D \int_{s \in B(x)} \int_{t \in B(x+d\cdot h)} \left \vert  K_{11} \cdot f_{\max} \cdot m^{\frac{3}{2}}(t-s)^{3} e^{-m(t-s)^2}\right \vert dtds\nonumber\\
  && + \sum_{d = 1}^D \int_{s \in B(x)} \int_{t \in B(x+d\cdot h)} \left \vert K_{12} \cdot (f_{\max} + f'_{\max}) \cdot m^{\frac{1}{2}}(t-s) e^{-m(t-s)^2}\right \vert dtds + \vert \Ecal \vert + \vert \tilde{\Ecal}\vert \nonumber\\
  &\overset{(iv)}{=}& K_{13} \cdot (\frac{f_{\max}\cdot h}{\sqrt{m}} + \frac{f_{\max}}{m} + \frac{h\cdot f'_{\max}}{\sqrt{m}}) + \vert \Ecal \vert + \vert \tilde{\Ecal}\vert , \label{eq:hist_right_nbh_rate}
\end{eqnarray}
where $K_{11}, K_{12}, K_{13}$ are positive constants that only depend on $a$. Equation (i) uses the Fubini's theorem; Inequality (ii) applies the mean value theorem to the function $g(s) = \frac{1}{s(1-s)}\exp(-\frac{A}{2s(1-s)})f(s)$, where $A$ is a constant; Inequality (iii) holds since $0< a \leq s \leq \tilde{s} \leq t \leq 1-a < 1$, thus $\frac{1}{(1 - \tilde{s})\tilde{s}}$ is bounded, and $\exp(-\frac{m(t-s)^2}{2\tilde{s}(1- \tilde{s})})$ attains the maximal when $\tilde{s} = 1/2$; Inequality (iv) is obtained via integral by part. Similarly, $\sum_{d = 1}^D [\Pbb(\hat{s}_1 \in B(x), s_1 \in B(x - d\cdot h)) - \Pbb(s_1 \in B(x), \hat{s}_1 \in B(x - d\cdot h))]$ has the same rate as \eqref{eq:hist_right_nbh_rate}. Putting \eqref{eq:hist_bias_decomposition}\eqref{eq:hist_remote_rate}\eqref{eq:hist_right_nbh_rate} together, we have
\begin{equation}
  |\Ebb \hat{f}_{n,m}(x) - \Ebb \hat{f}_n(x)| \leq K_{14} \cdot (f_{\max} + f'_{\max}) \cdot (\frac{1}{\sqrt{m}} + h + \frac{1}{mh}), \label{eq:hist_bias_rate}
\end{equation}
where $K_{14}$ is some constant that only depends on $a$. Combining Inequalities \eqref{eq:hist_bias_var}\eqref{eq:hist_var} \eqref{eq:hist_noiseless_noisy_bias}\eqref{eq:hist_noiseless_bias} and Inequality \eqref{eq:hist_bias_rate}, it follows that
$$R(a, 1 - a) \leq C_1 \cdot (h^2 + \frac{1}{m} + \frac{1}{m^2h^2} + \frac{1}{nh}),$$
The minimal risk is no larger than $C_4 \cdot n^{-\frac{2}{3}}$, which is attained when $h = C_3 \cdot n^{-\frac{1}{3}}$, $m \geq C_2 \cdot n^{\frac{2}{3}}$. Here $C_1, C_2, C_3, C_4$ are positive constants that only depend on $a$ and $f$. 
\end{proof}

\section{Proof of Theorem \ref{thm:noisy_hist_mass_upper_risk}}\label{appendix:bmu_proof_noisy_hist_mass_upper_risk}
\begin{proof}
  Note for the point mass function $p(x) = \sum_{k=1}^K \alpha(k) \Ibb(x = x_k)$, we have an additional information that only hold for the discrete case but not for the density case,
  $$\max_{x \in I_{k+d}}|x - x_k| = (d+\frac{1}{2})\cdot \frac{1}{K}.$$
  We follow the proof of Theorem \ref{thm:noisy_hist_density_upper_risk} and can show that
  $$R(a, 1 - a) \leq C_1 \cdot ( \frac{1}{n} + \frac{1}{m} + \frac{K^2}{m^2}).$$
  When $m \geq C_2 \cdot \sqrt{n}\max\{K, \sqrt{n}\}$, we have $R(a, 1-a) \leq C_3 \cdot \frac{1}{n}$. Here where $C_1, C_2, C_3 > 0$ do not depend on $n$, $m$ and $K$. 
\end{proof}

\section{Proof of Theorem \ref{thm:noisy_l1_grenander}}\label{appendix:bmu_proof_noisy_l1_grenander}

We first define a few concepts. Let $J$ denote the interval $[a, b)$, where $a = 0$ and $b = 1$ in our setup. We set
  $$ l(J) = b - a,  f(J) = \int_J f(t) dt,  \Delta f(J) = f(b) - f(a), $$
    $$  bf(J) = \int_J |f(J)/l(J) - f(t)|dt.$$
    Any finite increasing sequences $\{x_i\}_{0 \leq x_i \leq q}$ with $x_0 = a$, $x_q = b$ generates a partition $\mathcal{P}$ of $J$ into intervals $J_i = [x_{i-1}, x_i)$, $1 \leq i \leq q$. When no confusion arises, we put $f_i$ for $f(J_i)$, $\Delta f_i$ for $\Delta f(J_i)$ and so on. Set a functional $L(\mathcal{P},f, z)$ defined for positive $z$ by
      $$L(\mathcal{P}, f, z) = \sum_{i = 1}^q [bf(J_i) + z(f(J_i))^{1/2}] = \sum_{i = 1}^q [bf_i + z f_i^{1/2}].$$      

Before proving Theorem \ref{thm:noisy_l1_grenander}, we state the needed lemma, which is adapted from Lemma 1 in \citet{birge1989grenander}. The proof also follows that of Lemma 1 in \citet{birge1989grenander}.
\begin{lemma}\label{lemma:mean_l1_grenander}
  Let $\mathcal{P} = \{J_i\}_{1 \leq i \leq q}$ be some partition of $J$, $F$ an absolutely continuous distribution function, $F_{n,m}$ the corresponding empirical c.d.f based on $\hat{s}_i$ and $\tilde{F}_{n,m} = \tilde{F}_{n,m}^J$, $\tilde{F}_{n,m}^i = \tilde{F}_{n,m}^{J_i}$ the related Grenander estimators defined on the associative intervals. Define
  $$\bar{F}_{n,m}(x) = \sum_{i = 1}^q \tilde{F}_{n,m}^i(x) \Ibb_{x \in J_i},$$
  with $f$ and $\bar{f}_{n,m}$ to be the respective derivatives of $F$ and $\bar{F}_{n,m}$. Then
  $$\Ebb\left [ \int_J |\tilde{f}_{n,m}(x) - f(x)| dx\right ] \leq \Ebb \left [ \int_J |\bar{f}_{n,m}(x) - f(x)| dx \right ].$$
\end{lemma}

In the proof, there are many similar notations that might be confusing. I list all of them below for clarity and the convenience of reference.
\begin{itemize}
\item Let $F^{(m)}(x) := \mathbb{P}(\hat{s}_i \leq x)$, i.e. the c.d.f of $\hat{s}_i$. Let $f^{(m)}$ be the derivative of $F_{n,m}$.
\item For any interval $J$, $F_{n,m}^J(x) = \frac{1}{n} \sum_{i=1}^n \Ibb(\hat{s}_i \leq x; \hat{s}_i \in J)$, $F_n^J(x) = \frac{1}{n} \sum_{i=1}^n \Ibb(s_i \leq x; s_i \in J)$. $F_{n,m}$ and $F_n$ correspond to $F_{n,m}^J(x)$ and $F_{n}^J(x)$ with $J = [0, 1]$ in our setup.
\item $\tilde{F}^J_{n,m}$ and $\tilde{F}^J_n$ are the respective least concave majorants of $F_{n,m}^J$ and $F_n^J$ condition on the interval $J$. Let $\tilde{f}^J_{n,m}$ and $\tilde{f}^J_n$ be the derivatives of $\tilde{F}^J_{n,m}$ and $\tilde{F}^J_n$. $\tilde{F}_{n,m}$, $\tilde{F}_n$ and $\tilde{f}_{n,m}$, $\tilde{f}_n$ correspond to $\tilde{F}_{n,m}^J(x)$, $\tilde{F}_{n}^J(x)$ and $\tilde{f}^J_{n,m}$, $\tilde{f}^J_{n}$.
\item For any partition $\mathcal{P} = \{J_i\}_{1 \leq i \leq q}$, let $\bar{f}_{n,m}$ be the derivative of $\bar{F}_{n,m}(x) = \sum_{i = 1}^q \Ibb(x \in J_i) \tilde{F}_{n,m}^{J_i}(x) .$
\item $l(J) = b - a$,  $f(J) = \int_J f(t) dt$,  $\Delta f(J) = f(b) - f(a)$, $bf(J) = \int_J |f(J)/l(J) - f(t)|dt$.
\item For any partition $\mathcal{P} = \{J_i\}_{1 \leq i \leq q}$, $L^J(\mathcal{P}, f, z) = \sum_{i = 1}^q [bf(J_i) + z(f(J_i))^{1/2}] = \sum_{i = 1}^q [bf_i + z f_i^{1/2}]$. $L^J(f, z) = \inf_{\mathcal{P}} L^J(\mathcal{P}, f, z)$.
\item $M := \int_0^b f^p(t) dt$ for some $p > 2$; $H:= \lim_{x \rightarrow b-}f(x)$.
\item For an interval $I := [a^I, b^I)$,
  \begin{itemize}
  \item Let $N$ and $N_m$ be the number of $s_i$'s and $\hat{s}_i$'s that fall in $I$, respectively.
  \item Define $G$ and $G^{(m)}$ to be the respective conditional c.d.f's of the $s_i$'s and $\hat{s}_i$'s that fall in $I$,  $g$ and $g^{(m)}$ their derivatives.
  \item Define $G_N(x) = \sum_{i = 1}^n \Ibb[s_i \leq x; s_i \in I]$,  $G_{N_m}(x) := \sum_{i = 1}^n \Ibb[\hat{s}_i \leq x; \hat{s}_i \in I]$.
  \item Define $\tilde{G}_{N_m, m}$ and $\tilde{G}_N$ be the respective least concave majorants of $G_{N_m}$ and $G_N$ conditional on $I$. Let $\tilde{g}_{N_m, m}$ and $\tilde{g}_{N}$ be derivatives of $\tilde{G}_{N_m, m}$ and $\tilde{G}_N$ respectively.
  \end{itemize}
\end{itemize}

\begin{proof}
  
  We want to show that
  \begin{equation*}
    \Ebb_f \left [ \int_J |f(x) - \tilde{f}_{n,m}(x)| dx \right ] \leq 3 L^J(f, \tilde{K}\cdot n^{-\frac{1}{2}}) + \tilde{C}\cdot f(J) \cdot m^{-\frac{1}{2}}, \label{eq:functional_upper_bound}
  \end{equation*}
  where $\tilde{K}, \tilde{C}$ are universal constants. Then, the $L_1$ convergence of $\tilde{f}_{n,m}$ only hinges on the characteristics of $f$. For example, when $f$ is a decreasing function on $J = [0, 1]$ such that $M := \int_0^b f^p(t) dt < + \infty$ for some $p > 2$ and $H = \lim_{x \rightarrow b-}f(x) > 0$, then Proposition 4 of \citet{birge1989grenander} shows that
  \begin{equation}
    z^{-2/3}L^J(f, z) \leq 3/2(H/(H-h))^{p-1}(bMH^{2-p}/(p - 2))^{1/3},
    \label{eq:functional_bound}
  \end{equation}
  where $h^3 = z^2b^{-2}MH^{2-p}/(p-2)$ and $Mz^2 < (p-2)b^2H^{p+1}$. It implies that $\tilde{f}_{n,m}$ has an $L_1$ convergence rate at $3L^J(f,  \tilde{K} \cdot n^{-\frac{1}{2}}) + \tilde{C}\cdot f(J) \cdot m^{-\frac{1}{2}} \leq C_1 \cdot n^{-\frac{1}{3}} + C_2 \cdot m^{-\frac{1}{2}}$, where $C_1$, $C_2$ are some positive constants.

  By Lemma \ref{lemma:mean_l1_grenander} it is sufficient to prove that for any partition $\mathcal{P} = \{J_i\}_{1 \leq i \leq q}$ of $J$, we have
  $$\Ebb_f \left [ \int_J |f(x) - \bar{f}_{n,m}(x) | dx\right ] \leq 3 \sum_{i = 1}^q [bf_i + \sqrt{f_i} \cdot \tilde{K}\cdot n^{-\frac{1}{2}}] + \sum_{i = 1}^q \tilde{C}\cdot f_i \cdot m^{-\frac{1}{2}},$$
  where $\bar{f}_{n,m}$ is the derivative of $\bar{F}_{n,m}(x) = \sum_{i = 1}^q \tilde{F}_{n,m}^{J_i}(x) \Ibb(x \in J_i)$. This is certainly true if for any arbitrary sub-interval $I = [a^I, b^I)$ of $J$, the below inequality holds
  $$\Ebb_f \left [ \int_I |f(x) - \tilde{f}_{n,m}^I|dx\right ] \leq 3 \left [ bf(I) + \sqrt{f(I)}\cdot \tilde{K} \cdot n^{-\frac{1}{2}}\right ]  +  \tilde{C}\cdot f(I) \cdot m^{-\frac{1}{2}}.$$
  In order to prove this inequality, we assume there are $N$ $s_i$'s and $N_m$ $\hat{s}_i$'s falling in the interval $I$ respectively. Here, $N$ has a binomial distribution $Binomial(n, f(I))$ and $N_m$ has a binomial distribution $Binomial(n, f^{(m)}(I))$, where $f^{(m)}$ is the derivative of the c.d.f $F^{(m)} = \Pbb[\hat{s}_i \leq x]$. Then with $\tilde{f}_{n,m} = \tilde{f}_{n,m}^I$,
  \begin{equation}
    \int_I |f(x) - \tilde{f}_{n,m}(x)| dx \leq bf(I) + |f(I) - N/n| + |N/n - N_m/n| + b\tilde{f}_{n,m}(I).
    \label{eq:l1_decomposition}
  \end{equation}
  The only difficulty comes from the last term. Define $G$ and $G^{(m)}$ to be the respective conditional c.d.f's of the $s_i$'s and $\hat{s}_i$'s that fall in $I$,  $g$ and $g^{(m)}$ their derivatives. Then
  $$\Ebb_{f^{(m)}} [b\tilde{f}_{n,m}(I)|N_m] = N_m/n \Ebb_{g^{(m)}} [b\tilde{g}_{N_m, m}(I) | N_m]$$
  because the joint distribution of the $N_m$ $\hat{s}_i$'s falling in $I$ given $N_m$ is the same as the distribution of $N_m$ i.i.d variables from $G^{(m)}$. If $U(x)$ is the uniform c.d.f on $I$, then
  \begin{eqnarray*}
    1/2b\tilde{g}_{N_m, m}(I) &\overset{(a)}{=}& \sup_{x \in I}[\tilde{G}_{N_m, m}(x) - U(x)]\\
    &\overset{(b)}{=}& \sup_{x \in I} [G_{N_m}(x) - U(x)]\\
    &\leq& \sup_{x \in I} [ G_{N_m}(x) - G(x)] + \sup_{x \in I} [G(x) - U(x)]\\
    &\leq& \sup_{x \in I} [G_{N_m}(x) - G(x)] + 1/2 bg(I).
  \end{eqnarray*}
  Here Equation $(a)$ holds because this is an equivalent expression of the total variation for $\tilde{G}_{N_m, m}(x)$ with a non-increasing derivative and $U(x)$ with a flat density. Equation $(b)$ holds because
  \begin{itemize}
  \item $\tilde{G}_{N_m, m}(x) \geq G_{N_m}(x)$ for any $x$ and the equality occurs when the derivative of $\tilde{G}_{N_m, m}$ changes.
  \item $\tilde{G}_{N_m,m}(x) - U(x)$ attains the maximum at a point which corresponds to a change of the derivative of $\tilde{G}_{N_m, m}$.
  \end{itemize}
  
  Since $bf(I) = f(I)bg(I)$, we get
  $$\Ebb_{f^{(m)}} [b\tilde{f}_{n,m}(I) | N_m] \leq N_m/n \left [ 2\Ebb_{g^{(m)}} [\sup_{x \in I} (G_{N_m(x) }- G(x)) | N_m] + bf(I)/f(I)\right ],$$
  and using Corollary \ref{corollary:noisy_DKW},
  $$\Ebb_{f^{(m)}} [b\tilde{f}_{n,m}(I) | N_m] \leq N_m/n \left [ K \cdot (\frac{1}{\sqrt{N_m}} + \frac{1}{\sqrt{m}}) + bf(I)/f(I)\right ],$$

  where $K > 1$. Plug in this result into the inequality \eqref{eq:l1_decomposition}, and with the Cauchy-Schwarz inequality, we have
  \begin{eqnarray*}
    && \Ebb_f \int_I |f(x) - \tilde{f}_{n,m}(x)|dx\\
    &\leq& bf(I) + \Ebb |f(I) - N/n|\\
    && + K\sqrt{f^{(m)}(I)/n} + K \cdot f^{(m)}(I)/\sqrt{m} + bf(I) \cdot f^{(m)}(I)/f(I)  + \Ebb |N/n - N_m/n|.
    \end{eqnarray*}
    By the Cauchy-Schwarz inequality, it follows that
    $$\Ebb |f(I) - N/n| \leq \sqrt{\frac{f(I)(1 - f(I))}{n}}.$$
    Note that 
  \begin{eqnarray*}
    f^{(m)}(I) &=&\Ebb N_m/n\\
               &=& \Ebb N/n + \Ebb [N_m - N]/n \\
               &\leq& f(I) + \Ebb |N_m - N|/n\\
               &=& f(I) + \Ebb[\Ebb[|N_m - N|/n|N]]\\
               &\overset{(c)}{\leq}& f(I) + \Ebb[(\frac{C}{\sqrt{m}}) \cdot N/n]\\
               &=& f(I) (1 + C \cdot m^{-\frac{1}{2}})
  \end{eqnarray*}
  where Inequality (c) holds because of Proposition \ref{prop:CDF_deviation} with $C > 0$ (note that $N_m/n = F_{n,m}(b^I) - F_{n,m}(a^I)$, and $N/n = F_n(b^I) - F_n(a^I)$). Thus, for $m \geq C^2/9$, it follows that
  $$\Ebb_f \int_I |f(x) - \tilde{f}_{n,m}(x)|dx \leq 3(bf(I) + K\sqrt{f(I)/n}) + \frac{2f(I)}{\sqrt{m}} \cdot (K + C).$$
  Then we complete the initial claim by letting $\tilde{K} = K$ and $\tilde{C} = 2(K+C)$. Finally, as Proposition 4 in \citet{birge1989grenander}, we construct the partition $\mathcal{P} = \{J_i\}_{j \leq i\leq q}$ of $J$ where $j$ is the integer such that $jh \leq H < (j+1)h$, $J_q = \{x | f(x) \geq q\}$, $J_j = \{x | f(x) < (j+1)h\}$, and $J_i = \{x | ih \leq f(x) < (i+1)h\}$ for $q > i > j$. Since $f_{\max} < \infty$, there is only finite number of intervals. It can be shown that when $q$ is the smallest integer that is larger than $f_{\max}$, this partition can give the inequality \eqref{eq:functional_bound}.
\end{proof}

\section{Proof of local asymptotics}\label{appendix:bmu_proof_noisy_grenander_local_asymp}

\subsection{Local Inference of $F_{n,m}$ when $F$ is absolutely continuous}
Theorem \ref{thm:noisy_DKW} shows that the empirical CDF $F_{n,m}$ is a consistent estimator of the population CDF. We also want to understand the uncertainty of the empirical CDF. The Koml{\' o}s-Major-Tusn{\' a}dy (KMT) approximation shows that $\sqrt{n}(F_n(x) - F(x))$ can be approximated by a sequence of Brownian bridges $\{B_n(x), 0 \leq x \leq 1\}$ \citep{komlos1975approximation}. This result can be extended to the empirical CDF based on $\hat{s}_i$'s; see Theorem \ref{thm:noisy_KMT_approx}. The proof is similar to Theorem \ref{thm:noisy_DKW} by splitting $F_{n,m}(x) - F(x)$ into $F_{n,m}(x) - F^{(m)}(x)$ and $F^{(m)}(x) - F(x)$. The former can be bounded by the original KMT approximation and the latter one can be bounded using Proposition \ref{prop:CDF_deviation}. Another thing needs taking care of is to bound $B_n(F(x)) - B_n(F^{(m)}(x))$, which is in fact the sum of two independent Gaussian random variables with zero mean and variance upper bounded by $\vert F^{(m)}(x) - F(x) \vert$.

\begin{theorem}[Local inference of $F_{n,m}$]\label{thm:noisy_KMT_approx}
     Suppose $F$ corresponds to a density $f$ on $[0, 1]$ with $f_{\max} < \infty$.  There exists a sequence of Brownian bridges $\{B_n(x), 0 \leq x \leq 1\}$ such that
  $$\Pbb\left \{ \sup_{0 \leq x \leq 1} |\sqrt{n}(F_{n,m}(x) - F(x)) - B_n(F(x))| > \frac{\sqrt{2\pi} f_{\max}\cdot \sqrt{n}}{\sqrt{m}} + \frac{a\log n}{\sqrt{n}} + t   \right \} \leq b (e^{-c\sqrt{n}t} + e^{-dmt^2}),$$
  for all positive integers $n$ and all $t > 0$, where $a$, $b$, $c$ and $d$ are positive constants.
\end{theorem}

\subsection{Proof of Theorem \ref{thm:noisy_grenander_local_asymp}}
This proof is adapted from \citet{wang1992nonparametric}. Define
$$U_{n,m}(a) = \sup \{x: F_{n,m}(x) - ax \text{ is maximal }\}.$$
Then with probability one, we have the switching relation
\begin{equation}\label{eq:switching_relation}
  \ftil_{n,m}(t) \leq a \Leftrightarrow U_{n,m}(a) \leq t.
\end{equation}

By the relation \eqref{eq:switching_relation}, we have
$$\Pbb(\sqrt{n}(\ftil_{n,m}(t_0) - f(t_0)) \leq x) = \Pbb(U_{n,m}(f(t_0) + n^{-\frac{1}{2}}x) \leq t_0).$$
From the definition of $U_{n,m}$,  it follows that
\begin{eqnarray*}
  U_{n,m}(f(t_0) + n^{-\frac{1}{2}}x) &=& \sup \{s: F_{n,m} (s) - (f(t_0) + n^{-\frac{1}{2}} x)s \text{ is maximal }\}\\
   &=&\sup \{s: \sqrt{n} (F_{n,m} (s) - F(s)) + \sqrt{n}(F(s) - f(t_0)) -  xs \text{ is maximal }\}
\end{eqnarray*}
By Theorem \ref{thm:noisy_KMT_approx},
$$\sqrt{n} (F_{n,m} - F(s)) = B_n(F(s)) + \Ocal(\frac{\sqrt{n}}{\sqrt{m}}) + \Ocal_p (\frac{\log n}{\sqrt{n}}),$$
where $\{B_n, n\in N\}$ is a sequence of Brownian Bridges, constructed on the same space as the $F_n$. So the limit distribution of $U_{n,m}(f(t_0) + n^{-\frac{1}{2}}x)$ is the same as that of the location of the maximum of the process $\{B_n(F(s)) + \sqrt{n}(F(s) - f(t_0)(s) - xs, s \geq 0\}$. Note that $F(s)$ is concave and linear in $[a, b]$, then
$$F(s) = F(a) + f(t_0)(s - a) \text{ for } s \in [a, b],$$
and
$$F(s) - f(t_0)(s - a) < F(a) \text{ for } s \not \in [a, b].$$
Hence the location of the maximum of $\{(B_n(F(s)) + \sqrt{n}(F(s) - f(t_0)s) - xs, s\geq 0\}$ behaves asymptotically as that of
$$\{B(F(s)) -xs, a \leq s \leq b\} = \{B(F(a) + f(t_0)(s-a)) - xs, a \leq s \leq b\},$$
where $B$ is a standard Brownian bridge in $[0, 1]$. Thus,
\begin{eqnarray*}
  &&\Pbb (\sqrt{n}(\ftil_{n,m}(t_0) - f(t_0)) \leq x) \rightarrow\\
  && \Pbb(\text{ the location of the maximum of } \{B(F(s)) - xs, a\leq s \leq b \leq t_0\})\\
  &=& \Pbb (\hat{S}_{a,b}(t_0) \leq x),
\end{eqnarray*}
by the definition of $\hat{S}$. That completes the proof of Part (A). The proof of Part (B) follows in a similar manner.

\section{Proof of Theorem \ref{thm:true_c0_c1_valid}}\label{appendix:bmu_proof_true_c0_c1_valid}
\begin{proof}
  Let $\hat{\alpha}_l(c_l)$, $\hat{\alpha}_r(c_r)$, $\hat{\alpha}_{mid}(c_l, c_r)$, $\tilde{g}_l(c_l)$, $\tilde{g}_r(c_r)$ be the output of Algorithm \ref{algo:gU_grid_search} with the input $c_l$, $c_r$ and $d_l$, $d_r$. The corresponding estimator of $f$ is termed as $\ftil_{n,m}$.
  For simplicity, we consider $c_r^{(0)} = 1$, i.e., the case where there is only the decreasing part and the flat part. For a general case where $c_r^{(0)} < 1$, we just need to focus on $[0, \mu]$.

  By Theorem \ref{thm:noisy_l1_grenander}, we know that $\Ebb_f \int_0^{c_l^{(0)}} |\ftil_{n,m}(x) - f(x)|dx \leq K_1 \cdot N_l(c_l^{(0)})^{-\frac{1}{3}}$  when $m \geq C_1 \cdot [N_l(c_l^{(0)})]^{\frac{2}{3}}$ for some positive constants $K_1$ and $C_1$ that only depend on $f$. For $c_l > c_l^{(0)}$, it is easy to see that $\lim_{x \rightarrow (c_l)_{-}}\ftil_{m,n}(x) - \lim_{x \rightarrow (c_l)_{+}}\ftil_{m,n}(x) = \varepsilon \cdot N_l(c_l^{(0)})^{-1/2}$, where $\varepsilon$ is a residual term with $\vert\varepsilon\vert \leq K_2$ for some positive constant $K_2$, because the estimator of the flat region converges at a square-root rate \citep{wang1992nonparametric}. In other words, it is unlikely to find a desired gap beyond $c_l^{(0)}$.

  If $\lim_{x \rightarrow (c_l^{(0)})_{-}} \ftil_{n,m}(x) \geq \frac{\hat{\alpha}_{mid}(c_l^{(0)}, 1)}{1 - c_l^{(0)}} + d_l$, we select the desired $c_l = c_l^{(0)}$. Otherwise, define
  $$ t \text{ be the maximal } c_l \text{ such that } \left \{
  \begin{array}{l}
   c_l < c_l^{(0)}\\
  \ftil_{n,m}(c_l) \geq \frac{\hat{\alpha}_{mid}(c_l^{(0)}, 1)}{1 - c_l^{(0)}} + d_l.\\
  \end{array}
  \right .
  $$
  Then $\forall t < c_l \leq c_l^{(0)}$, it follows that
  \begin{eqnarray*}
    f(c_l) - \ftil_{n,m}(c_l) &=& f(c_l) - \lim_{x \rightarrow (c_l^{(0)})_{+}}f(x) - (\ftil_{n,m}(c_l) - \frac{\hat{\alpha}_{mid}(c_l^{(0)}, 1)}{1 - c_l^{(0)}})\\
    && + \lim_{x \rightarrow (c_l^{(0)})_{+}}f(x) - (\frac{\hat{\alpha}_{mid}(c_l^{(0)}, 1)}{1 - c_l^{(0)}})\\
                              &\geq& \delta_l - d_l - K_2 \cdot (N_{mid}(c_l^{(0)}, 1)^{-\frac{1}{2}}),
  \end{eqnarray*}
 When $d_l < \delta_l$, it implies that
$$\Ebb_f \int_{t}^{c_l^{(0)}} \vert f(x) - \ftil_{n,m}(x) \vert dx \geq (\delta_l - d_l  - K_2 \cdot N_{mid}(c_l^{(0)}, 1)^{-\frac{1}{2}}) \cdot (c_l^{(0)} - t).$$
  Since the $L_1$ distance between $f$ and $\ftil_{n,m}$ reduces at a cubic-root rate, it follows that $c_l^{(0)} - t \leq C_2 \cdot N_l(c_l^{(0)})^{-1/3}$ for some positive constant $C_2$. So $c_l = t$ is the desired cutoff.
  Finally, we have that
  \begin{eqnarray*}
    \int_0^1 |\ftil_{n,m}(x) - f(x)|dx &=& \int_0^{t} |\ftil_{n,m}(x) - f(x)|dx + \int_t^{c_l^{(0)}} |\frac{\hat{\alpha}_{mid}(t, 1)}{1 - t} - f(x)|dx\\
    && + \int_{c_l^{(0)}}^1 |\frac{\hat{\alpha}_{mid}(t, 1)}{1 - t} - f(x)|dx\\
                                       &\leq& \int_0^{t} |\ftil_{n,m}(x) - f(x)|dx + (\frac{1}{1-t} + f_{\max}) \cdot |c_l^{(0)} - t|\\
    && + |\frac{\hat{\alpha}_{mid}(t, 1)}{1 - t} - \frac{\alpha_{mid}(c_l^{(0)}, 1)}{1 - c_l^{(0)}}| \cdot |1 - c_l^{(0)}|\\
    &\leq& C_4 \cdot N_l(c_l^{(0)})^{-1/3},
  \end{eqnarray*}
  where the last inequality holds because $|c_l^{(0)} - t| \leq C_2 \cdot  N_l(c_l^{(0)})^{-1/3}$ and $f$ is bounded can imply  $|\frac{\hat{\alpha}_{mid}(t, 1)}{1 - t} - \frac{\alpha_{mid}(c_l^{(0)}, 1)}{1 - c_l^{(0)}}| \leq K_3 \cdot N_l(c_l^{(0)})^{-1/3}$. Here $C_4$, $K_3$ are two positive constants that only depend on $f$.
\end{proof}

\section{Proof of Theorem \ref{thm:feasible_c0_c1_CI}}\label{appendix:bmu_proof_feasible_c0_c1_CI}
\begin{proof}
  Given the interior point $\mu$ in the flat region, the left decreasing part and the right increasing part are disentangled. Therefore, we only need to consider the left side, and the right side can be proven in the same way.  
  A necessary condition for $c_l$ being identified as feasible for the change-point-gap constraint in Algorithm \ref{algo:gU_grid_search} is that
  $$\tilde{g}_l(c_l) \geq \frac{\hat{\alpha}_{mid}(c_l, \mu)}{\hat{\alpha}_l(\mu) \cdot (\mu - c_l)} + \frac{d_l}{\hat{\alpha}_l(\mu)},$$
  where $\hat{\alpha}_{mid}(c_l, \mu) = N_{mid}(c_l, \mu)/n$ and $\hat{\alpha}_l(\mu) = N_l(\mu)/n$. It is easy to see that 
  $$\hat{\alpha}_l(\mu) = \alpha_l(\mu) + \varepsilon_1 \cdot n^{-\frac{1}{2}},$$
  and
  $$\hat{\alpha}_{mid}(c_l, \mu) = \alpha_{mid}(c_l, \mu) + \varepsilon_2 \cdot n^{-\frac{1}{2}},$$
  where $\alpha_l(\mu) = \Ebb\hat{\alpha}_l(\mu)$ and $\alpha_{mid}(c_l, \mu) = \Ebb\hat{\alpha}_{mid}(c_l, \mu)$; $\varepsilon_1$ and $\varepsilon_2$ are residual terms with $\max (\vert \varepsilon_1 \vert, \vert \varepsilon_2 \vert) \leq K$  for some universal positive constant $K$ by McDiarmid inequality. So the necessary condition is
  \begin{equation}\label{eq:necessary_change_point_gap}
    \tilde{g}_l(c_l) \geq \frac{\alpha_{mid}(c_l, \mu)}{\alpha_l(\mu) \cdot (\mu - c_l)} + \frac{d_l}{\alpha_l(\mu)} + \varepsilon_3 \cdot n^{-\frac{1}{2}},
  \end{equation}
  with $\vert \varepsilon_3 \vert \leq C$ for some constant $C$ that only depends on $\alpha_l(\mu)$, $\mu - c_l$ and $d_l$. If $c_l > c_l^{(0)}$ and the constraint is violated at $c_l$, then for any $c_l' > c_l$ the constraint is violated at $c_l'$ with high probability since $\tilde{g}_l(c_l) \geq \tilde{g}_l(c_l')$ and $\frac{\alpha_{mid}(c_l, \mu)}{\alpha_l(\mu) \cdot (\mu - c_l)} = \frac{\alpha_{mid}(c_l', \mu)}{\alpha_l(\mu) \cdot (\mu - c_l')}$ ($c_l$ and $c_l'$ are both in the flat region). Therefore, to see if $\hat{c}_l > c_l^{(0)}$, we only need to investigate the smallest $c_l$ with $c_l > c_l^{(0)}$ that is in the searching space of Algorithm \ref{algo:gU_grid_search}. By Theorem \ref{thm:noisy_grenander_local_asymp}, when $m/\cdot N_l(c_l^{(0)}) \rightarrow \infty$, it follows that
  $$\sqrt{N_l(\mu)}(\tilde{g}_l(c_l) - \frac{\alpha_{mid}(c_l, \mu)}{\alpha_l(\mu) \cdot (\mu - c_l)}) \overset{d}{\rightarrow} \hat{S}_{[c_l^{(0)}, \mu]}(c_l).$$
  By the necessary condition \eqref{eq:necessary_change_point_gap}, then asymptotically we have
  \begin{eqnarray*}
    \Pbb[c_l > c_l^{(0)}] &\leq& \Pbb \left (\sqrt{N_l(\mu)}(\tilde{g}_l(c_l) - \frac{\alpha_{mid}(c_l, \mu)}{\alpha_l(\mu) \cdot (\mu - c_l)})  \geq \sqrt{N_l(\mu)} \cdot \frac{d_l}{\alpha_l(\mu)} + \varepsilon_3 \cdot \sqrt{\frac{N_l(\mu)}{n}}\right ) \\
    &\approx& \Pbb \left (\hat{S}_{c_l^{(0)}, \mu}(c_l)  \geq \sqrt{N_l(\mu)} \cdot \frac{d_l}{\alpha_l(\mu)} + \varepsilon_3 \cdot \sqrt{\frac{N_l(\mu)}{n}}\right )\\
                          &\leq&  \Pbb \left (\hat{S}_{c_l^{(0)}, \mu}(c_l)  \geq \sqrt{N_l(c_l^{(0)})} \cdot \frac{d_l}{\alpha_l(\mu)} - C \cdot \sqrt{\frac{N_l(\mu)}{n}}\right ).\\
                          &\leq&  \Pbb \left (\hat{S}_{c_l^{(0)}, \mu}(c_l)  \geq \sqrt{N_l(c_l^{(0)})} \cdot \frac{d_l}{\alpha_l(\mu)} - C \right ).\\
                          &\approx&  \Pbb \left (\hat{S}_{c_l^{(0)}, \mu}(c_l^{(0)})  \geq \sqrt{N_l(c_l^{(0)})} \cdot \frac{d_l}{\alpha_l(\mu)} - C \right ).
  \end{eqnarray*}
  where the last approximation holds because $c_l$ is the smallest one searching candidate with $c_l > c_l^{(0)}$, so it is very close to $c_l^{(0)}$.
\end{proof}

\section{Proof of Theorem \ref{thm:c0_c1_converge_finite_sample}}\label{appendix:bmu_proof_c0_c1_converge_finite_sample}
\begin{proof}
  Let $\hat{\alpha}_l(c_l)$, $\hat{\alpha}_r(c_r)$, $\hat{\alpha}_{mid}(c_l, c_r)$, $\tilde{g}_l(c_l)$, $\tilde{g}_r(c_r)$ be the output of Algorithm \ref{algo:gU_grid_search} with the input $c_l$, $c_r$ and $d_l$, $d_r$. The corresponding estimator of $f$ is termed as $\ftil_{n,m}$. Let $H(c_l, c_r, \ftil_{n,m})$ be the log likelihood function associated with the optimization problem \eqref{opt:gU_mle_known_c0_c1_emp_mass}, and $N(x) = \# \{\hat{s}_i: \hat{s}_i = x\}$. It follows that
  \begin{eqnarray*}
    \frac{1}{n}\Ebb_f H(c_l, c_r, \ftil_{n,m}) &=& \frac{1}{n} \Ebb_f \sum_{\hat{s}_i} \log \ftil_{n,m}(\hat{s}_i)\\
    &=& \Ebb_f \log \ftil_{n,m}(\hat{s}_1)\\
    &=& -KL(f||\ftil_{n,m}) + C,
  \end{eqnarray*}
  where $C = \Ebb_f \log f(\hat{s}_1)$. From the relations between total variation, Kullback-Leibler divergence and the $\chi^2$ distance,
  $$TV(P, Q) \leq \sqrt{KL(P || Q)} \leq \sqrt{\chi^2(P||Q)},$$ we have
  \begin{eqnarray*}
    \frac{1}{n}\Ebb_f H(c_l, c_r, \ftil_{n,m}) = -\Ebb_f KL(f || \ftil_{n,m}) + C &\leq& -\Ebb_f(\int_0^1 |f(x) - \ftil_{n,m}(x)|dx)^2 + C\\
    &\leq& -(\Ebb_f \int_0^1 |f(x) - \ftil_{n,m}(x)|dx)^2 + C,
  \end{eqnarray*}
  where the last inequality uses the Jensen's inequality. On the other hand, it follows that
  \begin{eqnarray*}
    \frac{1}{n}\Ebb_f H(c_l, c_r, \ftil_{n,m}) &=& -\Ebb_f KL(f || \ftil_{n,m}) + C \\
    &\geq& -\Ebb_f \int_0^1 (f(x) - \ftil_{n,m}(x))^2/f(x)dx + C\\
    &\geq& -\Ebb_f \int_0^1 (f(x) - \ftil_{n,m}(x))^2/f_{min}dx + C\\
    &\geq& -\frac{1}{f_{min}} \cdot (\Ebb_f\int_0^1 |f(x) - \ftil_{n,m}(x)|dx)^2 + C,
  \end{eqnarray*}
  Then the problem is reduced to bound $\Ebb_f\int_0^1 |f(x) - \ftil_{n,m}(x)|dx$. From Theorem \ref{thm:true_c0_c1_valid}, we know that when $m \geq C_1 \cdot \left (\max (N_l(c_l^{(0)}), N_r(c_r^{(0)})) \right )^{\frac{2}{3}}$, there exist $c_l$ and $c_r$ in the neighborhoods of $c_l^{(0)}$ and $c_r^{(0)}$ respectively, such that the resulting estimator $\ftil_{n,m}$ satisfies $\Ebb_f\int_0^1 |f(x) - \ftil_{n,m}(x)|dx \leq K_1 \cdot (N_l(c_l^{(0)})^{-1/3} + N_r(c_r^{(0)})^{-1/3})$ for some positive constants $C_1$ and $K_1$. Along with Lemma \ref{lemma:l1_delta01}, we conclude the desired result.
  ~\\
\end{proof}

\begin{lemma}\label{lemma:l1_delta01}
  Let $\ftil_{n,m}$ be the solution by Algorithm \ref{algo:gU_grid_search} with input $c_l$ and $c_r$ and the corresponding $flag = \textit{True}$. Assume $f_{\max} < \infty$ and $f_{\min} > 0$. If $m \geq C_1 \cdot \left (\max (N_l(c_l^{(0)}), N_r(c_r^{(0)})) \right )^{\frac{2}{3}}$, then $|\Delta_l| \leq C_2 \cdot N_l(c_l^{(0)})^{-1/3}, |\Delta_r| \leq C_3 \cdot N_r(c_r^{(0)})^{-1/3}$ is a necessary condition for
  $$\Ebb_f \int_0^1 |f(x) - \ftil_{n,m}(x)|dx \leq C_4 \cdot (N_l(c_l^{(0)})^{-1/3} + N_r(c_r^{(0)})^{-1/3}),$$
  where $C_1, C_2, C_3, C_4$ are four constants depending on $d_l$, $d_r$, $f_{\max}$ and $f_{\min}$; $\Delta_l = c_l - c_l^{(0)}$, $\Delta_r = c_r - c_r^{(0)}$.
\end{lemma}
\begin{proof}
  For simplicity, we consider $c_r^{(0)} = 1$, i.e., the case where there is only the decreasing part and the flat part. For a general case where $c_r^{(0)} < 1$, we just need to focus on $[0, \mu]$. If $c_l < c_l^{(0)}$, the $L_1$ distance between $\ftil_{n,m}$ and $f$ is
  \begin{eqnarray*}
    && \Ebb_f \int_0^1 |\ftil_{n,m}(x) - f(x)|dx\\
    &=& \Ebb_f \int_0^{c_l^{(0)}} |\ftil_{n,m}(x) - f(x)|dx  + \Ebb_f \int_{c_l^{(0)}}^1 |\ftil_{n,m}(x) - f(x)|dx \\
    &\overset{(a)}{\geq}& -K_1 \cdot N_l(c_l)^{-1/3} + \Ebb_f \int_{c_l^{(0)}}^1 |\ftil_{n,m}(x) - f(x)|dx\\
    &\overset{(b)}{=}& -K_1 \cdot N_l(c_l)^{-1/3} + \Ebb_f |\frac{1 - \int_0^{c_l} \ftil_{n,m}(x) dx}{1 - c_l} - \frac{1 - \int_0^{c_l^{(0)}} f(x) dx}{1 - c_l^{(0)}}|(1-c_l^{(0)})\\
    &\overset{(c)}{\geq}& -K_2 \cdot N_l(c_l^{(0)})^{-1/3} + |\frac{1 - \int_0^{c_l} f(x) dx}{1 - c_l} - \frac{1 - \int_0^{c_l^{(0)}} f(x) dx}{1 - c_l^{(0)}}|(1-c_l^{(0)})\\
    &=& -K_2 \cdot N_l(c_l^{(0)})^{-1/3} + |\frac{-\Delta_l + (1 - c_l) \int_{c_l}^{c_l^{(0)}} f(x) dx + \Delta_l \int_0^{c_l^{(0)}} f(x) dx}{1 - c_l}|\\
    &=& -K_2 \cdot N_l(c_l^{(0)})^{-1/3} + |\frac{-\Delta_l - (1 - c_l) \Delta_l \gamma + \Delta_l \int_0^{c_l^{(0)}} f(x) dx}{1 - c_l}|\\
    &=& -K_2 \cdot N_l(c_l^{(0)})^{-1/3} + \kappa|\Delta_l|,
  \end{eqnarray*}
  where $K_1$ and $K_2$ are two positive constants that only depend on $f$, $\min_{x \in [c_l, c_l^{(0)}]} f(x) \leq \gamma \leq \max_{x \in [c_l, c_l^{(0)}]} f(x)$, $\kappa = |\frac{1 + (1 - c_l) \gamma - \int_0^{c_l^{(0)}} f(x) dx}{1 - c_l}| < \infty$ since $f_{\max} < \infty$. The inequality (a) and the equation (c) are obtained using Theorem \ref{thm:noisy_l1_grenander}. The equation (b) makes use of assumption that the right hand side is a flat region. Then, if there exists $C_4 > 0$ such that $\Ebb_f \int_0^1 |\ftil_{n,m}(x) - f(x)|dx \leq C_4 \cdot N_l(c_l^{(0)})^{-1/3}$, then $|\Delta_l| \leq C_2 \cdot N_l(c_l^{(0)})^{-1/3}$, for some positive constant $C_2$.

  Next, we investigate the case when $c_l > c_l^{(0)}$. Denote $a := \lim_{x \rightarrow (c_l)_{-}} \ftil_{n,m}(x)$, $b := \lim_{x \rightarrow (c_l)_{+}}\ftil_{n,m}(x)$, $c := \int_{c_l^{(0)}}^{c_l} \ftil_{n,m}(x) dx$, and $\alpha \equiv f(x) (\text{when}~x > c_l^{(0)})$. It is easy to check that
$$|c - \alpha| \cdot |\Delta_l| \leq \int_{c_l^{(0)}}^{c_l} |\ftil_{n,m}(x) - \alpha| dx.$$
  Using the $L_1$ convergence of the Grenander estimator, it follows that 
  $$|\Delta_l| |c - \alpha| + (1-c_l) |b - \alpha| \leq \int_{c_l^{(0)}}^{c_l} |\ftil_{n,m}(x) - \alpha| dx + (1 - c_l)|b - \alpha| \leq \int_0^1 |\ftil_{n,m}(x) - f(x)| dx.$$
  If there exists $C_4 > 0$ such that $\int_0^1 |\ftil_{n,m}(x) - f(x)| dx  \leq C_4 \cdot N_l(c_l^{(0)})^{-1/3}$,  we have
  \begin{eqnarray*}
    \frac{|\Delta_l|}{1 - c_l^{(0)}} |a - b| &\leq& \frac{|\Delta_l|}{1 - c_l^{(0)}} |c - b|\\
    &\leq& \frac{|\Delta_l|}{1 - c_l^{(0)}} |c - \alpha| + \frac{|\Delta_l|}{1 - c_l^{(0)}} |\alpha - b|\\
    &\leq& \frac{|\Delta_l|}{1 - c_l} |c - \alpha| + |\alpha - b| \leq \frac{C_4}{1 - c_l} \cdot N_l(c_l^{(0)})^{-1/3}.
  \end{eqnarray*}
  If the output $flag$ of Algorithm \ref{algo:gU_grid_search} is \textit{True}, $a - b \geq \frac{\alpha_{mid}(c_l, \mu)}{\alpha_l(\mu) (1 - c_l)} + \frac{d_l}{\alpha_l(\mu)} + K_3 \cdot n^{-1/2}$ for some positive constant $K_3$. Then it must follow that $|\Delta_l| \leq C_2 \cdot N_l(c_l^{(0)})^{-1/3}$, where $C_2 > 0$. So far, we have proven the lemma for the left hand side. For the right hand side, it can be proven similarly.
\end{proof}

\section{More simulation studies}\label{sec:bmu_more_simulation}
Besides the linear valley model specified in Section \ref{sec:bmu_sim_data_generation}, we also consider a non-linear model and a misspecified model.

For the non-linear model, we replace the left linear part in the linear valley model with an unnormalized decreasing function $f_l = Beta(x/c_l; 0.5, 1.5)/c_l \cdot 3/20$, $x \in [0, c_l]$. We replace the right linear part with an unnormalized increasing function $f_r = Beta(x/(1 - c_r); 2, 0.8)/(1-c_r)\cdot 1/20$, $x \in (c_r, 1]$. Here $Beta(x; \alpha, \beta)$ is a density of Beta distribution with parameters $\alpha$ and $\beta$. In this case, we use $m = 10^4$. Applying \GFCMethod\/ to the synthetic data generated by this model, we observe similar results as those from the linear valley model; see Figure \ref{fig:general_convergence}\ref{fig:general_sensitivity}\ref{fig:general_comp_rlt}. It indicates that \GFCMethod\/ can detect the cutoff in a satisfying range as long as the underlying model satisfies Assumption \ref{asmpt:U_shape} and Assumption \ref{asmpt:density_gap}. \GFCMethod\/ is particularly useful when the middle part is ``uniform'' and not easy to tell apart the samples of the alternative distribution from those of the null distribution.

For the misspecified model, we replace the left linear part with an unnormalized unimodal function $f_l = Beta(x/c_l; 1.5, 5)/c_l \cdot 3$, $x \in [0, c_l]$. We replace the right linear part with an unnormalized unimodal function $f_r = Beta(x/(1 - c_r); 2.5, 1.5)/(1-c_r)$, $x \in (c_r, 1]$. The results can be found in Figure \ref{fig:unimodal_convergence},\ref{fig:unimodal_sensitivity},\ref{fig:unimodal_comp_rlt}. We observe that the estimated $c_r$ is not satisfying until $m$ attains $10^4$. So in the other experiments, we use $m = 10^4$. We find that although the variance of $\hat{c}_r$ becomes larger than the estimate for the correctly specified model, the detected cutoff tends to be larger than the truth. It implies that if the model does not align with Assumption \ref{asmpt:U_shape} and Assumption \ref{asmpt:density_gap}, \GFCMethod\/ is still useful because it is conservative.


\begin{figure}
  \begin{minipage}{0.33\textwidth}
    \centering
    \includegraphics[width=\textwidth]{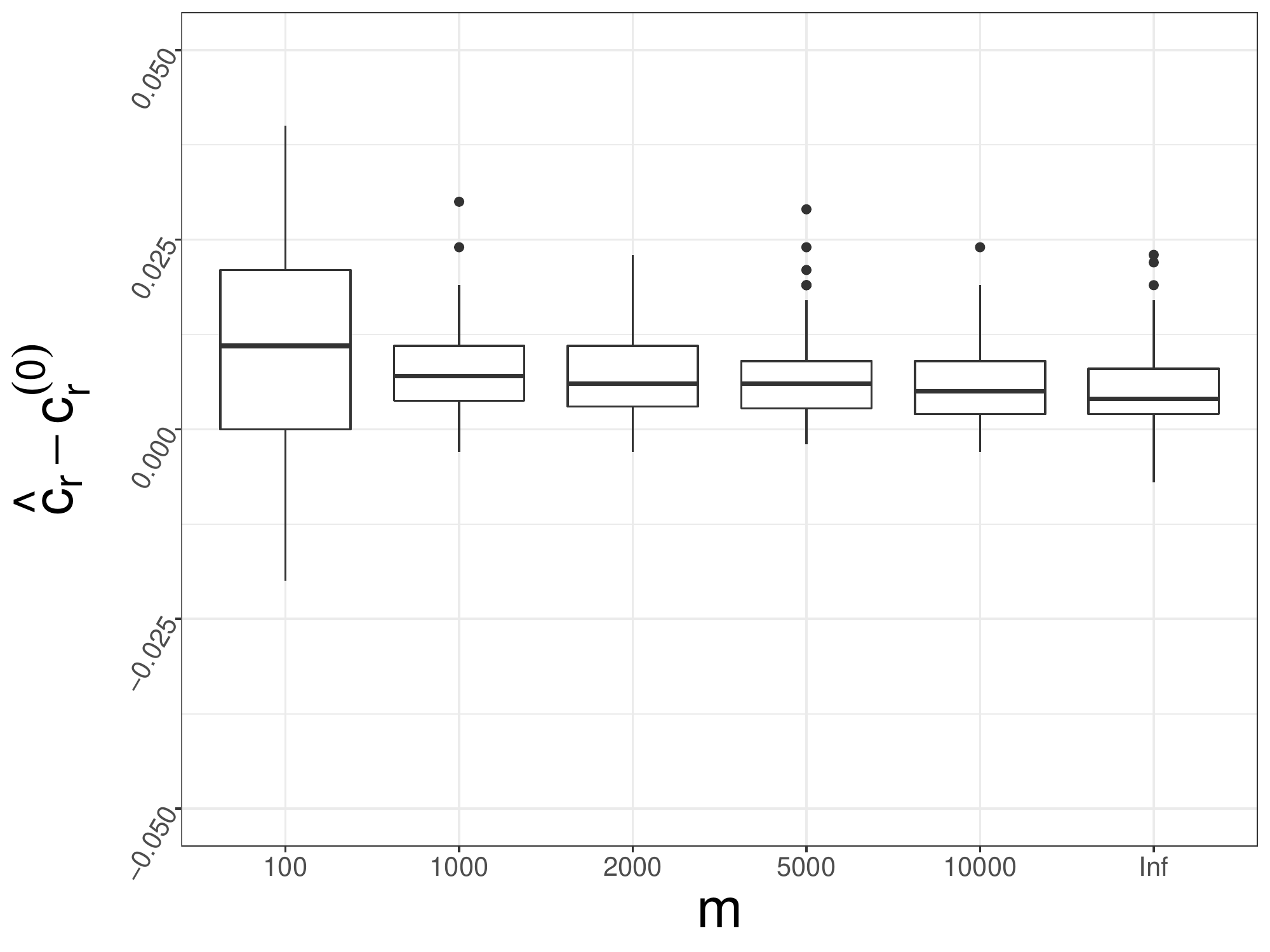}
    \subcaption{}
  \end{minipage}
  \begin{minipage}{0.33\textwidth}
    \centering
    \includegraphics[width=\textwidth]{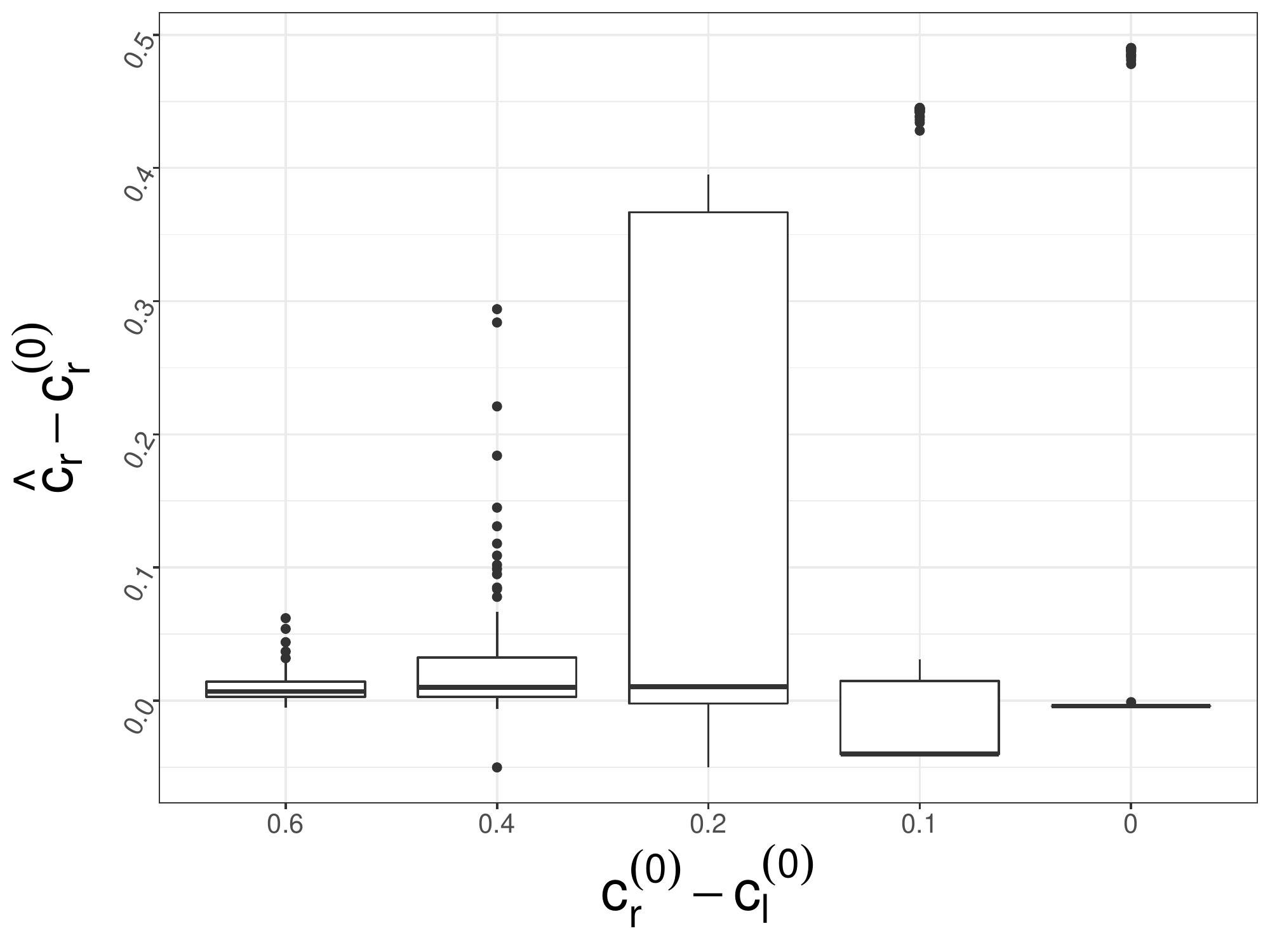}
    \subcaption{}
  \end{minipage}
    \begin{minipage}{0.33\textwidth}
      \centering
      \includegraphics[width=\textwidth]{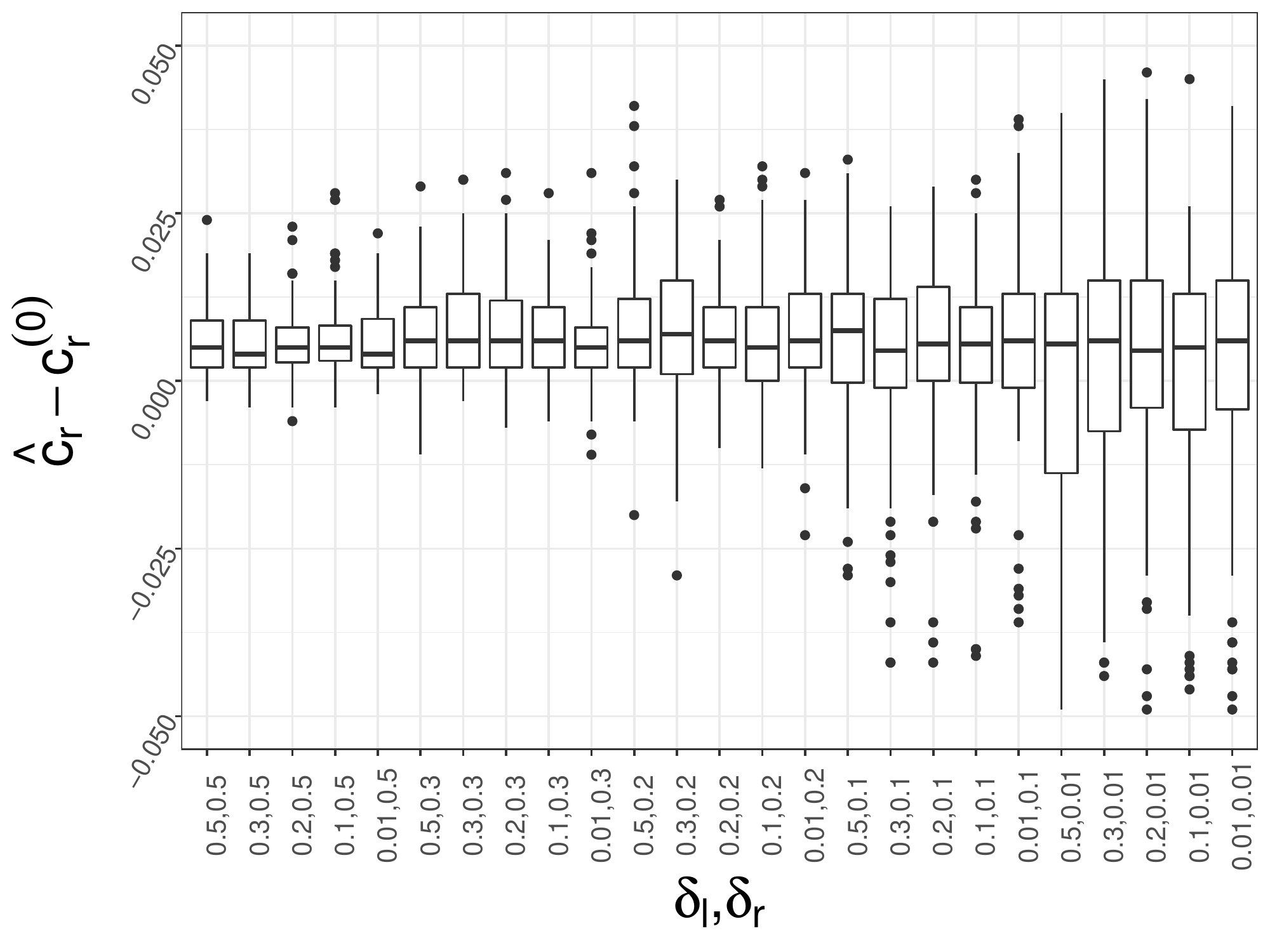}
    \subcaption{}
  \end{minipage}
  \caption{The estimation of $\hat{c}_r$ under the non-linear decreasing-uniform-increasing model. (a) with respect to $m$; (b) with respect to the width of the middle flat region; (c) with respect to the gap sizes.}\label{fig:general_convergence}
\end{figure}




\begin{figure}
  \begin{minipage}{0.49\textwidth}
    \centering
    \includegraphics[width=\textwidth]{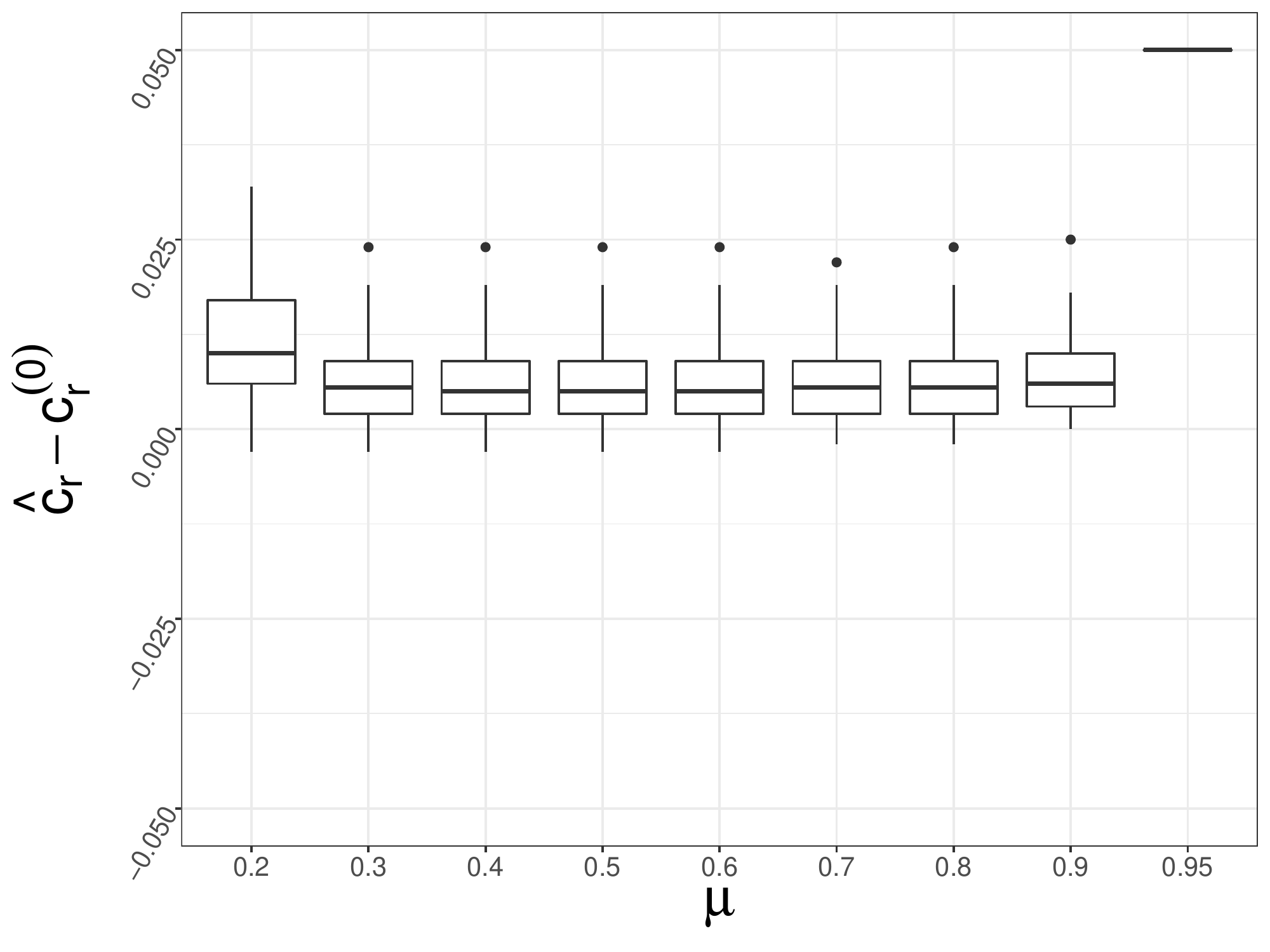}
    \subcaption{}
  \end{minipage}
  \begin{minipage}{0.49\textwidth}
    \centering
  \includegraphics[width=\textwidth]{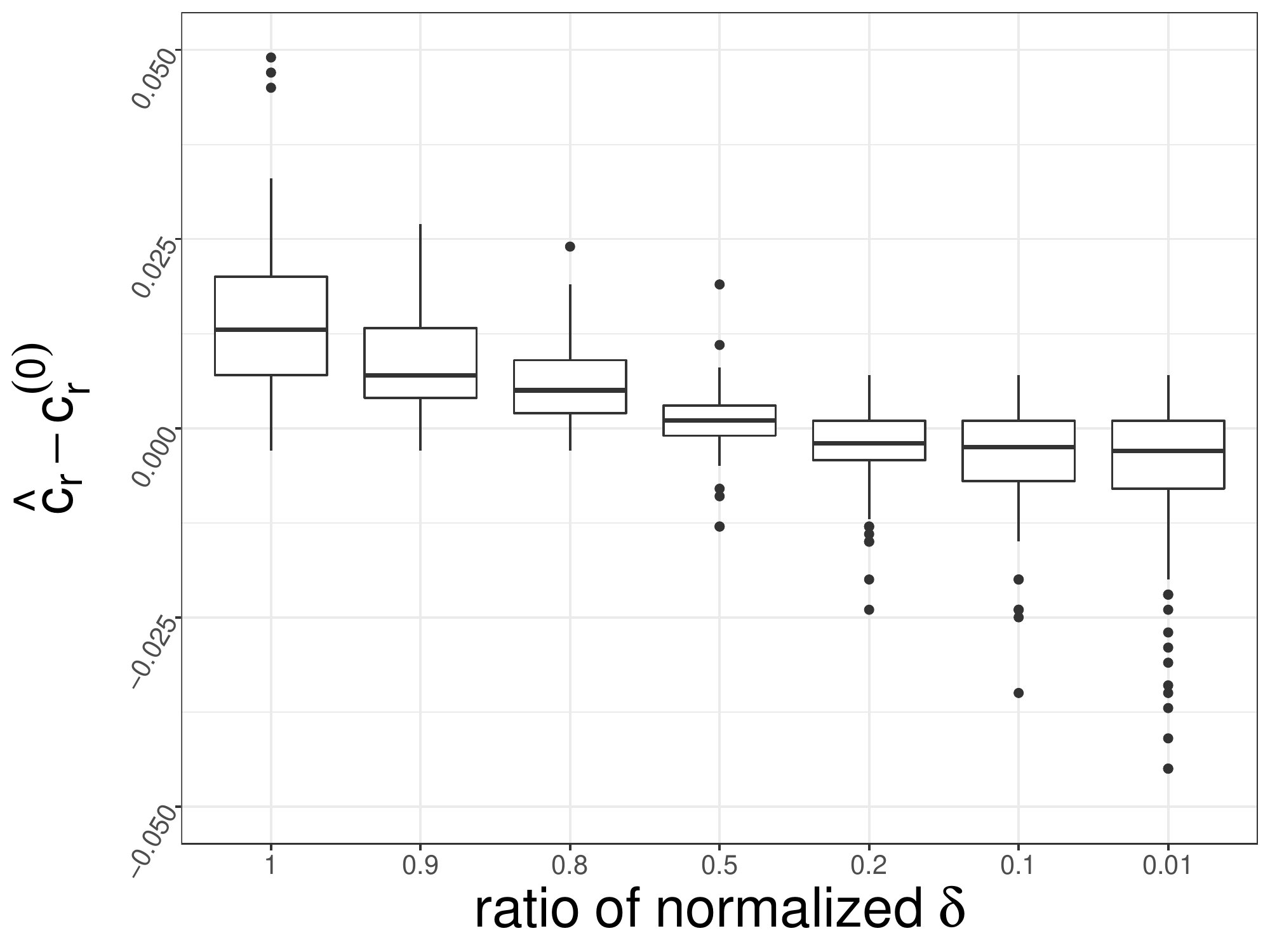}      
    \subcaption{}
  \end{minipage}
      \caption{The estimation of $\hat{c}_r$ under the non-linear decreasing-uniform-increasing model. (a) with respect to the choice of the middle point $\mu$; (b) with respect to the choice of the input $d_l$ and $d_r$. Here $d_l = \kappa \times \tilde{\delta}_l$ and $d_r = \kappa \times \tilde{\delta}_r$, where $\kappa$ is a ratio of the normalized $\delta$'s.}\label{fig:general_sensitivity}
\end{figure}



\begin{figure}[ht]
  \centering
  \includegraphics[width=0.6\textwidth]{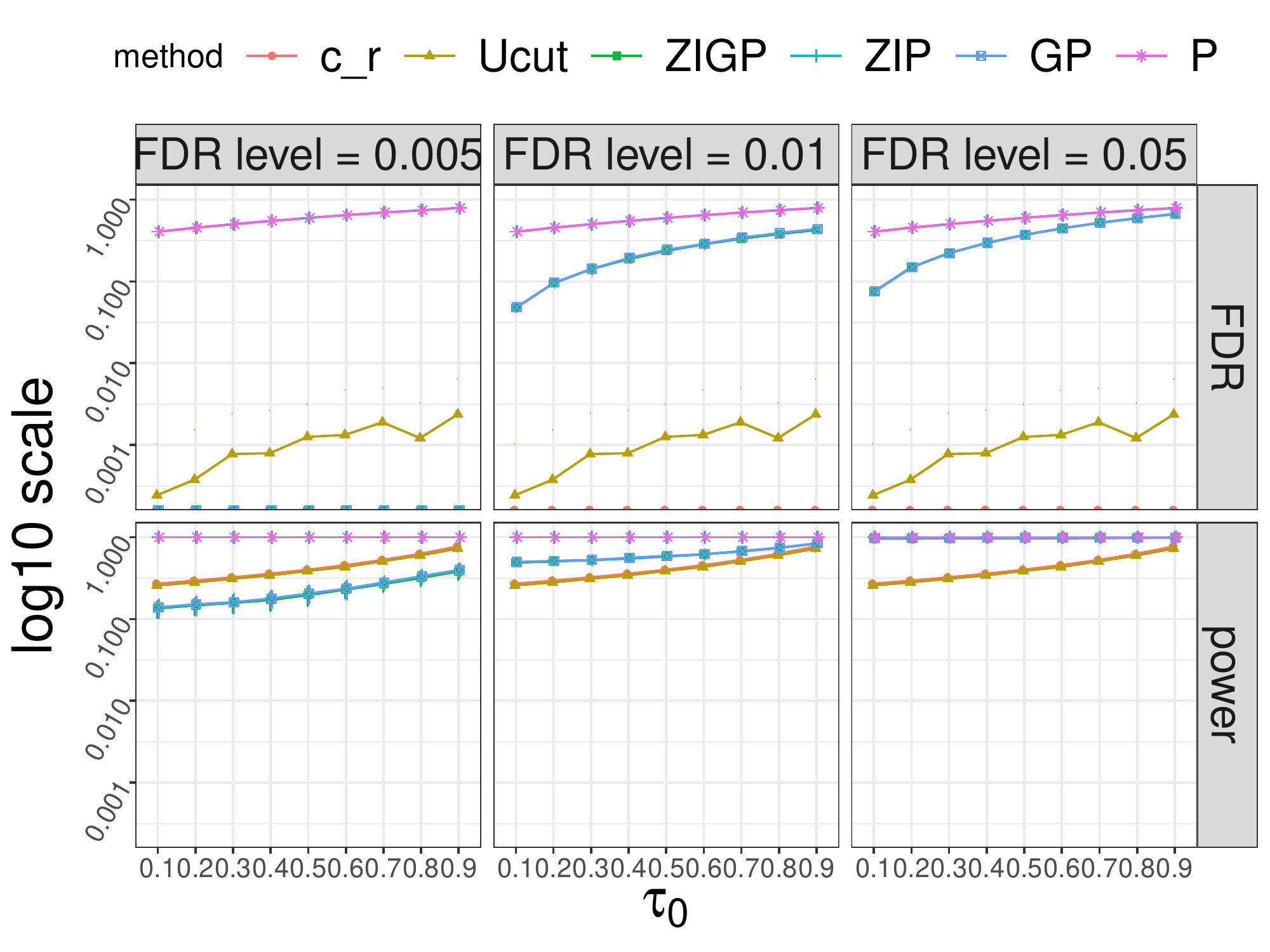}
  \caption{FDR and power of Algorithm \ref{algo:gU_grid_search} and other competing methods on the non-linear decreasing-uniform-increasing model.}\label{fig:general_comp_rlt}
\end{figure}


\begin{figure}
  \begin{minipage}{0.33\textwidth}
    \centering
    \includegraphics[width=\textwidth]{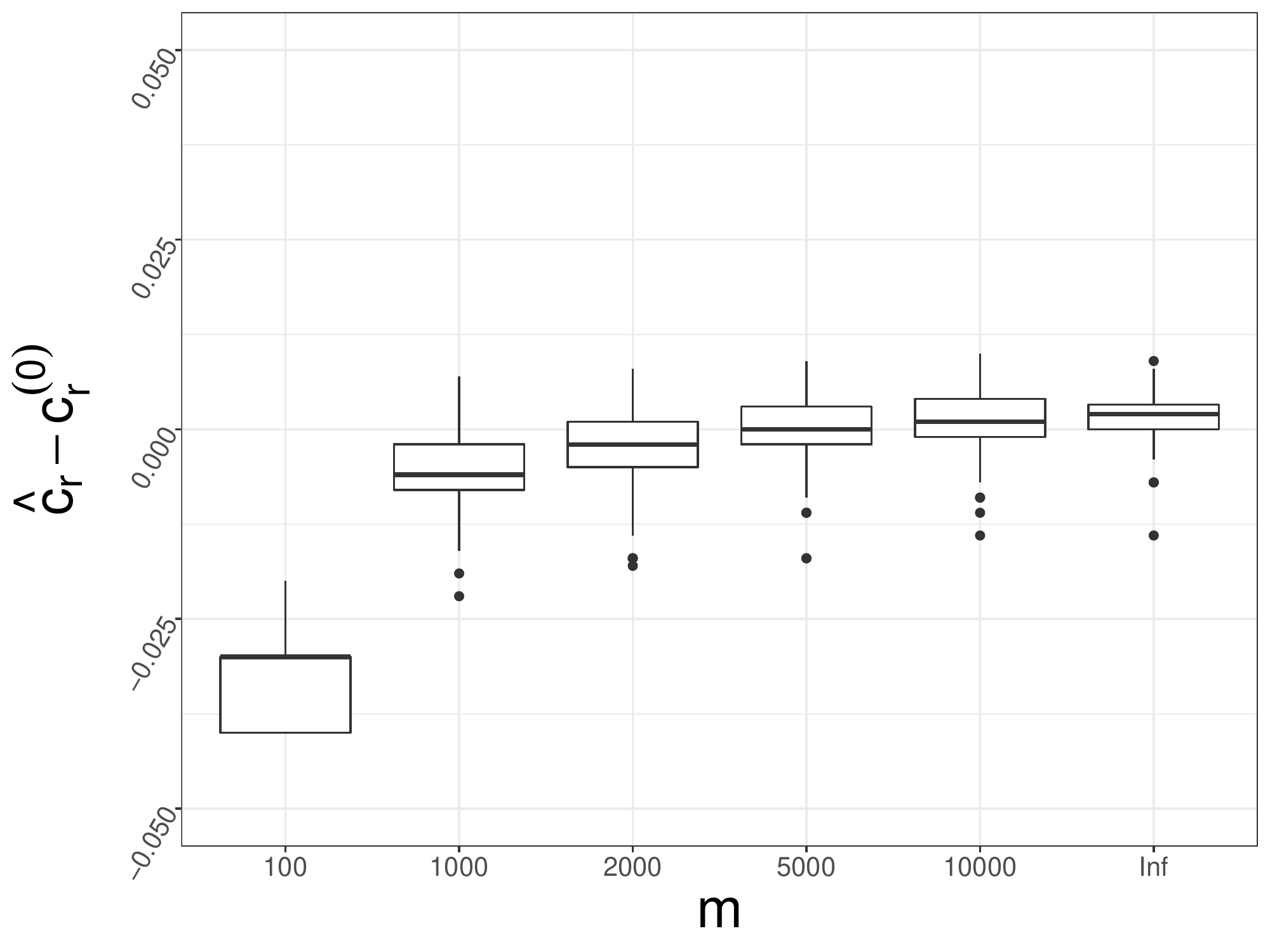}
    \subcaption{}
  \end{minipage}
  \begin{minipage}{0.33\textwidth}
    \centering
    \includegraphics[width=\textwidth]{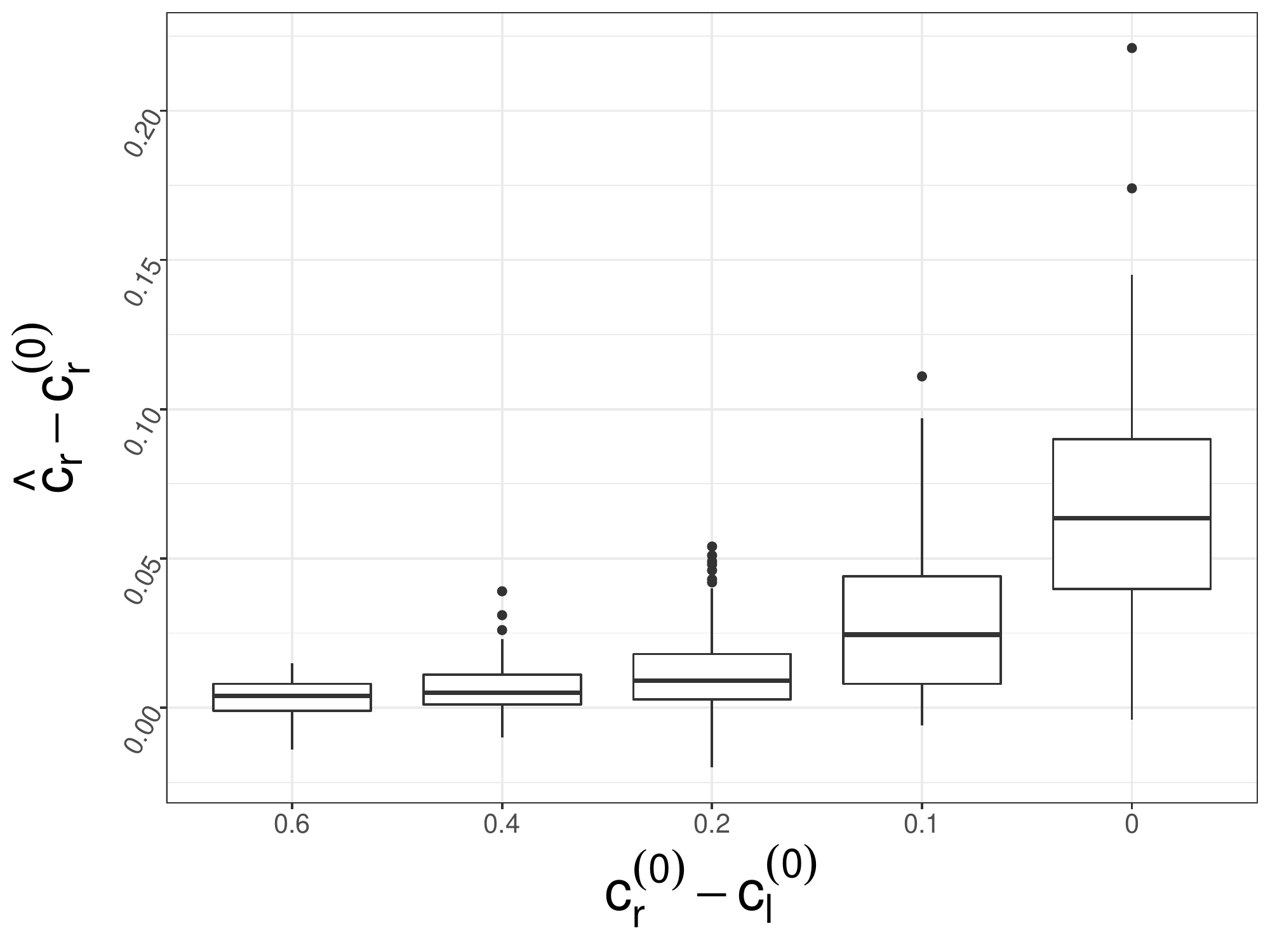}
    \subcaption{}
  \end{minipage}
    \begin{minipage}{0.33\textwidth}
      \centering
      \includegraphics[width=\textwidth]{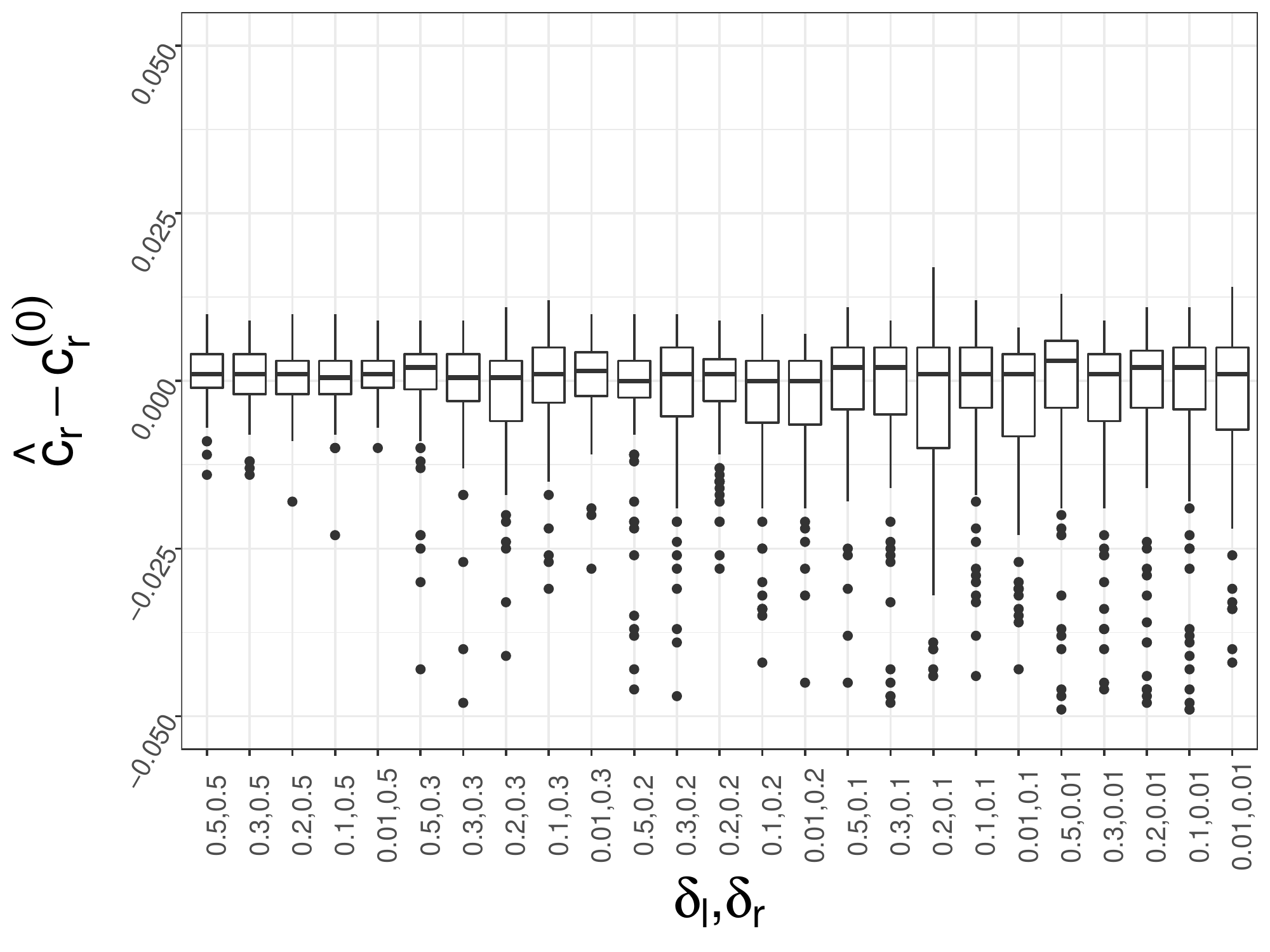}
    \subcaption{}
  \end{minipage}
  \caption{The estimation of $\hat{c}_r$ under the unimodal model. (a) with respect to $m$; (b) with respect to the width of the middle flat region; (c) with respect to the gap sizes.}\label{fig:unimodal_convergence}
\end{figure}




\begin{figure}
  \begin{minipage}{0.49\textwidth}
    \centering
    \includegraphics[width=\textwidth]{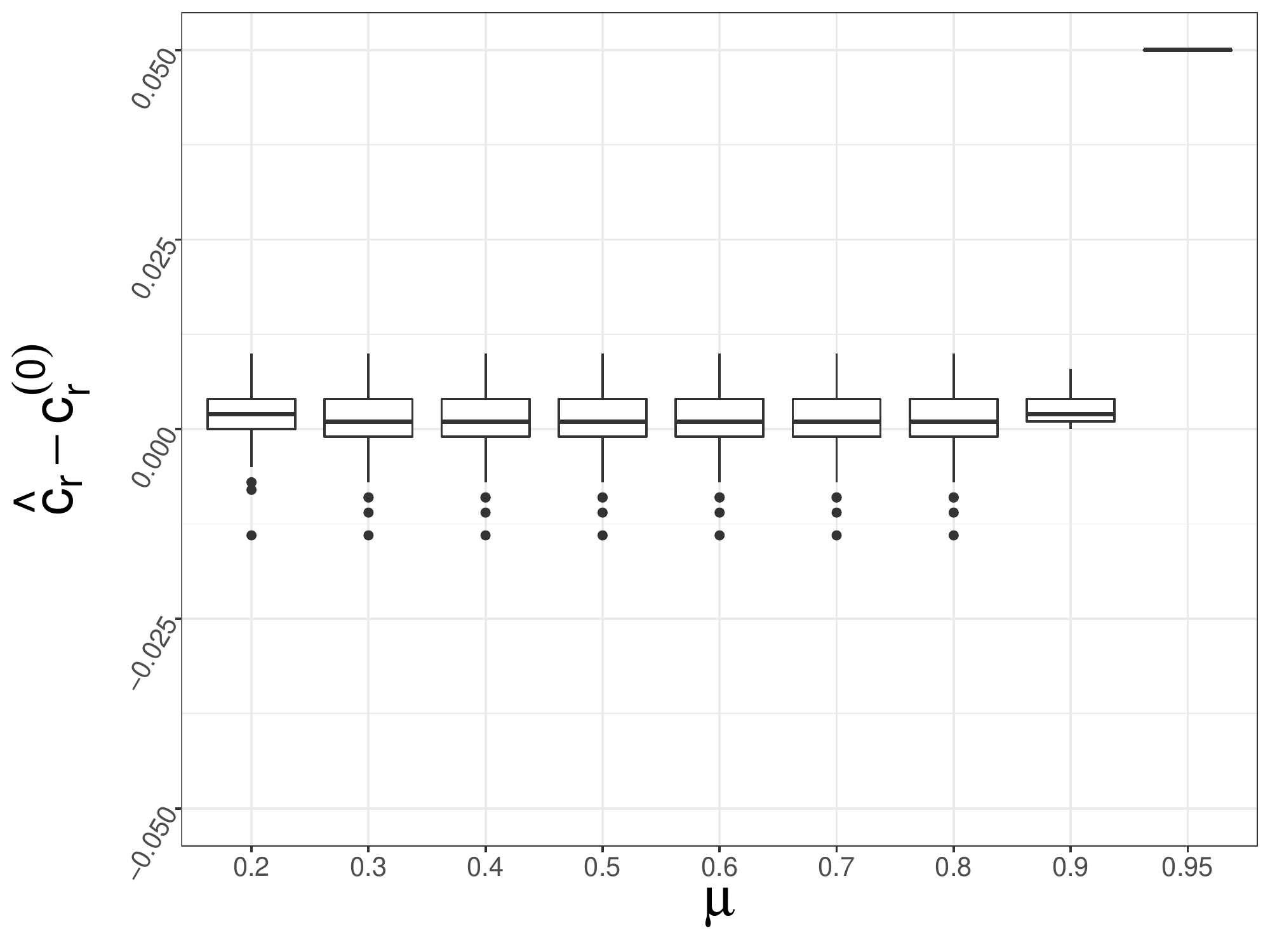}
    \subcaption{}
  \end{minipage}
  \begin{minipage}{0.49\textwidth}
    \centering
  \includegraphics[width=\textwidth]{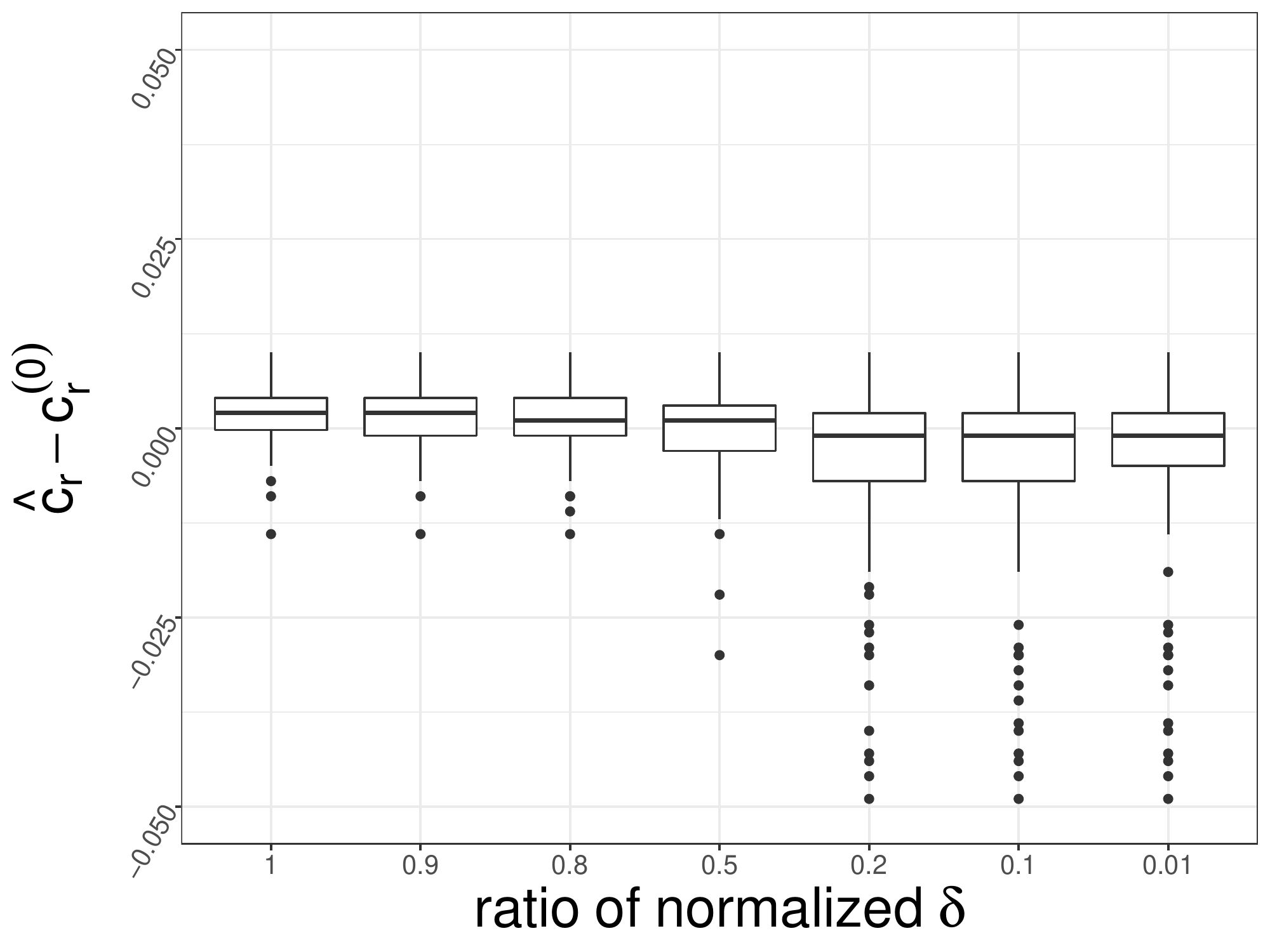}      
    \subcaption{}
  \end{minipage}
      \caption{The estimation of $\hat{c}_r$ under the unimodal model. (a) with respect to the choice of the middle point $\mu$; (b) with respect to the choice of the input $d_l$ and $d_r$. Here $d_l = \kappa \times \tilde{\delta}_l$ and $d_r = \kappa \times \tilde{\delta}_r$, where $\kappa$ is a ratio of the normalized $\delta$'s.}\label{fig:unimodal_sensitivity}
\end{figure}



\begin{figure}[ht] 
  \centering
  \includegraphics[width=0.6\textwidth]{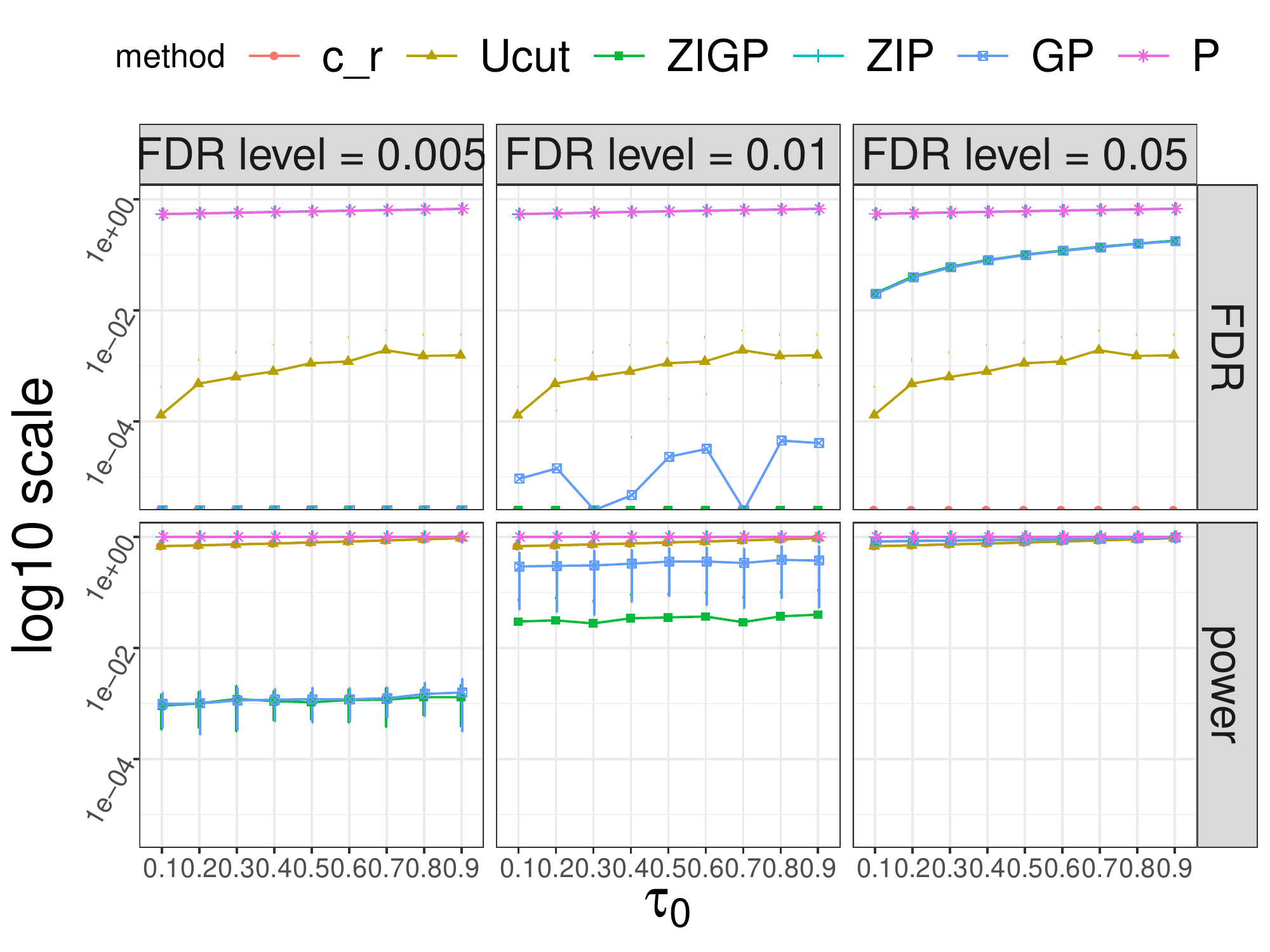}
  \caption{FDR and power of Algorithm \ref{algo:gU_grid_search} and other competing methods on the misspecified  model.}\label{fig:unimodal_comp_rlt}
\end{figure}
